\documentclass[runningheads,a4paper]{llncs}

\usepackage{graphicx}
\usepackage{url}
\newcommand{\keywords}[1]{\par\addvspace\baselineskip
	\noindent\keywordname\enspace\ignorespaces#1}

\usepackage{footmisc}
\usepackage{microtype}
\usepackage{booktabs} 
\usepackage{xspace}
\usepackage{color}
\usepackage{enumitem}
\usepackage{times}
\usepackage{appendix}
\usepackage{algorithmicx}
\usepackage{amsmath,amssymb,latexsym} 
\usepackage{algcompatible}
\usepackage{algpseudocode, algorithm}
\algrenewcommand\alglinenumber[1]{#1:}
\usepackage{listings}
\usepackage{verbatim}
\usepackage{multicol}
\usepackage{float}
\usepackage{pifont}
\usepackage{lineno}
\usepackage{array}
\usepackage{soul}
\usepackage{mathtools}
\usepackage{hyphenat}
\usepackage{setspace}
\usepackage{tcolorbox}
\usepackage{cite}
\usepackage{wrapfig}
\usepackage{lipsum}

\usepackage{layout}
\usepackage[english]{babel}
\usepackage[utf8]{inputenc}
\usepackage{fancyhdr}
\usepackage{tabularx}
\usepackage{ragged2e}
\usepackage{hyperref}
\usepackage{cleveref}

\usepackage{fancyvrb}

\usepackage{epstopdf}
\usepackage{epsfig}
\usepackage{etoolbox}
\usepackage{breqn}
\usepackage{relsize}

\usepackage{pgfplots}
\usepackage{subfiles}

\usepackage{authblk}
\usepackage{pdfpages}
\pgfplotsset{compat=1.7}
\usepackage{calligra}
\usepackage{calrsfs}

\usepackage{subfig}
\usepackage{todonotes}

\usepackage{xcolor}
\usepackage[justification=centering]{caption}

\usepackage{import}

\newcommand{\blank}[1]{\hspace*{#1}}

\algnewcommand\algorithmicswitch{\textbf{switch}}
\algnewcommand\algorithmiccase{\textbf{case}}
\algnewcommand\algorithmicassert{\texttt{assert}}
\algnewcommand\Assert[1]{\State \algorithmicassert(#1)}%
\algdef{SE}[SWITCH]{Switch}{EndSwitch}[1]{\algorithmicswitch\ #1\ \algorithmicdo}{\algorithmicend\ \algorithmicswitch}%
\algdef{SE}[CASE]{Case}{EndCase}[1]{\algorithmiccase\ #1}{\algorithmicend\ \algorithmiccase}%
\algtext*{EndSwitch}%
\algtext*{EndCase}%

\newcommand{\punt}[1]{}
\newcommand{\cmnt}[1]{}

\newcounter{history}

\newtheorem{observation}{Observation}

\newcommand{\secref}[1]{Section~\ref{sec:#1}}
\newcommand{\figref}[1]{Fig~\ref{fig:#1}}
\newcommand{\tabref}[1]{Table~\ref{tab:#1}}

\newcommand{\thmref}[1]{Theorem~\ref{thm:#1}}
\newcommand{\lemref}[1]{Lemma~\ref{lem:#1}}

\newcommand{\corref}[1]{Corollary~\ref{cor:#1}}

\newcommand{\eqnref}[1]{Eqn(\ref{eq:#1})}

\newcommand{\obsref}[1]{Observation~\ref{obs:#1}}
\newcommand{\asmref}[1]{Assumption~\ref{asm:#1}}

\newcommand{\linref}[1]{Line~\ref{lin:#1}}
\newcommand{\algoref}[1]{{Algo~\ref{algo:#1}}}
\newcommand{\subsecref}[1]{SubSection{\ref{subsec:#1}}}

\newcommand{\apnref}[1]{Appendix~\ref{apn:#1}}

\newcommand{\Lineref}[1]{Line~\ref{lin:#1}}

\newcommand{\ignore}[1]{}

\newcommand{\tobj} {t-object\xspace}
\newcommand{\txns}[1] {txns(#1)}
\newcommand {\comm}[1] {committed(#1)}
\newcommand {\aborted}[1] {aborted(#1)}

\newcommand {\live}[1] {live(#1)}
\newcommand {\term}[1] {term(#1)}

\newcommand {\confc}[1] {Conf}

\newcommand{\tseq} {t-sequential\xspace}

\newcommand{\lupdt}[2] {#2.lastUpdt(#1)}

\newcommand{\mr} {MR}
\newcommand{\tr} {TR}

\newcommand{\legal} {legal\xspace}
\newcommand{\legality} {legality\xspace}

\newcommand{\op} {operation\xspace}
\newcommand{\mth} {method\xspace}

\newcommand{\termop} {terminal operation\xspace}

\newcommand{\gen}[1] {gen(#1)}

\newcommand{\evts}[1] {evts(#1)}
\newcommand{\met}[1] {methods(#1)}

\newcommand{\tbeg} {\emph{STM\_begin}\xspace}
\newcommand{\tread} {\emph{STM\_read}\xspace}
\newcommand{\twrite} {\emph{STM\_write}\xspace}
\newcommand{\tins} {\emph{STM\_insert}\xspace}
\newcommand{\tdel} {\emph{STM\_delete}\xspace}
\newcommand{\tlook} {\emph{STM\_lookup}\xspace}
\newcommand{\tryc} {\emph{STM\_tryC}\xspace}
\newcommand{\trya} {\emph{STM\_tryA}\xspace}

\newcommand{\up} {\emph{up}}

\newcommand{\opq} {opaque\xspace}
\newcommand{\opty} {opacity\xspace}

\newcommand{\lopty} {LO\xspace}
\newcommand{\lopq} {locally-opaque\xspace}

\newcommand{\tab} {hash-table\xspace}

\newcommand{\lsl} {rblazy\text{-}list\xspace}

\newcommand{\rvm} {\emph{rvm}\xspace}

\newcommand{\fkmth}[3] {#3.firstKeyMth(#1, #2)}
\newcommand{\pkmth}[3] {#3.prevKeyMth(#1, #2)}

\newcommand{\udset}[1] {updtSet(#1)}

\newcommand{\txsetst}[1] {L\_txlog.setStatus($L\_txstatus \downarrow$, $ OK \downarrow$)}

\newcommand{\rn} {\textcolor{red}{RL}\xspace}
\newcommand{\bn} {\textcolor{blue}{BL}\xspace}

\newcommand{\rc} {\textcolor{red}{currs[0]}}
\newcommand{\bc} {\textcolor{blue}{currs[1]}}
\newcommand{\bp} {\textcolor{blue}{preds[0]}}
\newcommand{\rp} {\textcolor{red}{preds[1]}}

\newcommand{\shist}[2] {#2.subhist(#1)}
\newcommand{\subhist} {subhist\xspace}

\newcommand{\hmvotm} {\textit{HT-MVOSTM}\xspace}

\newcommand{\llog} {txLog\xspace}

\newcommand{\llgopn}[1] {$L\_rec.getOpn(L\_obj\_id \downarrow$, $ L\_key \downarrow$)}

\newcommand{\llgval}[1] {$L\_rec.getVal(L\_obj\_id \downarrow$, $ L\_key \downarrow$)}

\newcommand{\opg}[2] {OPG(#1, #2)}
\newcommand{\copg}[2] {CG(#1, #2)}
\newcommand{\lopg}[2] {OG(#1, #2)}

\newcommand{\mv} {mv}
\newcommand{\rvf} {rvf}
\newcommand{\rt} {rt}
\newcommand{\rvmt} {rv\_method\xspace}
\newcommand{\upmt} {upd\_method\xspace}

\newcommand{\valid} {valid}

\newcommand{\lsls}[1] {list\_lookup($L\_obj\_id \downarrow,  L\_key \downarrow, pred \uparrow, curr \uparrow$)}

\newcommand{\lslins}[1] {list\_Ins($pred \downarrow$, $curr \downarrow$, $node \uparrow$)}

\algrenewcommand{\algorithmiccomment}[1]{/* #1 */}

\newcommand{\stfdm} {starvation\text{-}freedom\xspace}
\newcommand{\stf} {starvation\text{-}free\xspace}

\newtheorem{assumption}{Assumption}

\newcommand{\commit}{\mathcal{C}}
\newcommand{\abort}{\mathcal{A}}

\newcommand {\inc} {incarnation\xspace}
\newcommand {\incs}[2] {#2.incarSet(T_#1)\xspace}
\newcommand {\itsen} {itsEnabled\xspace}
\newcommand {\itsenb}[2] {#2.itsEnabled(T_#1)\xspace}
\newcommand {\cdsen} {cdsEnabled\xspace}
\newcommand {\cdsenb}[2] {#2.cdsEnabled(T_#1)\xspace}
\newcommand {\finen} {finEnabled\xspace}
\newcommand {\finenb}[2] {#2.finEnabled(T_#1)\xspace}
\newcommand {\enbd} {finEnabled\xspace}

\newcommand{\tcts}[1] {cts_#1\xspace}
\newcommand{\tits}[1] {its_#1\xspace}
\newcommand{\twts}[1] {wts_#1\xspace}
\newcommand{\htlock}[2] {#2.lock_#1\xspace}
\newcommand{\htval}[2] {#2.vrt_#1\xspace}
\newcommand{\htstat}[2] {#2.state_#1\xspace}
\newcommand{\htits}[2] {#2.its_#1\xspace}
\newcommand{\htcts}[2] {#2.cts_#1\xspace}
\newcommand{\htwts}[2] {#2.wts_#1\xspace}
\newcommand{\htltl}[2] {#2.tltl_#1\xspace}
\newcommand{\htutl}[2] {#2.tutl_#1\xspace}

\newcommand{\val} {valid\xspace}

\newcommand{\gtcnt} {gcounter\xspace}

\newcommand{\tcntr} {tCntr\xspace}

\newcommand{\incv} {incrVal\xspace}

\newcommand{\glock} {lock\xspace}
\newcommand{\gval} {valid\xspace}

\newcommand{\tltl} {tltl\xspace}
\newcommand{\tutl} {tutl\xspace}
\newcommand{\ttltl}[1] {tltl_#1\xspace}
\newcommand{\ttutl}[1] {tutl_#1\xspace}
\newcommand{\tlock}[1] {lock_#1\xspace}
\newcommand{\tval}[1] {vrt_#1\xspace}
\newcommand{\tstat}[1] {state_#1\xspace}
\newcommand{\syst} {sys\text{-}time\xspace}
\newcommand {\incset} {incarSet\xspace}
\newcommand {\incn} {incNum\xspace}
\newcommand {\inum}[1] {T_#1.incNum\xspace}
\newcommand {\ninc} {nextInc\xspace}
\newcommand {\aptr} {application-transaction\xspace}
\newcommand {\nexti}[1] {T_#1.nextInc\xspace}
\newcommand {\inct}[2] {#2.incarCt(T_#1)\xspace}
\newcommand {\incct} {incarCt\xspace}
\newcommand {\affset} {affectSet\xspace}
\newcommand {\haffset}[2] {#2.affectSet(T_#1)\xspace}
\newcommand {\hmaxwts}[2] {#2.maxWTS(T_#1)\xspace}
\newcommand {\maxwts} {maxWTS\xspace}
\newcommand {\haffwts}[2] {#2.affwts(T_#1)\xspace}
\newcommand {\affwts} {affWTS\xspace}
\newcommand {\cdset} {cds\xspace}

\newcommand {\hcds}[2] {#2.cds(T_#1)\xspace}
\newcommand{\shset}[1] {#1.subhistSet\xspace}
\newcommand{\stsble} {strict\text{-}serializable\xspace}

\newcommand{\termed}[1] {#1.terminated\xspace}
\newcommand{\ksftm} {\textit{SF-KOSTM}\xspace}
\newcommand{\cts} {cts\xspace}
\newcommand{\its} {its\xspace}
\newcommand{\wts} {wts\xspace}
\newcommand{\hsyst}[1] {#1.sys\text{-}time\xspace}
\newcommand{\begt} {stm\text{-}begin\xspace}
\newcommand{\ct} {comTime\xspace}
\newcommand{\vt} {\texttt{vrt}\xspace}
\newcommand{\csref}[1]{Case~\ref{case:#1}}
\newcommand{\lrl} {largeRL\xspace}
\newcommand{\abl} {abortRL\xspace}
\newcommand{\srl} {smallRL\xspace}
\newcommand {\cis} {cis\xspace}
\newcommand {\hcis}[2] {#2.cis(T_#1)\xspace}
\newcommand {\depits} {depits\xspace}
\newcommand {\hdep}[2] {#2.depits(T_#1)\xspace}
\newcommand {\pawts} {partAffwts\xspace}
\newcommand {\hpawts}[2] {#2.partAffwts(T_#1)\xspace}

\begin{document}
	
	\mainmatter  
	\vspace{-3mm}
\title{Achieving Starvation-Freedom with Greater
	Concurrency in Multi-Version Object-based Transactional Memory Systems}

	\titlerunning{Starvation-Freedom in Multi-Version Object-based Transactional Memory Systems}
	
	%
	%
	
\author{Chirag Juyal\inst{1}\and Sandeep Kulkarni\inst{2}\and Sweta Kumari\inst{1}\and Sathya Peri\inst{1}\and Archit Somani\inst{1}\footnote{Author sequence follows the lexical order of last names. All the authors can be contacted at the addresses given above. Archit Somani's phone number: +91 - 7095044601.}\vspace{-3mm}}
\authorrunning{C.Juyal \and S.Kulkarni \and S.Kumari \and S.Peri \and A.Somani}

\institute{Department of Computer Science \& Engineering, IIT Hyderabad, Kandi, Telangana, India \\
	\texttt{(cs17mtech11014, cs15resch01004, sathya\_p, cs15resch01001)@iith.ac.in} \and Department of Computer Science, Michigan State University, MI, USA \\
	\texttt{sandeep@cse.msu.edu}}

	%
	%
	
	\maketitle
	
	\vspace{-6mm}
	\begin{abstract}
 To utilize the multi-core processors properly concurrent programming is needed. The main challenge is to design a correct and efficient concurrent program. Software Transactional Memory Systems (STMs) provide ease of multithreading to the programmer without worrying about concurrency issues as deadlock, livelock, priority inversion, etc. Most of the STMs work on read-write operations known as RWSTMs. Some STMs work at higher-level operations and ensure greater concurrency than RWSTMs. Such STMs are known as Single-Version Object-based STMs (SVOSTMs). The transactions of SVOSTMs can return \emph{commit} or \emph{abort}. Aborted SVOSTMs transactions retry. But in the current setting of SVOSTMs, transactions may \emph{starve}. So, we propose a \emph{Starvation-Freedom} in \emph{SVOSTM} as \emph{SF-SVOSTM} that satisfies the correctness criteria \emph{conflict-opacity}.

Databases and STMs say that maintaining multiple versions corresponding to each shared data-item (or key) reduces the number of aborts and improves the throughput. So, to achieve greater concurrency further, we propose \emph{Starvation-Freedom} in \emph{Multi-Version OSTM} as \emph{SF-MVOSTM} algorithm. The number of versions maintains by SF-MVOSTM either be unbounded with garbage collection as SF-MVOSTM-GC or bounded with latest $K$-versions as SF-KOSTM. SF-MVOSTM satisfies the correctness criteria as \emph{local opacity} and shows the performance benefits as compared with state-of-the-art STMs.

\vspace{-1mm}
	\end{abstract}
\vspace{-10mm}
\section{Introduction}
\label{sec:intro}

\noindent
In the era of multi-core processors, we can exploit the cores by concurrent programming. But developing an efficient concurrent program while ensuring the correctness is difficult. Software Transactional Memory Systems (STMs) are a convenient programming interface to access the shared memory concurrently while removing the concurrency responsibilities from the programmer. STMs ensure that consistency issues such as deadlock, livelock, priority inversion, etc will not occur. It provides a high-level abstraction to the programmer with the popular correctness criteria opacity \cite{GuerKap:Opacity:PPoPP:2008}, local opacity \cite{KuznetsovPeri:Non-interference:TCS:2017} which consider all the transactions (a piece of code) including aborted one as well in the equivalent serial history. This property makes it different from correctness criteria of database serializability,  strict-serializability \cite{Papad:1979:JACM} and ensures even aborted transactions read correct value in STMs which prevent from divide-by-zero, infinite loop, crashes, etc. Another advantage of STMs is composability which ensures the effect of multiple operations of the transaction will be atomic. This paper considers the optimistic execution of STMs in which transactions are writing into its local log until the successful validation.


A traditional STM system invokes following methods:(1) \tbeg{()}: begins a transaction $T_i$ with unique timestamp $i$. (2) $\tread_i${($k$)} (or $r_i(k)$): $T_i$ reads the value of key $k$ from shared memory. (3) $\twrite_i{(k,v)}$ (or $w_i(k,v)$): $T_i$ writes the value of $k$ as $v$ locally. (4) \emph{STM\_tryC$_i$()}: on successful validation, the effect of $T_i$ will be visible to the shared memory and $T_i$ returns commit otherwise (5) \emph{STM\_tryA$_i$()}: $T_i$ returns abort. These STMs are known as \emph{read-write STMs (RWSTMs)} because it is working at lower-level operations such as read and write. 

Herlihy et al.\cite{HerlihyKosk:Boosting:PPoPP:2008}, Hassan et al. \cite{ Hassan+:OptBoost:PPoPP:2014}, and Peri et al. \cite{Peri+:OSTM:Netys:2018} have shown that working at higher-level operations such as insert, delete and lookup on the linked-list and hash table gives better concurrency than RWSTMs. STMs which work on higher-level operations are known as \emph{Single-Version Object-based STMs (SVOSTMs)} \cite{Peri+:OSTM:Netys:2018}. It exports the following methods: (1) \emph{STM\_begin()}: begins a transaction $T_i$ with unique timestamp $i$ same as RWSTMs. (2) \emph{STM\_lookup$_i(k)$} (or $l_i(k)$): $T_i$ lookups key $k$ from shared memory and returns the value. (3) \emph{STM\_insert$_i(k,v)$} (or $i_i(k,v)$): $T_i$ inserts a key $k$ with value $v$ into its local memory. (4) \emph{STM\_delete$_i(k)$}(or $d_i(k)$): $T_i$ deletes key $k$. (5) \emph{STM\_tryC$_i$()}: the actual effect of \emph{STM\_insert()} and \emph{STM\_delete()} will be visible to the shared memory after successful validation and $T_i$ returns commit otherwise (6) \emph{STM\_tryA$_i$()}: $T_i$ returns abort. 

\setlength{\intextsep}{0pt}

\begin{figure*}
	\centerline{\scalebox{0.41}{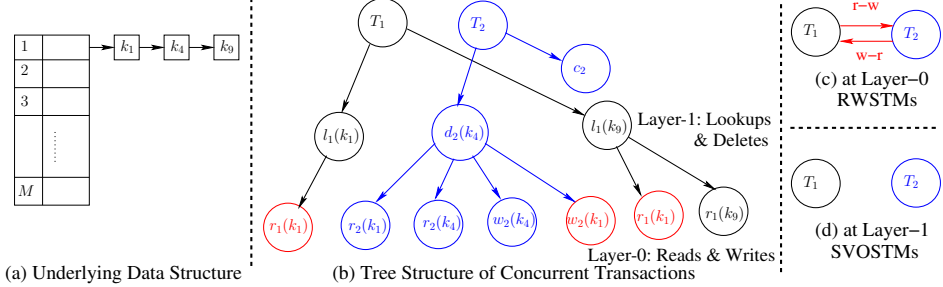}}
\vspace{-.2cm}	\caption{Advantage of SVOSTMs over RWTSMs}
	\label{fig:adv}
\end{figure*} 

\noindent
\textbf{Motivation to work on SVOSTMs}: \figref{adv} represents the advantage of SVOSTMs over RWSTMs while achieving greater concurrency and reducing the number of aborts. \figref{adv}.(a) depicts the underlying data structure as a hash table (or $ht$) with $M$ buckets and bucket $1$ stores three keys $k_1, k_4$ and $k_9$ in the form of the list. Thus, to access $k_4$, a thread has to access $k_1$ before it. \figref{adv}.(b) shows the tree structure of concurrent execution of two transactions $T_1$ and $T_2$ with RWSTMs at layer-0 and SVOSTMs at layer-1 respectively. Consider the execution at layer-0, $T_1$ and $T_2$ are in conflict because write operation of $T_2$ on key $k_1$ as $w_2(k_1)$ is occurring between two read operations of $T_1$ on $k_1$ as $r_1(k_1)$. Two transactions are in conflict if both are accessing the same key $k$ and at least one transaction performs write operation on $k$. So, this concurrent execution cannot be atomic as shown in \figref{adv}.(c). To make it atomic either $T_1$ or $T_2$ has to return abort. Whereas execution at layer-1 shows the higher-level operations $l_1(k_1)$, $d_2(k_4)$ and $l_1(k_9)$ on different keys $k_1, k_4$ and $k_9$ respectively. All the higher-level operations are isolated to each other so tree can be pruned \cite[Chap 6]{WeiVoss:TIS:2002:Morg} from layer-0 to layer-1 and both the transactions return commit with equivalent serial schedule $T_1T_2$ or $T_2T_1$ as shown in \figref{adv}.(d). Hence, some conflicts of RWSTMs does not matter at SVOSTMs which reduce the number of aborts and improve the concurrency using SVOSTMs.

\noindent
\textbf{Starvation-Freedom}: For long-running transactions along with high conflicts, starvation can occur in SVOSTMs. So, SVOSTMs should ensure the progress guarantee as \emph{starvation-freedom} \cite[chap 2]{HerlihyShavit:AMP:Book:2012}. SVOSTMs is said to be \emph{\stf}, if a thread invoking a transaction $T_i$ gets the opportunity to retry $T_i$ on every abort (due to the presence of a fair underlying scheduler with bounded termination) and $T_i$ is not \emph{parasitic}, i.e., if scheduler will give a fair chance to $T_i$ to commit then $T_i$ will eventually return commit. If a transaction gets a chance to commit, still it is not committing because of the infinite loop or some other error such transactions are known as Parasitic transactions \cite{Bushkov+:Live-TM:PODC:2012}. 

We explored another well known non-blocking progress guarantee \emph{wait-freedom} for STM  which ensures every transaction commits regardless of the nature of concurrent transactions and the underlying scheduler \cite{HerlihyShavit:Progress:Opodis:2011}. However, Guerraoui and  Kapalka  \cite{Bushkov+:Live-TM:PODC:2012, tm-book} showed that achieving wait-freedom is impossible in dynamic STMs in which data-items (or keys) of transactions are not known in advance. So in this paper, we explore the weaker progress condition of \emph{\stfdm} for SVOSTM while assuming that the keys of the transactions are not known in advance.  

\noindent
\textbf{Related work on Starvation-free STMs:} Some researchers Gramoli et al. \cite{Gramoli+:TM2C:Eurosys:2012}, Waliullah and Stenstrom \cite{WaliSten:Starve:CC:2009}, Spear et al. \cite{Spear+:CCM:PPoPP:2009}, Chaudhary et al. \cite{Chaudhary+:KSFTM:Corr:2017} have explored starvation-freedom in RWSTMs. Most of them assigned priority to the transactions. On conflict, higher priority transaction returns commit whereas lower priority transaction returns abort. On every abort, a transaction retries a sufficient number of times, will eventually get the highest priority and returns commit.
We inspired with these researchers and propose a novel \emph{Starvation-Free SVOSTM (SF-SVOSTM)} which assigns the priority to the transaction on conflict. In SF-SVOSTM whenever a conflicting transaction $T_i$ aborts, it retries with $T_j$ which has higher priority than $T_i$. To ensure the \stfdm, this procedure will repeat until $T_i$ gets the highest priority and eventually returns commit.%


\vspace{.2mm}
\begin{figure}
	\centerline{\scalebox{0.36}{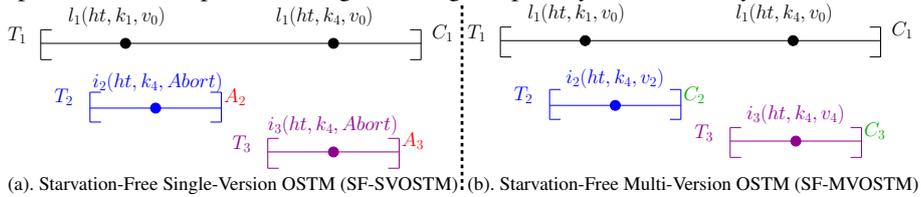}}
	\vspace{-.2cm}	\caption{Benefits of Starvation-Free Multi-Version OSTM over SF-SVOSTM}
	\label{fig:advmvostm}
\end{figure}

\noindent
\textbf{Motivation to Propose Starvation-Freedom in Multi-Version OSTM}: In SF-SVOSTM, if the highest priority transaction becomes slow (for some reason) then it may cause several other transactions to abort and bring down the progress of the system. \figref{advmvostm}.(a). demonstrates this in which the highest priority transaction $T_1$ became slow so, it is forcing the conflicting transactions $T_2$ and $T_3$ to abort again and again until $T_1$ commits. Database, RWSTMs \cite{Kumar+:MVTO:ICDCN:2014,LuScott:GMV:DISC:2013,Fern+:LSMVSTM:PPoPP:2011,Perel+:2011:SMV:DISC} and SVOSTMs \cite{Juyal+:SSS:2018} say that maintaining multiple versions corresponding to each key reduces the number of aborts and improves throughput.

So, this paper proposes the first \emph{Starvation-Free Multi-Version OSTM (SF-MVOSTM)} which maintains multiple versions corresponding to each key. \figref{advmvostm}.(b). shows the benefits of using SF-MVOSTM in which $T_1$ lookups from the older version with value $v_0$ created by transaction $T_0$ (assuming as initial transaction) for key $k_1$ and $k_4$. Concurrently, $T_2$ and $T_3$ create the new versions for key $k_4$. So, all the three transactions commit with equivalent serial schedule $T_1T_2T_3$. So, SF-MVOSTM improves the concurrency than SF-SVOSTM while reducing the number of aborts and ensures the \stfdm.



\noindent
\textbf{Contributions of the paper:}  We propose two Starvation-Free OSTMs as follows:
\vspace{-.1cm}
\begin{itemize}
\item Initially, we propose Starvation-Freedom for Single-Version OSTM as SF-SVOSTM which satisfies correctness criteria as \emph{conflict-opacity (or co-opacity)}\cite{Peri+:OSTM:Netys:2018}.
\item To achieve the greater concurrency further, we propose Starvation-Freedom for Multi-Version OSTM as SF-MVOSTM in \secref{pm} which maintains multiple versions corresponding to each key and satisfies the correctness as \emph{local opacity} \cite{KuznetsovPeri:Non-interference:TCS:2017}.
\item We propose SF-SVOSTM and SF-MVOSTM for \emph{hash table} and \emph{linked-list} data structure describe in \subsecref{dsde} but its generic for other data structures as well.
\item SF-MVOSTM works for unbounded versions with \emph{Garbage Collection (GC)} as SF-MVOSTM-GC which deletes the unwanted versions from version list of keys and for bounded/finite versions as SF-KOSTM which stores finite say latest $K$ number of versions corresponding to each key $k$. So, whenever any thread creates $(K+1)^{th}$ version of key, it replaces the oldest version of it. The most challenging task is achieving \stfdm in bounded version OSTM because say, a highest priority transaction rely on the oldest version that has been replaced. So, in this case highest priority transaction has to return abort and hence make it harder to achieve starvation-freedom unlike the approach follow in SF-SVOSTM. Thus, in this paper we propose a novel approach SF-KOSTM which bridges the gap by developing \stf OSTM while maintaining bounded number of versions. 
\item \secref{exp} shows that SF-KOSTM is best among all propose Starvation-Free OSTMs (SF-SVOSTM, SF-MVOSTM, and SF-MVOSTM-GC) for both hash table and linked-list data structure. Proposed hash table based SF-KOSTM (HT-SF-KOSTM) performs 3.9x, 32.18x, 22.67x, 10.8x and 17.1x average speedup on \emph{max-time} for a transaction to commit than state-of-the-art STMs HT-KOSTM \cite{Juyal+:SSS:2018}, HT-SVOSTM \cite{Peri+:OSTM:Netys:2018}, 
ESTM \cite{Felber+:ElasticTrans:2017:jpdc}, RWSTM \cite [Chap. 4]{ WeiVoss:TIS:2002:Morg}, and HT-MVTO \cite{Kumar+:MVTO:ICDCN:2014} respectively. Proposed list based SF-KOSTM (list-SF-KOSTM)  performs 2.4x, 10.6x, 7.37x, 36.7x, 9.05x, 14.47x, and 1.43x average speedup on \emph{max-time} for a transaction to commit than state-of-the-art STMs list-KOSTM \cite{Juyal+:SSS:2018}, list-SVOSTM \cite{Peri+:OSTM:Netys:2018}, Trans-list \cite{ZhangDech:LockFreeTW:SPAA:2016}, Boosting-list \cite{HerlihyKosk:Boosting:PPoPP:2008}, NOrec-list \cite{Dalessandro+:NoRec:PPoPP:2010}, list-MVTO \cite{Kumar+:MVTO:ICDCN:2014}, and list-KSFTM \cite{Chaudhary+:KSFTM:Corr:2017} respectively.


\end{itemize}

\vspace{-.6cm}
\section{System Model and Preliminaries}
\label{sec:sm}
This section follows the notion and definition described in \cite{Juyal+:SSS:2018,tm-book}, 
we assume a system of $n$ processes/threads, $th_1,\ldots,th_n$  that run in a completely asynchronous manner and communicate through a set of \emph{keys} $\mathcal{K}$ (or \emph{{transaction-object}s}). We also assume that none of the threads crash or fail abruptly. In this paper, a thread executes higher-level methods on $\mathcal{K}$ via atomic \emph{transactions} $T_1,\ldots,T_n$ and receives the corresponding response. 

\noindent
\textbf{Events and Methods:}
Threads execute the transactions with higher-level methods (or operations) which internally invoke multiple read-write (or lower-level) operations known as \emph{events} (or $evts$). 
Transaction $T_i$ of the system at read-write level invokes \emph{STM\_begin()}, \emph{STM\_read$_i$(k)},  \emph{STM\_write$_i$(k,v)}, \emph{STM\_tryC$_i$()} and \emph{STM\_tryA$_i$()} as defined in \secref{intro}. We denote a method $m_{ij}$ as the $j^{th}$ method of $T_i$. Method \emph{invocation} (or $inv$) and \emph{response} (or $rsp$) on higher-level methods are also considered as an event.

A thread executes higher-level operations on $\mathcal{K}$ via transaction $T_i$ are known as \emph{methods} (or $mths$). $T_i$ at object level (or higher-level) invokes \emph{STM\_begin()}, \emph{STM\_lookup$_i$($k$) (or $l_i(k)$)},  \emph{STM\_insert$_i$($k,v$) (or $i_i(k,v)$)}, \emph{STM\_delete$_i$($k$) (or $d_i(k)$)}, \emph{STM\_tryC$_i$()}, and  \emph{STM\_tryA$_i$()} methods described in \secref{intro}. 
Here, \emph{STM\_lookup()}, and \emph{STM\_delete()} return the value from underlying data structure so, we called these methods as \emph{return value methods (or $\rvmt{s}$)}. Whereas, \emph{STM\_insert()}, and \emph{STM\_delete()} are updating the underlying data structure after successful \emph{STM\_tryC()} so, we called these methods as \emph{update} \mth{s} (or $\upmt{s}$).


\noindent
\textbf{Transactions:} We follow multi-level transactions\cite{WeiVoss:TIS:2002:Morg} model which consists of two layers. Layer 0 (or lower-level) composed of read-write operations whereas layer 1 (or higher-level) comprises of object-level methods which internally calls multiple read-write events. Formally, we define a transaction $T_i$ at higher-level as the tuple $\langle \evts{T_i}, <_{T_i}\rangle$, here $<_{T_i}$ represents the total order among all the events of $T_i$. 
Transaction $T_i$ cannot invoke any more operation after returning commit ($\commit$) or abort ($\abort$). Any operation that returns $\commit$ or $\abort$ are known as \emph{\termop{s}} represented as $Term(T_i)$. The transaction which neither committed nor aborted is known as \emph{live transactions} (or $trans.live$).

\noindent
\textbf{Histories:} 
A history $H$ consists of multiple transactions, a transaction calls multiple methods and each method internally invokes multiple read-write events. So, a history is a collection of events belonging to the different transactions is represented as $\evts{H}$. Formally, we define a history $H$ as the tuple $\langle \evts{H}, <_H\rangle$, here $<_H$ represents the total order among all the events of $H$. If all the method invocation of $H$  match with the corresponding response then such history is known as \emph{complete history} denoted as $\overline{H}$. Suppose total transactions in $H$ is $H.trans$, in which number of committed and aborted transactions are $H.committed$ and $H.aborted$ then the \emph{incomplete history} or \emph{live history} is defined as: $H.incomp$ = $H.live$ = \emph{($H.trans$ - $H.committed$ - $H.aborted$)}. This paper considers only \emph{well-form history} which ensures (1) the response of the previous method has received then only the transaction $T_i$ can invoke another method. (2) transaction can not invoke any other method after receiving the response as $\commit$ or $\abort$. 
  
 \noindent
 \textbf{Sequential Histories:} Following \cite{KuznetsovPeri:Non-interference:TCS:2017, KuznetsovRavi:ConcurrencyTM:OPODIS:2011},  if all the methods of the history are complete and isolated from each other, i.e. all the methods of history are atomic and following the total order in $H$. So, with this assumption, the only relevant methods of a transaction $T_i$ are $l_i(k,v)$, 
 $l_i(k,\abort)$, $i_i(k, v)$, $d_i(k, v)$, $\tryc_i(\commit)$ (or $c_i$ for short), 
 $\tryc_i(\abort)$, and $\trya_i(\abort)$ (or $a_k$ for short). We form the complete history $\overline{H}$, by immediately inserting the abort to all the live transactions of $H$. But, if the first operation of \emph{STM\_tryC$_i$()} has returned successfully then $T_i$ returns commit.
 
 \noindent
 \textbf{Real-Time Order and Serial History:} Two complete method $m_{ij}$ and $m_{xy}$ are said to be in method real-time order (or MR), if the response of $m_{ij}$ happens before the invocation of $m_{xy}$. Formally, if ($rsp(m_{ij}) <_H inv(m_{xy})$)$\implies$($m_{ij} \prec_H^{\mr} m_{xy}$). Following \cite{Papad:1979:JACM}, if transaction $T_i$ terminates (either commits or aborts) before beginning of $T_j$ then $T_i$ and $T_j$  follows transaction real-time order (or TR). Formally, if ($Term(T_i) <_H \tbeg_j()$)$\implies$($T_i \prec_H^{\tr} T_j$).  If all the transactions of a history $H$ follow the real-time order, i.e. transactions are atomic then such history is known as \emph{serial} \cite{Papad:1979:JACM} or \emph{t-sequential} \cite{KuznetsovRavi:ConcurrencyTM:OPODIS:2011} history. 
 Formally, $\langle (H \text{ is serial}) \implies (\forall T_{i} \in \txns{H}: (T_i \in Term{(H)}) \land (\forall T_{i}, T_{j} \in \txns{H}: (T_{i} \prec_H^{\tr} T_{j}) \lor (T_{j} \prec_H^{\tr} T_{i}))\rangle$. In the serial history all the methods within a transaction are also ordered, it is also sequential.

 
 \noindent
 \textbf{Valid and Legal History:}  If \emph{$rv\_method()$}  on key $k$ returns the value from any of the previously committed transactions then such \emph{$rv\_method()$} is known as \emph{valid method}. If all the \emph{$rv\_methods()$} of history $H$ return the value from any of the previously committed transactions then such $H$ is known as \emph{valid history}.
 
 Whereas, If \emph{$rv\_method()$}  on key $k$ returns the value from previous closest committed transaction then such \emph{$rv\_method()$} is known as \emph{legal method}. If all the \emph{$rv\_methods()$} of history $H$ returns the value from the previous closest committed transactions then such $H$ is known as \emph{legal history}.

 \noindent
 \textbf{Sub-history:} Following \cite{Chaudhary+:KSFTM:Corr:2017}, A \textit{sub-history} ($SH$) of a history ($H$) is a collection of events represented as $\evts{SH}$. We denote a $SH$ as the tuple $\langle \evts{SH},$ $<_{SH}\rangle$, here $<_{SH}$ represents the
 total order of among all the events of $SH$. Formally, $SH$ is defined as: (1) $<_{SH} \subseteq <_{H}$; (2) $\evts{SH} \subseteq \evts{H}$; (3) If an event of a transaction $T_k\in\txns{H}$ is in $SH$ then all the events of $T_k$ in $H$ should also be in $SH$. 
 
 Let $P$ is a subset of $H$ which consists of a set transactions $\txns{H}$. Then the sub-history of $H$ is represented as $\shist{P}{H}$ which comprises of $\evts{}$ of $P$. 
 
 \noindent
 \textbf{Conflict-opacity (or Co-opacity):} It is a well known correctness criteria of STMs that is polynomial time verifiable. Co-opacity \cite{Peri+:OSTM:Netys:2018} is a subclass of opacity. A history $H$ is co-opaque, if there exist a serial history $S$ with following properties: (1) $S$ and $\overline{H}$ are equivalent, i.e. $\evts{S} = \evts{\overline{H}}$. (2) $S$ should be legal. (3) $S$ should respect the transactional real-time as $H$ $\prec_S^{\tr} \subseteq \prec_H^{\tr}$. (4) $S$ preserves conflict-order of $H$, i.e. $\prec_S^{Conf} \subseteq \prec_H^{Conf}$.
 
 \noindent
 \textbf{Opacity and Strict Serializability:} A popular correctness criteria for STMs is opacity \cite{GuerKap:Opacity:PPoPP:2008}. A valid history $H$ is opaque, if there exist a serial history $S$ with following properties: (1) $S$ and $\overline{H}$ are equivalent, i.e. $\evts{S} = \evts{\overline{H}}$. (2) $S$ should be legal. (3) $S$ should respect the transactional real-time order as $H$, i.e. $\prec_S^{\tr} \subseteq \prec_H^{\tr}$.
 
 A valid history $H$ is strict serializable \cite{Papad:1979:JACM}, if sub-history  of $H$ as $\subhist(H)$ (consists of all the committed transactions of $H$) is opaque, i.e., $H.subhist(H.committed)$ is \emph{opaque}. So, if a history is \emph{strict serializable} then it will be \opq as well but  the vice-versa is not true. Unlike opacity, \emph{strict serializability} does not consider aborted transactions.
 
 A commonly used correctness criteria in databases is \emph{serializability} \cite{Papad:1979:JACM} which considers committed transactions only. Guerraoui \& Kapalka \cite{GuerKap:Opacity:PPoPP:2008} show that STMs considers the correctness of \emph{aborted} transactions along with \emph{committed} transactions as well. So,  \emph{serializability} is not suitable for STMs. On the other hand, \emph{opacity} \cite{GuerKap:Opacity:PPoPP:2008} is a popular correctness criteria for STMs which considers \emph{aborted} as well as \emph{committed} transactions. Similar to \emph{opacity} but less restrictive, an another popular correctness criteria is \emph{local opacity} defined below.

 \noindent
 \textbf{Local Opacity (or LO):} It is an another popular correctness criteria for STMs which is less restrictive than opacity. Local-opacity \cite{KuznetsovPeri:Non-interference:TCS:2017} is a superclass of opacity. It considers a set of sub-histories for history $H$, denoted as $\shset{H}$ as follows: (1) We construct a $\subhist$ corresponding to each aborted transaction $T_i$ while including all the events $\evts{}$ of previously \emph{committed} transactions along with all successful $\evts{}$ of $T_i$ (i.e., $\evts{}$ which has not returned $\mathcal{A}$ yet) and put \emph{committed} $\mathcal{C}$ immediately after last successful operation of $T_i$; (2) We construct one more sub-history corresponding to last \emph{committed} transaction $T_l$ which considers all the previously \emph{committed} transactions only including  $T_l$.
 
 If all the above defined sub-histories $\shset{H}$ for a history $H$ are opaque then $H$ is said to be \emph{\lopq} \cite{KuznetsovPeri:Non-interference:TCS:2017}. We can observe that a sub-history corresponding to an aborted transaction $T_i$ contains $\evts{}$ from only one aborted transaction which is $T_i$ itself and does not consider any other live/aborted transactions. Similarly, sub-history corresponding to last \emph{committed} transaction $T_l$ does not consider any $\evts{}$ of aborted and live transactions. Thus in \lopty, neither aborted nor live transaction cause any other transaction to abort. Researchers have been proved that \lopty 
 \cite{KuznetsovPeri:Non-interference:TCS:2017} ensures greater concurrency than \opty \cite{GuerKap:Opacity:PPoPP:2008}. If a history is \opq then it will be \lopq as well but  the vice-versa is not true. On the other hand, if a history is \lopq then it will be \stsble as well, but the vice-versa need not be true. Please find  illustrations of local opacity in \secref{gcoflo}. 

\vspace{-.1cm}
\section{The Proposed SF-KOSTM Algorithm}
\label{sec:pm}
\vspace{-.1cm}

In this section, we propose \emph{Starvation-Free K-version OSTM (SF-KOSTM)} algorithm which maintains $K$ number of versions corresponding to each key. The value of $K$ can vary from 1 to $\infty$. When $K$ is equal to 1 then SF-KOSTM boils down to \emph{Starvation-Free Single-Version OSTM (SF-SVOSTM)}. When $K$ is $\infty$ then SF-KOSTM maintains unbounded versions corresponding to each key known as \emph{Starvation-Free Multi-Version OSTM (SF-MVOSTM)} algorithm. To delete the unused version from the version list of SF-MVOSTM, it calls a separate Garbage Collection (GC) method \cite{Kumar+:MVTO:ICDCN:2014} and proposes SF-MVOSTM-GC.  In this paper, we propose SF-SVOSTM and all the variants of SF-KOSTM \emph{(SF-MVOSTM, SF-MVOSTM-GC, SF-KOSTM)} for two data structures \emph{hash table} and \emph{linked-list} but it is generic for other data structures as well.

\subsecref{prelim} describes the definition of \emph{starvation-freedom} followed by our assumption about the scheduler that helps us to achieve \stfdm in SF-KOSTM. \subsecref{dsde} explains the design and data structure of SF-KOSTM. \subsecref{working} shows the working of SF-KOSTM algorithm. 


\vspace{-.2cm}
\subsection{Description of Starvation-Freedom}
\label{subsec:prelim}
\vspace{-.1cm}
\vspace{-.1cm}
\begin{definition}
\label{defn:stf}
\textbf{Starvation-Freedom:} An STM system is said to be \stf if a thread invoking a non-parasitic transaction $T_i$ gets the opportunity to retry $T_i$ on every abort, due to the presence of a fair scheduler, then $T_i$ will eventually commit.
\end{definition}
\vspace{-.1cm}
Herlihy \& Shavit \cite{HerlihyShavit:Progress:Opodis:2011} defined the fair scheduler which ensures that none of the thread will crash or delayed forever. Hence, any thread $Th_i$ acquires the lock on the shared data-items while executing transaction $T_i$ will eventually release the locks. So, a thread will never block other threads to progress. To satisfy the \stfdm for SF-KOSTM, we assumed bounded termination for the fair scheduler.
\vspace{-.1cm}
\begin{assumption}
	\label{asm:bdtm}
	\textbf{Bounded-Termination:} For any transaction $T_i$, invoked by a thread $Th_i$, the fair system scheduler ensures, in the absence of deadlocks, $Th_i$ is given sufficient time on a CPU (and memory, etc) such that $T_i$ terminates ($\commit$ or $\abort$) in bounded time. 
\end{assumption}
\vspace{-.1cm}
In the proposed algorithms, we have considered \emph{TB} as the maximum time-bound of a transaction $T_i$ within this either $T_i$ will return commit or abort in the absence of deadlock. Approach for achieving the \emph{deadlock-freedom} is motivated from the literature in which threads executing transactions acquire the locks in increasing order of the keys and releases the locks in bounded time either by committing or aborting the transaction. We consider an assumption about the transactions of the system as follows.
\vspace{-.1cm}
\begin{assumption}
	\label{asm:self}
	We assume, if other concurrent conflicting transactions do not exist in the system then every transaction will commit. i.e. (a) If a transaction $T_i$ is executing in the system with the absence of other conflicting transactions then $T_i$ will not self-abort. (b) Transactions of the system are non-parasitic as explained in \secref{intro}.
\end{assumption}
\vspace{-.1cm}
\noindent If transactions self-abort or parasitic then ensuring \stfdm is impossible.

\vspace{-.1cm}
\subsection{Design and Data Structure of SF-KOSTM Algorithm}
\label{subsec:dsde}
\vspace{-.1cm}
In this subsection, we show the design and underlying data structure of SF-KOSTM algorithm to maintain the shared data-items (or keys).

To achieve the \emph{Starvation-Freedom} in \emph{K-version Object-based STM (SF-KOSTM)}, we use chaining hash table (or $ht$) as an underlying data structure where the size of the hash table is \emph{M} buckets as shown in \figref{dsdesign}.(a) and we propose HT-SF-KOSTM. Hash table with bucket size $one$ becomes the linked-list data structure for SF-KOSTM represented as list-SF-KOSTM.  The representation of SF-KOSTM is similar to MVOSTM \cite{Juyal+:SSS:2018}. Each bucket stores multiple nodes in the form of linked-list between the two sentinel nodes \emph{Head}(-$\infty$) and \emph{Tail}(+$\infty$). \figref{dsdesign}.(b) illustrates the structure of each node as $\langle$\emph{key, lock, mark, vl, nNext}$\rangle$. Where, \emph{key} is the unique value from the range of [1 to $\mathcal{K}$] stored in the increasing order between the two sentinel nodes similar to linked-list based concurrent set implementation \cite{Heller+:LazyList:PPL:2007, Harris:NBList:DISC:2001}. The \emph{lock} field is acquired by the transaction before updating (inserting or deleting) on the node. \emph{mark} is the boolean field which says anode is deleted or not. If \emph{mark} sets to true then node is logically deleted else present in the hash table. Here, the deletion is in a lazy manner similar to concurrent linked-list structure \cite{Heller+:LazyList:PPL:2007}. The field \emph{vl} stands for version list. SF-KOSTM maintains the finite say latest $K$-versions corresponding to each key to achieving the greater concurrency as explained in \secref{intro}. Whenever $(K+1)^{th}$ version created for the key then it overwrites the oldest version corresponding to that key. 
 If $K$ is equal to 1, i.e., version list contains only one version corresponding to each key which boils down to \emph{Starvation-Free Single-Version OSTM (SF-SVOSTM)}. So, the data structure of SF-SVOSTM is same as SF-KOSTM with one version. The field \emph{nNext} points to next available node in the linked-list. From now onwards, we will use the term key and node interchangeably.
 \begin{figure}
 	\centerline{\scalebox{0.35}{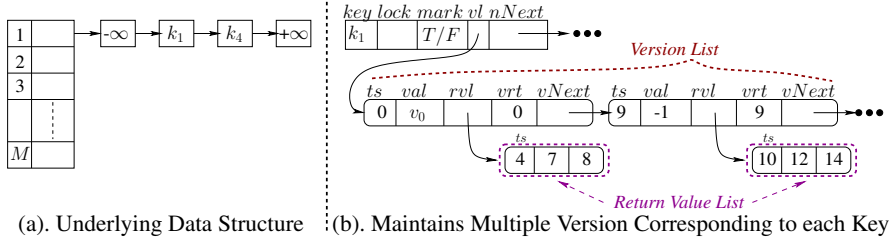}}
\vspace{-.3cm} 	\caption{Design and Data Structure of SF-KOSTM}
 	\label{fig:dsdesign}
 \end{figure} 
\begin{figure}
	\centering
	\begin{minipage}[b]{0.48\textwidth}
		\centering
		\captionsetup{justification=centering}
		\scalebox{.36}{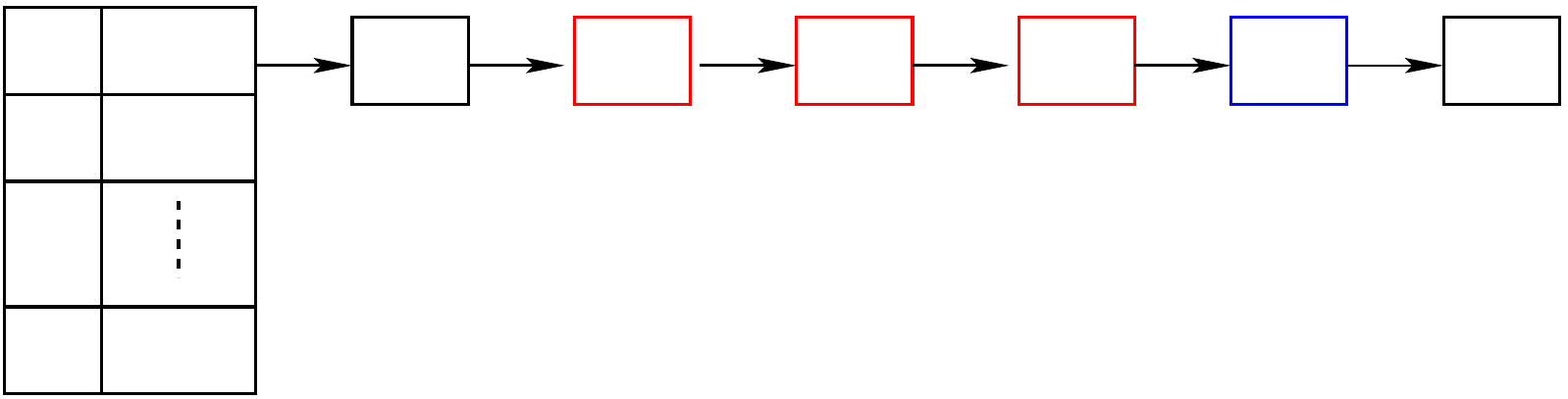}
		 \vspace{-.5cm}
		\caption{Searching $k_9$ over \emph{lazy-list}}
		\label{fig:stmore}
	\end{minipage}   
	\hfill
	\begin{minipage}[b]{0.48\textwidth}
		\scalebox{.36}{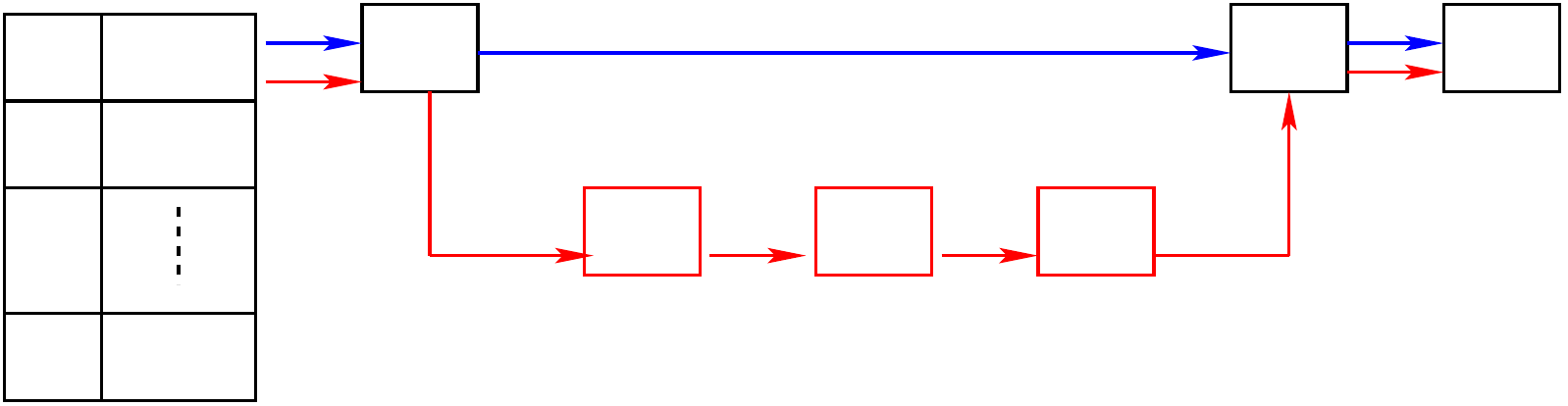}
		\centering
		\captionsetup{justification=centering}
		 \vspace{-.5cm}
		\caption{Searching $k_9$ over $\lsl{}$}
		\label{fig:stless}
	\end{minipage}
\end{figure}

\noindent
The structure of the \emph{vl} is $\langle$\emph{ts, val, rvl, vrt, vNext}$\rangle$ as shown in \figref{dsdesign}.(b). \emph{ts} is the unique timestamp assigned by the \tbeg{()}. If value (\emph{val}) is \emph{nil} then version is created by the \tdel{()} otherwise \tins{()} creates a version with not \emph{nil} value. To satisfy the correctness criteria as \emph{local opacity}, \tdel{()} also maintains the version corresponding to each key with \emph{mark} field as $true$. It allows the concurrent transactions to lookup from the older version of the marked node and returns the value as not \emph{nil}. \emph{rvl} stands for \emph{return value list} which maintains the information about lookup transaction that has lookups from a particular version. It maintains the timestamp (\emph{ts}) of \rvmt{s} (\tlook{()} or \tdel{()}) transaction in it. \emph{vrt} stands for \emph{version real time} which helps to maintain the \emph{real-time order} among the transactions. \emph{vNext} points to next available version in the version list. 

Maintaining the deleted node along with the live (not deleted) node will increase the traversal time to search a particular node in thelist. Consider \figref{stmore}, where red color depicts the deleted node $\langle k_1, k_2, k_4 \rangle$ and blue color depicts the live node $\langle k_9\rangle$. When any method of SF-KOSTM searches the key $k_9$ then it has to traverse the deleted nodes $\langle k_1, k_2, k_4 \rangle$ as well before reach to $k_9$ that increases the traversal time.  

This motivated us to modify the lazy-list structure of a node to form a skip list based on red and blue links. We called it as a \emph{red-blue lazy-list} or \emph{\lsl{}}. This idea has been explored by Peri et al. in SVOSTMs\cite{Peri+:OSTM:Netys:2018}. \emph{\lsl{}} maintains two-pointer corresponding to each node such as red link (\rn{}) and blue link (\bn{}). Where \bn{} points to the live node and \rn{} points to live node as well as deleted node. Let us consider the same example as discussed above with this modification, key $k_9$ is directly searched from the head of the list with the help of \bn{} as shown in \figref{stless}. In this case, traversal time is efficient because any method of SF-KOSTM need not traverse the deleted nodes. To maintain the \rn{} and \bn{} in each node we modify the structure of \emph{lazy-list} as $\langle \emph{key, lock, mark, vl, \rn, \bn, nNext} \rangle$ and called it as \emph{\lsl}.

\vspace{-.4cm}
\subsection{Working of SF-KOSTM Algorithm}
\label{subsec:working}
In this subsection, we describe the working of SF-KOSTM algorithm which includes the detail description of SF-KOSTM methods and challenges to make it starvation-free. This description can easily be extended to SF-MVOSTM and SF-MVOSTM-GC as well. 

SF-KOSTM invokes \tbeg{()}, \tlook{()}, \tdel{()}, \tins{()}, and \tryc{()} methods. \tlook{()} and \tdel{()} work as \rvmt{s()} which lookup the value of key $k$ from shared memory and return it. Whereas \tins{()} and \tdel{()} work as \upmt{s()} that modify the value of $k$ in shared memory. We propose optimistic SF-KOSTM, so, \upmt{s()} first update the value of $k$ in transaction local log $txLog$ and the actual effect of \upmt{s()} will be visible after successful \tryc{()}. Now, we explain the functionality of each method as follows:

\noindent
\textbf{\tbeg{():}} 
When a thread $Th_i$ invokes transaction $T_i$ for the first time (or first incarnation), \tbeg{()} assigns a unique timestamp known as \emph{current timestamp} $(cts)$ using atomic global counter ($gcounter$). If $T_i$ gets aborted then thread $Th_i$ executes it again with new incarnation of $T_i$, say $T_j$ with the new $cts$ until $T_i$ commits but retains its initial $cts$ as \emph{initial timestamp} $(its)$. $Th_i$ uses $its$ to inform the SF-KOSTM system that whether $T_i$ is a new invocation or an incarnation. 
 If $T_i$ is the first incarnation then $its$ and $cts$ are same as $cts_i$ so, $Th_i$ maintains $\langle \emph{$its_i$, $cts_i$}\rangle$. If $T_i$ gets aborted and retries with $T_j$ then $Th_i$ maintains $\langle \emph{$its_i$, $cts_j$}\rangle$.



By assigning priority to the lowest $its$ transaction (i.e. transaction have been in the system for longer time) in \emph{Single-Version OSTM}, \emph{Starvation-Freedom} can easily achieved as explained in \secref{intro}. The detailed working of \emph{Starvation-Free Single-Version OSTM (SF-SVOSTM)} is in \apnref{ap-sfostm}. But achieving \emph{Starvation-Freedom} in finite \emph{K-versions OSTM (SF-KOSTM)} is challenging. Though the transaction $T_i$ has lowest $its$ but $T_i$ may return abort because of finite versions $T_i$ did not find a correct version to lookup from or overwrite a version. \tabref{sfillus1} shows the key insight to achieve the \stfdm in finite K-versions OSTM. Here, we considered two transaction $T_{10}$ and $T_{20}$ with $cts$ 10 and 20 that performs \emph{STM\_lookup()} $($or $l)$ and \emph{STM\_insert()} $($or $i)$ on same key $k$. We assume that a version of $k$ exists with $cts$ 5, so, \emph{STM\_lookup()} of $T_{10}$ and $T_{20}$ find a previous version to lookup and never return abort. Due to the optimistic execution in SF-KOSTM, effect of \emph{STM\_insert()} comes after successful \emph{STM\_tryC()}, so \emph{STM\_lookup()} of a transaction comes before effect of its \emph{STM\_insert()}. Hence, total six permutations are possible as defined in \tabref{sfillus1}. We can observe from the \tabref{sfillus1} that in some cases $T_{10}$ returns abort. But if $T_{20}$ gets the lowest $its$ then $T_{20}$ never returns abort. This ensures that a transaction with lowest $its$ and highest $cts$ will never return abort. But achieving highest $cts$ along with lowest $its$ is bit difficult because new transactions are keep on coming with higher $cts$ using $gcounter$. So, to achieve the highest $cts$, we introduce a new timestamp as \emph{working timestamp (wts)} which is significantly larger than $cts$. 

\begin{table}[H]
	\small
	\begin{tabular}{ | m{1cm} | m{4cm}| m{7cm} | } 
		\hline
		S. No.& Execution Sequence & Possible actions by Transactions  \\ 
		\hline
		1. & $l_{10}(k), i_{10}(k), l_{20}(k), i_{20}(k)$ & $T_{20}(k)$ lookups the version inserted by $T_{10}$. No conflict. \\
		\hline
		2. & $l_{10}(k), l_{20}(k), i_{10}(k), i_{20}(k)$ & Conflict detected at $i_{10}(k)$. Either abort $T_{10}$ or $T_{20}$. \\
		\hline
		3. & $l_{10}(k), l_{20}(k), i_{20}(k), i_{10}(k)$ & Conflict detected at $i_{10}(k)$. Hence, abort $T_{10}$. \\
		\hline
		4. & $l_{20}(k), l_{10}(k), i_{20}(k), i_{10}(k)$ & Conflict detected at $i_{10}(k)$. Hence, abort $T_{10}$. \\
		\hline
		5. & $l_{20}(k), l_{10}(k), i_{10}(k), i_{20}(k)$ & Conflict detected at $i_{10}(k)$. Either abort $T_{10}$ or $T_{20}$. \\
		\hline
		6. & $l_{20}(k), i_{20}(k), l_{10}(k), i_{10}(k)$ & Conflict detected at $i_{10}(k)$. Hence, abort $T_{10}$.\\
		\hline
	\end{tabular}
	\captionsetup{justification=centering}
	\caption{Possible Permutations of Methods}
	\label{tab:sfillus1}
\end{table}

\vspace{-.4cm}
\noindent
\tbeg{()} maintains the $wts$ for transaction $T_i$ as $wts_i$, which is potentially higher timestamp as compare to $cts_i$. So, we derived,
\setlength\abovedisplayskip{0pt}
\begin{equation}
\label{eq:wtsf}
\twts{i} = \tcts{i} + C * (\tcts{i} - \tits{i});
\vspace{-.2cm}
\end{equation}
where C is any constant value greater than 0. When $T_i$ is issued for the first time then $wts_i$, $cts_i$, and $its_i$ are same. If $T_i$ gets aborted again and again then drift between the $cts_i$ and $wts_i$ will increases. The advantage for maintaining $wts_i$ is if any transaction keeps getting aborted then its $wts_i$ will be high and $its_i$ will be low. Eventually, $T_i$ will get chance to commit in finite number of steps to achieve starvation-freedom. For simplicity, we use timestamp (\emph{ts}) $i$ of $T_i$ as $wts_i$, i.e., $\langle wts_i$ = $i\rangle$ for SF-KOSTM. 
\vspace{-.2cm}
\begin{observation}
Any transaction $T_i$ with lowest $its_i$ and highest $wts_i$ will never abort.
\end{observation}
\vspace{-.1cm}
Sometimes, the value of $wts$ is significantly larger than $cts$. So, $wts$ is unable to maintain \emph{real-time order} between the transactions which violates the correctness of SF-KOSTM (shown in \figref{sfmv-correct} in \apnref{wtsdis}). To address this issue SF-KOSTM uses the idea of timestamp ranges \cite{Riegel+:LSA:DISC:2006,Guer+:disc1:2008, Crain+:RI_VWC:ICA3PP:2011} along with $\langle its_i, cts_i, wts_i \rangle$ for transaction $T_i$ in \tbeg{()}. It maintains the \emph{transaction lower timestamp limit ($tltl_i$)} and \emph{transaction upper timestamp limit ($tutl_i$)} for $T_i$. Initially, $\langle its_i, cts_i, wts_i, tltl_i \rangle$ are the same for $T_i$. $tutl_i$ would be set as a largest possible value denoted as $+\infty$ for $T_i$. After successful execution of \emph{\rvmt{s()}} or \tryc{()} of $T_i$, $tltl_i$ gets incremented and $tutl_i$ gets decremented\footnote{Practically $\infty$ can't be decremented for $tutl_i$ so we assign the highest possible value to $tutl_i$ which gets decremented.} to respect the real-time order among the transactions. \tbeg{()} initializes the \emph{transaction local log $(txLog_i)$} for each transaction $T_i$ to store the information in it. Whenever a transaction starts it atomically sets its \emph{status} to be \emph{live} as a global variable. Transaction \emph{status} can be $\langle \emph{live, commit, false}\rangle$. After successful execution of \tryc{()}, $T_i$ sets its \emph{status} to be \emph{commit}. If \emph{status} of the transaction is \emph{false} then it returns \emph{abort}. For more details please refer the high level view of \tbeg{()} in \algoref{begin1}.

\begin{algorithm}[H] 
	\caption{\emph{STM\_begin($its$)}: This method is invoke by a thread $Th_i$ to start a new transaction $T_i$. It pass a parameter $its$ which is the initial timestamp of the first incarnation of $T_i$. If $T_i$ is the first incarnation then $its$ is $nil$.}
	\label{algo:begin1} 
	\setlength{\multicolsep}{0pt}
	\begin{algorithmic}[1]
		\makeatletter\setcounter{ALG@line}{0}\makeatother
		\Procedure{\emph{STM\_begin($its$)}}{}  \label{lin:begin1}
		\State Create a local log $txLog_i$ for each transaction. \label{lin:begin2}
		\If {($its == nil$)} \label{lin:begin3}
		\State /* Atomically get the value from the global counter and set it to its, cts, and \blank{1cm} wts.*/ \label{lin:begin4}
		\State $its_i$ = $cts_i$ = $wts_i$ = \emph{gcounter.get\&Inc()}; \label{lin:begin5}
		\Else  \label{lin:begin6}
		\State /*Set the $its_i$ to first incarnation of $T_i$ $its$*/ \label{lin:begin7}
		\State $its_i$ = $its$; \label{lin:begin8}
		\State /*Atomically get the value from the global counter for $cts_i$*/ \label{lin:begin9}
		\State $cts_i$ = \emph{gcounter.get\&Inc()}. \label{lin:begin10}
		\State /*Set the $wts$ value with the help of $cts_i$ and $its_i$*/ \label{lin:begin11}
		\State $wts_i$ = $cts_i$+C*($cts_i$-$its_i$). \label{lin:begin12}
		\EndIf \label{lin:begin13}
		\State /*Set the $tltl_i$ as $cts_i$*/ \label{lin:begin14}
		\State $tltl$ = $cts_i$.   \label{lin:begin15}
		\State /*Set the $tutl_i$ as possible large value*/ \label{lin:begin16}
		\State $tutl_i$ = $\infty$. \label{lin:begin17}
		\State /*Initially, set the $status_i$ of $T_i$ as $live$*/ \label{lin:begin18}
		\State $status_i$ = $live$;  \label{lin:begin19}
		\State return $\langle cts_i, wts_i\rangle$ \label{lin:begin20}
		
		\EndProcedure
	\end{algorithmic}
\end{algorithm}



\noindent
\textbf{\tlook{()} and \tdel{()} as \rvmt{s()}:} \emph{\rvmt{s(ht, k, val)}} return the value (\emph{val}) corresponding to the key \emph{k} from the shared memory as hash table (\emph{ht}). We show the high level overview of the \rvmt{s()} in \algoref{rvmt}. First, it identifies the key $k$ in the transaction local log as \emph{$txLog_i$} for transaction $T_i$. If $k$ exists then it updates the $txLog_i$ and returns the $val$ at \Lineref{rvm2}.

If key $k$ does not exist in the $txLog_i$ then before identify the location in share memory \rvmt{s()} check the \emph{status} of $T_i$ at \Lineref{rvm5}. If \emph{status} of $T_i$ (or $i$) is \emph{false} then $T_i$ has to \emph{abort} which says that $T_i$ is not having the lowest $its$ and highest $wts$ among other concurrent conflicting transactions. So, to propose starvation-freedom in SF-KOSTM other conflicting transactions set the \emph{status} of $T_i$ as \emph{false} and force it to \emph{abort}.

If \emph{status} of $T_i$ is not \emph{false} and key $k$ does not exist in the $txLog_i$ then it identifies the location of key $k$ optimistically (without acquiring the locks similar to the \emph{lazy-list}\cite{Heller+:LazyList:PPL:2007}) in the shared memory at \Lineref{rvm7}. SF-KOSTM maintains the shared memory in the form of hash table with $M$ buckets as shown in \subsecref{dsde}, where each bucket stores the keys in \emph{\lsl}. Each node contains two pointer $\langle \rn, \bn \rangle$. So, it identifies the two \emph{predecessors (pred)} and two \emph{current (curr)} with respect to each node. First, it identifies the pred and curr for key $k$ in \bn{} as $\langle \bp, \bc \rangle$. After that it identifies the pred and curr for key $k$ in \rn{} as $\langle \rp, \rc \rangle$. If $\langle \rp, \rc \rangle$ are not marked then $\langle \bp=\rp, \bc=\rc\rangle$. SF-KOSTM maintains the keys are in increasing order. So, the order among the nodes are $\langle \bp.key \leq \rp.key < k \leq \rc.key \leq \bc.key\rangle$.

\rvmt{s()} acquire the lock in predefined order on all the identified preds and currs for key $k$ to avoid the deadlock at \Lineref{rvm8} and do the \emph{rv\_Validation()} as shown in \algoref{rvm}. If $\langle \bp \vee \bc\rangle$ is marked or preds are not pointing to identified currs as $\langle (\bp.\bn \neq \bc) \vee (\rp.\rn \neq \rc)\rangle$ then it releases the locks from all the preds and currs and identify the new preds and currs for $k$ in shared memory. 

\begin{algorithm}[H]
	
	\caption{\emph{rv\_Validation(preds[], currs[])}: It is mainly used for \emph{rv\_method()} validation.}
	\label{algo:rvm} 
	\begin{algorithmic}[1]
		\makeatletter\setcounter{ALG@line}{21}\makeatother
		\Procedure{$rv\_Validation{(preds[], currs[])}$}{} \label{lin:rvv1}
		\If{$((\bp.mark) || (\bc.mark) || ((\bp.\bn)\neq \blank{1.7cm} \bc) || ((\rp.\rn) \neq {\rc})$)}\label{lin:rvv2}
		return $\langle false \rangle$. 
		\Else{} return $\langle true \rangle$. \label{lin:rvv3}
		\EndIf \label{lin:rvv4}
		\EndProcedure \label{lin:rvv5}
	\end{algorithmic}
\end{algorithm}

\vspace{.4mm}
\setlength{\intextsep}{0pt}
\begin{algorithm}
	\caption{\emph{\rvmt{s(ht, k, val)}:} It can either be $\tdel_i(ht, k, val)$ or $\tlook_i(ht, k, val)$ on key $k$ by transaction $T_i$.} \label{algo:rvmt} 	
	\setlength{\multicolsep}{0pt}
		\begin{algorithmic}[1]
			\makeatletter\setcounter{ALG@line}{26}\makeatother
			\Procedure{$\rvmt{s}_i$}{$ht, k, val$}	\label{lin:rvm1}	
			\If{($k \in \llog_i$)} 
			\State Update the local log of $T_i$ and return $val$. \label{lin:rvm2}
			\Else \label{lin:rvm3}
			\State /*Atomically check the \emph{status} of its own transaction $T_i$ (or $i$).*/ \label{lin:rvm4} 
			\If{(\emph{i.status == false})} return $\langle abort_i \rangle$. \label{lin:rvm5}
			\EndIf\label{lin:rvm6}
			\State Identify the \emph{preds[]} and \emph{currs[]} for key $k$ in bucket $M_k$ of \emph{\lsl} using \blank{1cm} \bn and \rn . \label{lin:rvm7}
			\State Acquire locks on \emph{preds[]} \& \emph{currs[]} in increasing order of keys to avoid the \blank{1cm} deadlock. \label{lin:rvm8}
			\If{(\emph{!rv\_Validation(preds[], currs[])})} \label{lin:rvm9}
			\State Release the locks and goto \linref{rvm7}. \label{lin:rvm10}
			\EndIf \label{lin:rvm11}
			\If{($k  ~ \notin ~ M_k.\lsl$)} \label{lin:rvm12} 
			\State Create a new node $n$ with key $k$ as: $\langle$\emph{key=k, lock=false, mark=true, \blank{1.5cm} vl=ver, nNext=}$\phi\rangle$./*$n$ is marked*/  \label{lin:rvm13}
			\State Create version $ver$ as:$\langle$\emph{ts=0, val=nil, rvl=i, vrt=0, vNext=$\phi$}$\rangle$. \label{lin:rvm14}
			\State Insert $n$ into $M_k.$\emph{\lsl} s.t. it is accessible only via \rn{s}. /*\emph{lock} sets \blank{1.5cm} \emph{true}*/  \label{lin:rvm15}  	
			\State Release locks; update the $\llog_i$ with $k$.\label{lin:rvm16}
			\State return $\langle$\emph{val}$\rangle$. /*\emph{val} as $nil$*/\label{lin:rvm17}
			\EndIf \label{lin:rvm18}
			\State Identify the version $ver_j$ with $ts=j$ such that $j$ is the \emph{largest timestamp \blank{1cm} smaller (lts)} than $i$.\label{lin:rvm19}
			\If{($ver_j$ == $nil$)} /*Finite Versions*/\label{lin:rvm111}
			\State return $\langle abort_i\rangle$ 
			\ElsIf{($ver_j.vNext$ != $nil$)}\label{lin:rvm20}
			\State /*$tutl_i$ should be less then \emph{vrt} of next version $ver_j$*/\label{lin:rvm21}
			\State 	Calculate $tutl_i$ = min($tutl_i$, $ver_j.vNext$ $.vrt-1$).\label{lin:rvm22}
			\EndIf\label{lin:rvm23}
			\State /*$tltl_i$ should be greater then $vrt$ of $ver_j$*/\label{lin:rvm24}
			\State 	Calculate $tltl_i$ = max($tltl_i$, $ver_j.vrt+1$).\label{lin:rvm25}
			\State /*If limit has crossed each other then abort $T_i$*/\label{lin:rvm26}
			\If{($tltl_i$ $>$ $tutl_i$)} return $\langle abort_i \rangle$. \label{lin:rvm27}
			\EndIf\label{lin:rvm28}
			\State Add $i$ into the $rvl$ of $ver_j$. \label{lin:rvm29} 
			\State Release the locks; update the $\llog_i$ with $k$ and value. \label{lin:rvm30}
			\EndIf \label{lin:rvm31}
			\State return $\langle$$ver_j.val$$\rangle$. \label{lin:rvm32}
			\EndProcedure \label{lin:rvm33}
		\end{algorithmic}
\end{algorithm}
\vspace{.2mm}

If key $k$ does not exist in the \emph{\lsl{}} of corresponding bucket $M_k$ at \Lineref{rvm12} then it creates a new node $n$ with key $k$ as $\langle \emph{key=k, lock=false, mark=true, vl=ver, nNext=}\phi \rangle$ at \Lineref{rvm13} and creates a version (\emph{ver}) for transaction $T_0$ as $\langle$\emph{ts=0, val=nil, rvl=i, vrt=0, vNext=}$\phi\rangle$ at \Lineref{rvm14}. Transaction $T_i$ creates the version of $T_0$, so, other concurrent conflicting transaction (say $T_p$) with lower timestamp than $T_i$, i.e., $\langle p < i\rangle$ can lookup from $T_0$ version. Thus, $T_i$ save $T_p$ to abort while creating a $T_0$ version and ensures greater concurrency. After that $T_i$ adds its \emph{$wts_i$} in the \emph{rvl} of $T_0$ and sets the \emph{vrt} 0 as timestamp of $T_0$ version. Finally, it insert the node $n$ into $M_k.\lsl$ such that it is accessible via \rn{} only at \Lineref{rvm15}. \rvmt{()} releases the locks and update the $txLog_i$ with key $k$ and value as $nil$ (\Lineref{rvm16}). Eventually, it returns the \emph{val} as $nil$ at \Lineref{rvm17}.      

\vspace{-.5mm}

If key $k$ exists in the $M_k.\lsl$ then it identifies the current version $ver_j$ with $ts$ = $j$ such that $j$ is the \emph{largest timestamp smaller (lts)} than $i$ at \Lineref{rvm19} and there exists no other version with timestamp $p$ by $T_p$ on same key $k$ such that $\langle \emph{j $<$ p $<$ i}\rangle$. If $ver_j$ is $nil$ at \Lineref{rvm111} then SF-KOSTM returns $abort$ for transaction $T_i$ because it does not found version to lookup otherwise it identifies the next version with the help of $ver_j.vNext$. If next version ($ver_j.vNext$ as $ver_k$) exist then $T_i$ maintains the $tutl_i$ with minimum of $\langle tutl_i$ $\vee$ $ver_k.vrt-1\rangle$ at \Lineref{rvm22} and $tltl_i$ with maximum of $\langle tltl_i \vee $ $ver_j.vrt+1$$\rangle$ at \Lineref{rvm25} to respect the \emph{real-time order} among the transactions. If $tltl_i$ is greater than $tutl_i$ at \Lineref{rvm27} then transaction $T_i$ returns \emph{abort} (fail to maintains real-time order) otherwise it adds the $ts$ of $T_i$ ($wts_i$) in the $rvl$ of $ver_j$ at \Lineref{rvm29}. Finally, it releases the lock and update the $txLog_i$ with key $k$ and value as current version value ($ver_j.val$) at \Lineref{rvm30}. Eventually, it returns the value as $ver_j.val$ at \Lineref{rvm32}.

\noindent
\textbf{\tins{()} and \tdel{()} as \upmt{s()}:} Actual effect of \tins{()} and \tdel{()} come after successful \tryc{()}. They create the version corresponding to the key in shared memory. We show the high level view of \tryc{()} in \algoref{tryc}. First, \tryc{()} checks the \emph{status} of the transaction $T_i$ at \Lineref{tc3}. If \emph{status} of $T_i$ is \emph{false} then $T_i$ returns \emph{abort} with similar reasoning explained above in \rvmt{()}.

If \emph{status} is not false then \tryc{()} sort the keys (exist in $txLog_i$ of $T_i$) of \upmt{s()} in increasing order. 
It takes the method ($m_{ij}$) from $txLog_i$ one by one and identifies the location of the key $k$ in \emph{$M_k$.\lsl} as explained above in \rvmt{()}. After identifying the preds and currs for $k$ it acquire the locks in predefined order to avoid the deadlock at \Lineref{tc10} and calls \emph{tryC\_Validation()} to validate the methods of $T_i$.


\noindent
\textbf{\emph{tryC\_Validation()}:} It identifies whether the methods of invoking transaction $T_i$ are able to create or delete a version corresponding to the keys while ensuring the \emph{starvation-freedom} and maintaining the \emph{real-time order} among the transactions.


First, it do the \emph{rv\_Validation()} at \Lineref{tcv2} as explained in \rvmt{()}. If \emph{rv\_Validation()} is successful and key $k$ exists in the $M_k.\lsl$ then it identifies the current version $ver_j$ with $ts=j$ such that $j$ is the \emph{largest timestamp smaller (lts)} than $i$ at \Lineref{tcv5}. If $ver_j$ is $null$ at \Lineref{tcv111} then SF-KOSTM returns $abort$ for transaction $T_i$ because it does not find the version to replace otherwise after identifying the current version $ver_j$ it maintains the Current Version List (currVL), Next Version List (nextVL), All Return Value List (allRVL), Large Return Value List (largeRVL), Small Return Value List (smallRVL) from $ver_j$ of key $k$ at \Lineref{tcv6}. currVL and nextVL maintain the previous closest version and next immediate version of all the keys accessed in \tryc{()}. allRVL keeps the currVL.rvl whereas largeRVL and smallRVL stores all the $wts$ of currVL.rvl such that ($wts_{currVL.rvl}$ $>$ $wts_i$) and  ($wts_{currVL.rvl}$ $<$ $wts_i$) respectively. Acquire the locks on \emph{status} of all the transactions present in allRVL list including $T_i$ it self in predefined order to avoid the deadlock at \Lineref{tcv8}. First, it checks the \emph{status} of its own transaction $T_i$ at \Lineref{tcv10}. If \emph{status} of $T_i$ is \emph{false} then $T_i$ has to \emph{abort} the same reason as explained in \rvmt{()}. 

If the \emph{status} of $T_i$ is not \emph{false} then it compares the $its_i$ of its own transaction $T_i$ with the $its_p$ of other transactions $T_p$ present in the largeRVL at \Lineref{tcv13}. Along with this it checks the \emph{status} of $p$. If above conditions $\langle (its_i < its_p) \&\& (p.status==live))\rangle$ succeed then it includes $T_p$ in the Abort Return Value List (abortRVL) at \Lineref{tcv14} to abort it later otherwise abort $T_i$ itself at \Lineref{tcv15}.  

\vspace{.2mm}
\begin{algorithm}
	\caption{\emph{tryC\_Validation():} It is use for \tryc{()} validation.}
	\label{algo:trycVal} 
	\begin{algorithmic}[1]
		\makeatletter\setcounter{ALG@line}{62}\makeatother
		\Procedure{\emph{tryC\_Validation{()}}}{} \label{lin:tcv1}
		\If{(\emph{!rv\_Validation()})}
		Release the locks and \emph{retry}.\label{lin:tcv2}
		\EndIf\label{lin:tcv3}
		\If{(k $\in$ $M_k.\lsl$)}\label{lin:tcv4}
		\State Identify the version $ver_j$ with $ts=j$ such that $j$ is the \emph{largest timestamp \blank{1cm} smaller (lts)} than $i$ and there exists no other version with timestamp $p$ by $T_p$ \blank{1cm} on key $k$ such that $\langle \emph{j $<$ p $<$ i}\rangle$.\label{lin:tcv5}
		\If{($ver_j$ == $null$)} /*Finite Versions*/
		\State return $\langle abort_i\rangle$ \label{lin:tcv111}
		\EndIf
		\State Maintain the list of $ver_j$, $ver_j.vNext$, $ver_j.rvl$, $(ver_j.rvl>i)$, and \blank{1cm} $(ver_j.rvl<i)$ as prevVL, nextVL, allRVL, largeRVL and smallRVL \blank{1cm} respectively for all key \emph{k} of $T_i$.\label{lin:tcv6}
		\State /*$p$ is the timestamp of transaction $T_p$*/
		\If{($p$ $\in$ allRVL)} /*Includes $i$ as well in allRVL*/\label{lin:tcv7}
		\State Lock \emph{status} of each $p$ in pre-defined order. \label{lin:tcv8}
		\EndIf\label{lin:tcv9}
		\If{(\emph{i.status == false})} return $\langle false \rangle$. \label{lin:tcv10}
		\EndIf\label{lin:tcv11}
		\ForAll{($p$ $\in$ largeRVL)}\label{lin:tcv12}
		\If{(($its_i$$<$$its_p$)$\&\&$(\emph{p.status==live}))} \label{lin:tcv13}
		\State Maintain \emph{abort list} as abortRVL \& includes \emph{p} in it.\label{lin:tcv14}
		\Else{} return $\langle false \rangle$. /*abort $i$ itself*/\label{lin:tcv15}
		\EndIf\label{lin:tcv16}
		
		\EndFor\label{lin:tcv17}
		\ForAll{($ver$ $\in$ nextVL)}\label{lin:tcv21}
		\State Calculate $tutl_i$ = min($tutl_i$, $ver.vNext$$.vrt-1$).\label{lin:tcv22}
		\EndFor\label{lin:tcv23}
		\ForAll{($ver$ $\in$ currVL)}\label{lin:tcv18}
		\State Calculate $tltl_i$ = max($tltl_i$, $ver.vrt+1$).\label{lin:tcv19}
		\EndFor\label{lin:tcv20}
		\State /*Store current value of global counter as commit time and increment it.*/
		\State $comTime$ = \emph{gcounter.add\&get(incrVal)};  \label{lin:tcv201}
		\State Calculate $tutl_i$ = min($tutl_i$, $comTime$); \label{lin:tcv202}
		\If{($tltl_i$ $>$ $tutl_i$)} /*abort $i$ itself*/
		\State return $\langle false \rangle$. \label{lin:tcv24}
		\EndIf\label{lin:tcv25}
		\ForAll{($p$ $\in$ smallRVL)}\label{lin:tcv26}
		\If{($tltl_p$ $>$ $tutl_i$)}\label{lin:tcv27}
		\If{(($its_i$$<$$its_p$)$\&\&$(\emph{p.status==live}))} \label{lin:tcv28}
		\State Includes $p$ in abortRVL list.\label{lin:tcv29}
		\Else{} return $\langle false \rangle$. /*abort $i$ itself*/\label{lin:tcv30}
		\EndIf\label{lin:tcv31}
		\EndIf	\label{lin:tcv32}	
		\EndFor \label{lin:tcv33}
		\State $tltl_i$ = $tutl_i$. /*After this point $i$ can't abort*/\label{lin:tcv34}
		\ForAll{($p$ $\in$ smallRVL)}
		\State /*Only for \emph{live} transactions*/ \label{lin:tcv35}
		\State Calculate the $tutl_p$ = min($tutl_p$, $tltl_i-1$).\label{lin:tcv36}
		\EndFor \label{lin:tcv37}
		\ForAll{($p$ $\in$ abortRVL)} \label{lin:tcv38}
		\State Set the \emph{status} of $p$ to be $false$. \label{lin:tcv39}
		\EndFor \label{lin:tcv40}
		\EndIf\label{lin:tcv41}
		\State return $\langle true \rangle$.\label{lin:tcv42}
		\EndProcedure \label{lin:tcv43}
	\end{algorithmic}
\end{algorithm}
\vspace{.2mm}

After that \tryc{()} maintains the $tltl_i$ and $tutl_i$ of transaction $T_i$ at \Lineref{tcv19} and \Lineref{tcv22}. The requirement of $tltl_i$ and $tutl_i$ is explained above in the \rvmt{()}. If limit of $tltl_i$ crossed with $tutl_i$ then $T_i$ have to abort at \Lineref{tcv24}. If $tltl_p$ greater than $tutl_i$ at \Lineref{tcv27} then it checks the $its_i$ and $its_p$. If $\langle (its_i < its_p) \&\& (p.status==live))\rangle$ then add the transaction $T_p$ in the abortRVL for all the smallRVL transactions at \Lineref{tcv29} otherwise, \emph{abort} $T_i$ itself at \Lineref{tcv30}.

At \Lineref{tcv34}, $tltl_i$ would be equal to $tutl_i$ and after this step transaction $T_i$ will never \emph{abort}. $T_i$ helps the other transaction $T_p$ to update the $tutl_p$ which exists in the smallRVL and still $live$ then it sets the $tutl_p$ to minimum of $\langle tutl_p \vee tltl_i-1\rangle$ to maintain the real-time order among the transaction at \Lineref{tcv36}. At \Lineref{tcv39}, \tryc{()} aborts all other conflicting transactions which are present in the abortRVL while modifying the \emph{status} field to be \emph{false} to achieve \emph{starvation-freedom}.

\cmnt{
\emph{tryC\_Validation()} identifies whether the methods of invoking transaction $T_i$ are able to insert or delete a version corresponding to the keys while ensuring the \emph{starvation-freedom} and maintaining the \emph{real-time order} among the transactions. It follow the following steps for validation. Step 1: First, it do the \emph{rv\_Validation()} as explained in \rvmt{()} above. Step 2: If \emph{rv\_Validation()} is successful and key $k$ is exist in the $M_k.\lsl$ then it identifies the current version $ver_j$ with $ts=j$ such that $j$ is the \emph{largest timestamp smaller (lts)} than $i$. If $ver_j$ is \emph{not exist} then SF-KOSTM returns $abort$ for transaction $T_i$ because of bounded K-versions otherwise maintains the information of $ver_j$ and its next version ($ver_j.vNext$) which helps transaction $T_i$ to sets its $tltl_i$ and $tutl_i$. Step 3: If $wts_i$ of $T_i$ is less then other $live$ transactions $wts$ exist in $ver_j.rvl$ then $T_i$ sets the $status$ to be $false$ to all conflicting \emph{live} transactions otherwise $T_i$ returns $abort$. The detailed descriptions are in \apnref{ap-rcode}.
}

If all the steps of the \emph{tryC\_Validation()} is successful then the actual effect of the \tins{()} and \tdel{()} will be visible to the shared memory. At \Lineref{tc16}, \tryc{()} checks for \emph{poValidation()}. When two subsequent methods $\langle m_{ij}, m_{ik}\rangle$ of the same transaction $T_i$ identify the overlapping location of preds and currs in \emph{\lsl}. Then \emph{poValidation()} updates the current method $m_{ik}$ preds and currs with the help of previous method $m_{ij}$ preds and currs. 

If $m_{ij}$ is \tins{()} and key $k$ is not exist in the $M_k.\lsl{}$ then it creates the new node $n$ with key $k$ as $\langle \emph{key=k, lock=false, mark=false, vl=ver, nNext=}\phi \rangle$ at \Lineref{tc18}. Later, it creates a version ($ver$) for transaction $T_0$ and $T_i$ as $\langle$ \emph{ts=0, val=nil, rvl=i, vrt=0, vNext=i} $\rangle$ and $\langle \emph{ts=i, val=v, rvl=}\phi, \emph{vrt=i, vNext=}\phi \rangle$ at \Lineref{tc19}. The $T_0$ version created by transaction $T_i$ to helps other concurrent conflicting transactions (with lower timestamp than $T_i$) to lookup from $T_0$ version. Finally, it insert the node $n$ into $M_k.\lsl$ such that it is accessible via \rn{} as well as \bn{} at \Lineref{tc20}. If $m_{ij}$ is \tins{()} and key $k$ is exist in the $M_k.\lsl{}$ then it creates the new version $ver_i$ as $\langle \emph{ts=i, val=v, rvl=}\phi, \emph{vrt=i, vNext=}\phi \rangle$ corresponding to key $k$. If the limit of the version reach to $K$ then SF-KOSTM replaces the oldest version with $(K+1)^{th}$ version which is accessible via \rn{} as well as \bn{} at \Lineref{tc22}. 

\vspace{.6mm}
\begin{algorithm}
	
	\caption{\emph{\tryc($T_i$)}: Validate the \upmt{s()} of $T_i$ and returns \emph{commit}.}
	\setlength{\multicolsep}{-1pt}
	\label{algo:tryc}
		\begin{algorithmic}[1]
			\makeatletter\setcounter{ALG@line}{111}\makeatother
			\Procedure{$\tryc{(T_i)}$}{} \label{lin:tc1}
			\State /*Atomically check the \emph{status} of its own transaction  $T_i$ (or $i$)*/ \label{lin:tc2}
			\If{(\emph{i.status == false})} return $\langle abort_i \rangle$. \label{lin:tc3}
			\EndIf \label{lin:tc4}
			\State /*Sort the $keys$ of $\llog_i$ in increasing order.*/ \label{lin:tc5}
			\State /*Method ($m$) will be either \tins or \emph{STM\_delete}*/\label{lin:tc6}
			\ForAll{($m_{ij}$ $\in$ $\llog_i$)} \label{lin:tc7}
			\If{\hspace{-.07cm}($m_{ij}$==\tins$||$$m_{ij}$==\tdel)\hspace{-.07cm}}\label{lin:tc8}
			\State Identify the \emph{preds[]} \& \emph{currs[]} for key \emph{k} in bucket $M_k$ of \emph{\lsl} \blank{1.5cm} using
			\bn\& \rn. \label{lin:tc9}
			\State Acquire the locks on \emph{preds[]} \& \emph{currs[]} in increasing order of keys to \blank{1.5cm} avoid deadlock. \label{lin:tc10}
			\If{($!tryC\_Validation()$)}
			\State return $\langle abort_i \rangle$.\label{lin:tc11}
			\EndIf\label{lin:tc12}
			\EndIf\label{lin:tc13}
			\EndFor\label{lin:tc14}
			\ForAll{($m_{ij}$ $\in$ $\llog_i$)}\label{lin:tc15}
			\State \emph{poValidation()}  modifies the \emph{preds[]} \& \emph{currs[]} of current method which \blank{1cm} would have been updated by previous method of the same transaction.\label{lin:tc16}
			\If{(($m_{ij}$==\tins)\&\&(k$\notin$$M_k$.\lsl))}\label{lin:tc17}
			\State Create new node $n$ with $k$ as: $\langle$\emph{key=k, lock=false, mark= false, vl=ver, \blank{1.5cm} nNext=$\phi$}$\rangle$. \label{lin:tc18}
			\State Create first version $ver$ for $T_0$ and next for $i$: $\langle$\emph{ts=i, val=v, rvl=$\phi$, vrt=i, \blank{1.5cm} vNext=$\phi$}$\rangle$.\label{lin:tc19}
			\State Insert node $n$ into $M_k$.\emph{\lsl} such that it is accessible via \rn{} as \blank{1.5cm} well as \bn{}.\label{lin:tc20}
			\State /*\emph{lock} sets \emph{true}*/  
			\ElsIf{($m_{ij}$ == \tins{})}\label{lin:tc21}
			\State Add $ver$: $\langle$\emph{ts=i, val=v, rvl=$\phi$, vrt=i, vNext=$\phi$}$\rangle$ into $M_k$.\emph{\lsl} \& \blank{1.5cm} accessible via \rn{}, \bn. /*\emph{mark=false}*/\label{lin:tc22}
			\EndIf\label{lin:tc23}
			\If{($m_{ij}$ == \tdel{})}\label{lin:tc24}
			\State Add \emph{ver}:$\langle$\emph{ts=i, val=nil, rvl=$\phi$, vrt=i, vNext=$\phi$}$\rangle$ into  $M_k$.\emph{\lsl} \& \blank{1.5cm} accessible via \rn only. /*\emph{mark=true}*/\label{lin:tc25}
			\EndIf\label{lin:tc26}
			\State Update \emph{preds[]} \& \emph{currs[]} of $m_{ij}$ in $\llog_i$.\label{lin:tc27}
			
			\EndFor \label{lin:tc28}
			\State Release the locks; return $\langle commit_i \rangle$.\label{lin:tc29}
			\EndProcedure \label{lin:tc30}
		\end{algorithmic}
\end{algorithm}

\vspace{.2mm}

\noindent
If $m_{ij}$ is \tdel{()} and key $k$ is exist in the $M_k.\lsl{}$ then it creates the new version $ver_i$ as $\langle \emph{ts=i, val=nil, rvl=}\phi, \emph{vrt=i, vNext=}\phi \rangle$ which is accessible via \rn{} only at \Lineref{tc25}. At last it updates the preds and currs of each $m_{ij}$ into its $txLog_i$ to help the upcoming methods of the same transactions in \emph{poValidation()} at \Lineref{tc27}. Finally, it releases the locks on all the keys in predefined order and returns \emph{commit} at \Lineref{tc29}.

\cmnt{
\vspace{-.15cm}

\begin{theorem}
Any legal history $H$ generated by SF-SVOSTM satisfies co-opacity.
\end{theorem}
\vspace{-.3cm}
\begin{theorem}
Any valid history $H$ generated by SF-KOSTM satisfies  local-opacity.
\end{theorem}
\vspace{-.3cm}
\begin{theorem}
SF-SVOSTM and SF-KOSTM ensure starvation-freedom in presence of a fair scheduler that satisfies \asmref{bdtm}(bounded-termination) and in the absence of parasitic transactions that satisfies \asmref{self}.
\end{theorem}
\vspace{-.2cm}
Please find the proof of theorems in \apnref{gcofo}, \apnref{gcoflo}, and \apnref{ap-cc}.
}
\section{Graph Characterization of Co-opacity and SF-SVOSTM Correctness}
\label{sec:gcofo}

This section describes the graph characterization of the history $H$ which helps to prove the correctness of STMs. We follow the  graph characterization by Guerraoui and Kapalka \cite{tm-book} and modified it for sequential histories with high-level methods to prove the correctness of SF-SVOSTM. 

As discussed in \secref{sm}, SF-SVOSTM executes high-level methods through transactions on history $H$ which internally invoke multiple read-write (or lower-level) operations including invocation and response known as \emph{events} (or $evts$). So, high-level methods are interval instead of dots (atomic). Methods of same transaction $T_i$ are always real-time ordered, i.e., none of the methods of $T_i$ overlaps each other. But due to the concurrent execution of history $H$ with methods are interval, two methods from different transactions may overlap. Thus, we order the overlapping methods of transactions based on their \emph{linearization point (LP)}. We consider \emph{first unlocking point} of each successful method as the \emph{LP} of the respective method.

In the concurrent execution of a history $H$, we make high-level methods of a transaction as atomic based on their \emph{linearization points (LPs)} as defined above. But, as we know from \secref{sm}, a transaction internally invokes multiple high-level methods and transactions are overlapping to each other in concurrent history $H$. So, with the help of \emph{graph characterization of co-opacity}, SF-SVOSTM ensures the atomicity of the transaction. 

A graph for the conflict-opacity (co-opacity) is represented as $\copg{H}{\ll} = (V, E)$ which consists of $V$ vertices and $E$ edges. Here, each committed transaction is consider as a vertex and edges are as follows:
\begin{itemize}
	\item \textit{Conflict (or \confc{}) edge}: The conflict edges between two transactions depends on the conflicts between them. Two transactions $T_i$ and $T_j$ of the sequential history are said to be in conflict, if one of the following holds:
	\begin{itemize}
		\item \textbf{tryC-tryC} conflict: Two transactions $T_i$ \& $T_j$ are in tryC-tryC conflict (1) If $T_i$ and $T_j$ are committed; (2) Both $T_i$ \& $T_j$ update the same key $k$ of the \tab{}, $ht$, i.e., $(\langle ht,k \rangle \in \udset{T_i}) \land (\langle ht,k \rangle \in \udset{T_j})$, here $\udset{T_i}$ is set of keys in which $T_i$ performs update methods (or $\upmt{s}$); (3) and \tryc{()} of $T_i$ has completed before \tryc{()} of $T_j$, i.e., $\tryc_i() \prec_{H}^{\mr} \tryc_j()$. 
		
		\item \textbf{tryC-rv} conflict: Two transactions $T_i$ \& $T_j$ are in tryC-rv conflict (1) If $T_i$ has updated the key $k$ of \tab{}, $ht$ and committed; (2) After that $T_j$ invokes a \rvmt{} $rvm_{j}$ on the key same $k$ of \tab{} $ht$ and returns the value updated by $T_i$, i.e., $\tryc_i \prec_{H}^{\mr} \rvm_{j}$.
		
		\item \textbf{rv-tryC} conflict: Two transactions $T_i$ \& $T_j$ are in rv-tryC conflict (1) $T_i$ invokes a \rvmt{} $rvm_{i}$ on the key $k$ of \tab{} $ht$ and returns the value updated by $T_k$, i.e., $\tryc_k \prec_{H}^{\mr} \rvm_{i}$; (2) After that $T_j$ update the same key $k$ of the \tab{}, $ht$, i.e., $(\langle ht,k \rangle \in \udset{T_j})$ and $T_j$ returns commit, i.e., $\rvm_{i} \prec_{H}^{\mr} \tryc_j$.
	\end{itemize}
	If any of the above defined conflicts occur then conflict edge goes from $T_i$ to $T_j$. As described in \secref{sm}, \emph{STM\_lookup()}, and \emph{STM\_delete()} return the value from underlying data structure so, we called these methods as \emph{return value methods (or $\rvmt{s}$)}. Whereas, \emph{STM\_insert()}, and \emph{STM\_delete()} are updating the underlying data structure after successful \emph{STM\_tryC()} so, we called these as \emph{update} \mth{s} (or $\upmt{s}$). So, the conflicts are defined between the \mth{s} that accesses the shared memory.  \emph{(STM\_tryC$_i$(), STM\_tryC$_j$()), (STM\_tryC$_i$(), STM\_lookup$_j$()), (STM\_lookup$_i$(), STM\_tryC$_j$()), (STM\_tryC$_i$(), STM\_delete$_j$()) and (STM\_delete$_i$(), STM\_tryC$_j$())} are the possible conflicting \mth{s}.  
	
	\item \textit{real-time (or \rt) edge}: If transaction $T_i$ returns commit before the beginning of other transaction $T_j$ then real-time edge goes from $T_i$ to $T_j$. Formally, ($\tryc_i() \prec_{H}$ \emph{STM\_begin$_j$()}) $\implies$  $T_i$ $\rightarrow$ $T_j$.
	
\end{itemize}

For better understanding, we consider a history $H$: $l_1(ht, k_5, nil), l_2(ht, k_7, nil),\\ d_1(ht, k_6, nil), C_1, i_2(ht, k_5, v_2), C_2, l_3(ht, k_5, v_2), i_3(ht, k_7, v_3), C_3$ and show the time line view of it in \figref{gco}.(a). We construct $\copg{H}{\ll} = (V, E)$ shown in \figref{gco}.(b). There exist a (rv-tryC) edge between $T_1$ to $T_2$ because $T_2$ updates the key $k_5$ with value $v_2$ after $T_1$ lookups it. $T_3$ begins after the commit of $T_1$ and $T_2$ so, \rt{} edges are going from $T_1$ to $T_3$ and $T_2$ to $T_3$. Here, $T_3$ lookups key $k_5$ after updated by $T_2$ and returns the value $v_2$. So, (tryC-rv) edge  is going from $T_2$ to $T_3$. Hence, $H$ constructs an acyclic graph (CG) with equivalent serial schedule $T_1 T_2 T_3$.

\begin{figure}
	\centerline{
		\scalebox{0.4}{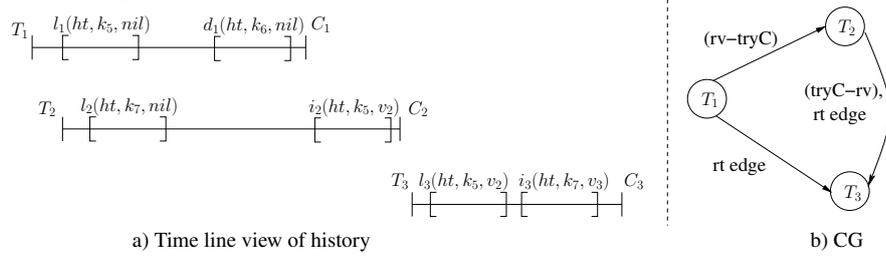}}
	\captionsetup{justification=centering}
	\caption{Illustration of Graph Characterization of Co-opacity}
	\label{fig:gco}
\end{figure}

\begin{lemma}
	\label{lem:concycle}
For any \legal{} \tseq{} history $S$ the conflict graph $\copg{S, \ll_S}$ is acyclic.
\end{lemma}

\begin{proof}
	\tseq{} history $S$ consists of multiple transactions, we order all the them into real-time order on the basis of their increasing order of timestamp (TS). For example, consider two transaction $T_i$ and $T_j$ with $TS(T_i)$ is less than $TS(T_j)$ then $T_i$ will occur before $T_j$ in $S$. Formally, $TS(T_i) < TS(T_j)$ $\Leftrightarrow T_i <_S T_j$.	 To proof the order between transactions, we analyze all the edges of $\copg{S, \ll_S}$ one by one: 
	\begin{itemize}
		\item \rt{} edges: It follow that any transaction begin after commit of previous transaction only. Hence, all the \rt{} edges go from a lower TS transaction $T_i$ to higher TS transaction $T_j$ and follow timestamp order. 
		
		\item \confc{} edges: If any transaction $T_j$ lookups key $k$ from $T_i$ in $S$ then $T_i$ has to be committed before invoking of lookup of $T_j$. Similarly, other conflicting edges are following TS order as $TS(T_i) < TS(T_j)$ $\Leftrightarrow T_i <_S T_j$. Thus, all the \confc{} edges go from a lower TS transaction to higher TS transaction.
		
		
	\end{itemize}
	Hence, all the edges of $\copg{S, \ll_S}$ are following increasing of the TS of the transactions, i.e. all the edges goes from lower TS transaction to higher TS transaction in $S$. Conflict graph $\copg{S, \ll_S}$ is acyclic.
\end{proof}

\begin{theorem}
	\label{thm:cocg}
	A history H is co-opaque iff $\copg{H}{\ll_H}$ is acyclic.
\end{theorem}

\begin{proof}
	\textbf{(if part):} First, we consider $\copg{H}{\ll_H}$ is acyclic and we need to prove that history $H$ is co-opaque. Since $\copg{H}{\ll_H}$ is acyclic, we apply topological sort on $\copg{H}{\ll_H}$ and generate a \tseq{} history $S$ such that $S$ is equivalent to $\overline{H}$.
	$\copg{H}{\ll_H}$ maintains real-time edges as well and $S$ has been generated from it. So, $S$ also respect real-time order \rt{} as $H$. Formally, $\prec_H^{RT} \subseteq \prec_S^{RT}$.
	
	Since $\copg{H}{\ll_H}$ maintains all the conflicting (or \confc{}) edges as well defined above. $S$ has been generated by applying topological sort on $\copg{H}{\ll_H}$. So, $S$ respects all the conflicting edges present in $H$. Formally, $\prec_H^{Conf} \subseteq \prec_S^{Conf}$. 
	
	It can be seen in $\copg{H}{\ll_H}$ that \rvmt{s()} on any key $k$ by transaction $T_i$ returns the value written on $k$ by previous closest committed transaction $T_j$. $H$ maintains all the \rvmt{s()} in conflicting (or \confc{}) edges of $\copg{H}{\ll_H}$. Since $S$ has been generated by applying topological sort on $\copg{H}{\ll_H}$. So, $S$ returns all the value of the \rvmt{s()} from previous closest committed transactions. Hence, $S$ is $legal$.
	
	$S$ satisfies all the properties of \emph{co-opacity} and equivalent to $H$ because $S$ has been generated from the topological sort on $\copg{H}{\ll_H}$. Hence, history H is co-opaque.

	\textbf{(Only if part):} Now, we consider $H$ is co-opaque and we have to prove that $\copg{H}{\ll_H}$ is acyclic.
	Since $H$ is co-opaque there exists an equivalent  \legal{} \tseq{} history $S$ to $\overline{H}$ which maintains real-time (\rt{}) and conflict ($Conf$) order of $H$. From the \lemref{concycle}, we can say that conflict graph $\copg{S, \ll_S}$ is acyclic. As we know, $\copg{H}{\ll_H}$ is the subgraph of $\copg{S, \ll_S}$.
	Hence, $\copg{H}{\ll_H}$ is acyclic.

	
	\cmnt{
	\begin{itemize}
		\item \rt{} edges: We have that $S$ respects real-time order of $H$, i.e $\prec_{H}^{RT} \subseteq \prec_{S}^{RT}$. Hence, all the \rt{} edges of $H$ are a subset of $S$. 
		
		\item \rvf{} edges: Since $\overline{H}$ and $S$ are equivalent, the return value-from relation of $\overline{H}$ and $S$ are the same. Hence, the \rvf{} edges are the same in $G_H$ and $G_S$. 
		
		\item \mv{} edges: Since the version-order and the \op{s} of the $H$ and $S$ are the same, from \lemref{eqv_hist_mvorder} it can be seen that $\overline{H}$ and $S$ have the same \mv{} edges as well.
	\end{itemize}
	
	Thus, the graph $G_H$ is a subgraph of $G_S$. Since we already know that $G_S$ is acyclic from \lemref{seracycle}, we get that $G_H$ is also acyclic. 

}
\end{proof}
\cmnt{
To prove that a STM system satisfies opacity, it is useful to consider graph characterization of histories. In this section, we describe the graph characterization of Guerraoui and Kapalka \cite{tm-book} modified for sequential histories.

Consider a history $H$ which consists of multiple version for each \tobj. The graph characterization uses the notion of \textit{version order}. Given $H$ and a \tobj{} $k$, we define a version order for $k$ as any (non-reflexive) total order on all the versions of $k$ ever created by committed transactions in $H$. It must be noted that the version order may or may not be the same as the actual order in which the version of $k$ are generated in $H$. A version order of $H$, denoted as $\ll_H$ is the union of the version orders of all the \tobj{s} in $H$. 

Consider the history $H3$ as shown in \figref{mvostm3} $: lu_1(ht, k_{x, 0}, null), lu_2(ht, k_{x, 0}, null),\\ lu_1(ht, k_{y, 0}, null), lu_3(ht, k_{z, 0}, null), ins_1(ht, k_{x, 1}, v_{11}), ins_3(ht,
k_{y, 3}, v_{31}), ins_2(ht, \\k_{y, 2}, v_{21}), ins_1(ht, k_{z, 1}, v_{12}), c_1, c_2, lu_4(ht, k_{x, 1}, v_{11}), lu_4(ht, k_{y, 2}, v_{21}), ins_3(ht, k_{z, 3},\\ v_{32}), c_3, lu_4(ht, k_{z, 1}, v_{12}), lu_5(ht, k_{x, 1}, v_{11}), lu_6(ht, k_{y, 2}, v_{21}), c_4, c_5, c_6$. Using the notation that a committed transaction $T_i$ writing to $k_x$ creates a version $k_{x, i}$, a possible version order for $H3$ $\ll_{H3}$ is: $\langle k_{x, 0} \ll k_{x, 1} \rangle, \langle k_{y, 0} \ll k_{y, 2} \ll k_{y, 3} \rangle, \langle k_{z, 0} \ll k_{z, 1} \ll k_{z, 3} \rangle $.}
\cmnt{
Consider the history $H4: l_1(ht, k_{x, 0}, NULL) l_2(ht, k_{x, 0}, NULL) l_1(ht, k_{y, 0}, NULL) l_3(ht, k_{z, 0},\\ NULL) i_1(ht, k_{x, 1}, v_{11}) i_3(ht, k_{y, 3}, v_{31}) i_2(ht, k_{y, 2}, v_{21}) i_1(ht, k_{z, 1}, v_{12}) c_1 c_2 l_4(ht, k_{x, 1}, v_{11}) l_4(ht,\\ k_{y, 2}, v_{21}) i_3(ht, k_{z, 3}, v_{32}) c_3 l_4(ht, k_{z, 1}, v_{12}) l_5(ht, k_{x, 1}, v_{11}), l_6(ht, k_{y, 2}, v_{21}) c_4, c_5, c_6$. In our representation, we abbreviate \tins{} as $i$, \tdel{} as $d$ and \tlook{} as $l$. Using the notation that a committed transaction $T_i$ writing to $k_x$ creates a version $k_{x, i}$, a possible version order for $H4$ $\ll_{H4}$ is: $\langle k_{x, 0} \ll k_{x, 1} \rangle, \langle k_{y, 0} \ll k_{y, 2} \ll k_{y, 3} \rangle, \langle k_{z, 0} \ll k_{z, 1} \ll k_{z, 3} \rangle $. 
}
\cmnt{
\begin{figure}
	\centering
	\captionsetup{justification=centering}
	\centerline{\scalebox{0.45}{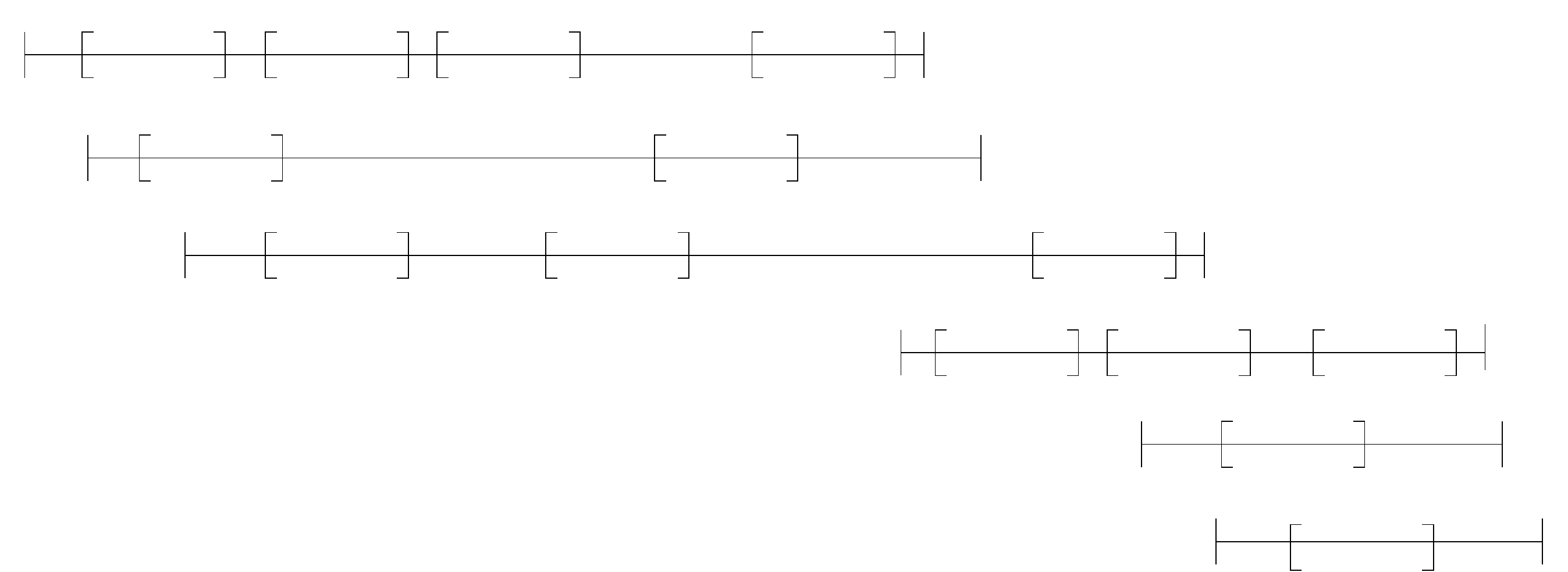}}
	\caption{History $H3$ in time line view}
	\label{fig:mvostm3}
\end{figure}
}
\cmnt{
\begin{figure}[tbph]
	\centerline{\scalebox{0.45}{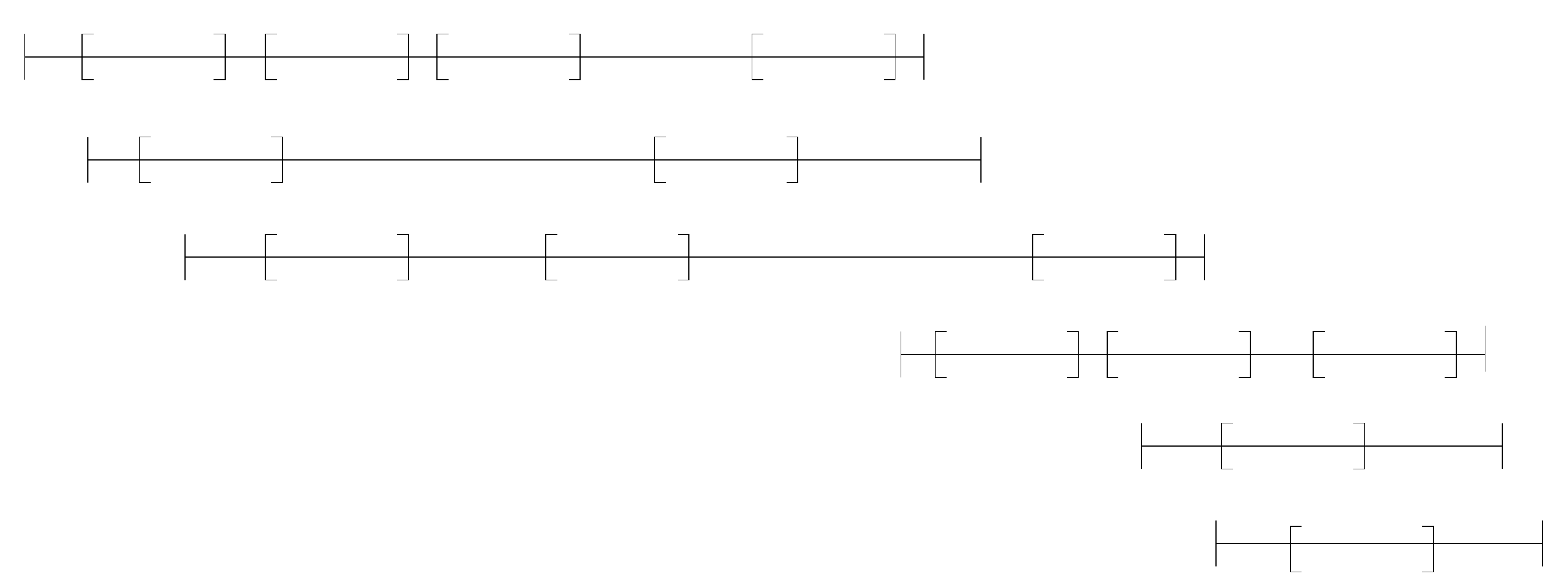}}
	\caption{History $H4$ in time line view}
	\label{fig:mvostm3}
\end{figure}
}

\begin{theorem}
	\label{thm:coopqhprev}
	A legal SF-SVOSTM history $H$ is co-opaque iff $\copg{H}{\ll_H}$ is acyclic.
\end{theorem}
\begin{proof}
	\textbf{(if part):} First, we consider $H$ is \emph{legal} and $\copg{H}{\ll_H}$ is acyclic then we need to prove that history $H$ is co-opaque. Since $\copg{H}{\ll_H}$ is acyclic, we apply topological sort on $\copg{H}{\ll_H}$ and obtained a \tseq{} history $S$ which is equivalent to $\overline{H}$.
	$S$ also respect real-time edges (or \rt{}) and \confc{} edges as $\overline{H}$. Formally, $S$ respects $\prec$$^{RT}_H$ = $\prec$$^{RT}_{\overline{H}}$ and $\prec$$^{Conf}_H$ = $\prec$$^{Conf}_{\overline{H}}$. 
	
	Since Conflict relation between two methods of SF-SVOSTM in $S$ are also present in $\overline{H}$. Formally, $\prec_{\overline{H}}^{Conf} \subseteq \prec_S^{Conf}$. Given that $H$ is \emph{legal} which implies that  $\overline{H}$ is also \emph{legal}. So, we can say that $S$ is \emph{legal}. Collectively, $H$ satisfies all the necessary conditions of \emph{co-opacity}. Hence, history $H$ is co-opaque.
	
	\textbf{(Only if part):} Now, we consider $H$ is co-opaque and legal then we have to prove that $\copg{H}{\ll_H}$ is acyclic.
	Since $H$ is co-opaque there exists an equivalent  \legal{} \tseq{} history $S$ to $\overline{H}$ which maintains real-time (\rt{}) and conflict ($Conf$) order of $H$, i.e, S respects $\prec$$^{RT}_H$ and $\prec$$^{Conf}_H$ ((from the definition of co-opacity\cite{KuznetsovPeri:Non-interference:TCS:2017})). So, we can observe from the conflict graph construction that $\copg{H}{\ll_H}$ = $\copg{H}{\ll_{\overline{H}}}$ and both are the subgraph of $\copg{S}{\ll_S}$. Since $S$ is a  \tseq{} history, so $\copg{S}{\ll_S}$ is acyclic. As we know, any subgraph of any acyclic graph is also acyclic and $\copg{H}{\ll_H}$ is the subgraph of $\copg{S}{\ll_S}$. 
	Hence, $\copg{H}{\ll_H}$ is acyclic.

	
\end{proof}

\begin{theorem}
	\label{thm:coopq}
	Any legal history $H$ generated by SF-SVOSTM satisfies co-opacity.
\end{theorem}

\begin{proof}
In order to prove this, we construct the co-opacity graph $\copg{H}{\ll}$ generated by SF-SVOSTM algorithm and prove that $\copg{H}{\ll}$ graph is acyclic. After that with the help of \thmref{coopqhprev}, we can say that generated $\copg{H}{\ll}$ graph is acyclic so any legal history $H$ generated by SF-SVOSTM is co-opaque.

To prove the $\copg{H}{\ll}$ generated by SF-SVOSTM algorithm is acyclic. We construct $\copg{H}{\ll} = (V, E)$ which consists of $V$ vertices and $E$ edges. Here, each committed transaction is consider as a vertex and edges are as follows:
\begin{itemize}
	\item \textit{real-time (or \rt) edge}: If transaction $T_i$ returns commit before the beginning of other transaction $T_j$ then real-time edge goes from $T_i$ to $T_j$ in SF-SVOSTM. Formally, ($\tryc_i() \prec_{H}$ \emph{STM\_begin$_j$}()) $\implies$ $TS(T_i) < TS(T_j)$ $\implies$ $T_i$ $\rightarrow$ $T_j$.
	
	\item \textit{conflict (or \confc{}) edge}: The conflict edges between two transactions depends on the conflicts between them. Two transactions $T_i$ and $T_j$ of the sequential history are said to be in conflict, if both of them  access same key $k$ and at least one transaction performs \emph{update} method. As described in \secref{sm}, \emph{STM\_lookup()}, and \emph{STM\_delete()} return the value from underlying data structure so, we called these methods as \emph{return value methods (or $\rvmt{s}$)}. Whereas, \emph{STM\_insert()}, and \emph{STM\_delete()} are updating the underlying data structure after successful \emph{STM\_tryC()} so, we called these as \emph{update} \mth{s} (or $\upmt{s}$). So, the conflicts are defined between the \mth{s} that accesses the shared memory.  \emph{(STM\_tryC$_i$(), STM\_tryC$_j$()), (STM\_tryC$_i$(), STM\_lookup$_j$()), (STM\_lookup$_i$(), STM\_tryC$_j$()), (STM\_tryC$_i$(), STM\_delete$_j$()) and (STM\_delete$_i$(), STM\_tryC$_j$())} are the possible conflicting \mth{s} in SF-SVOSTM. On conflict between two transaction $T_i$ and $T_j$ where $TS(T_i) < TS(T_j)$, the conflict edge going from $T_i$ to $T_j$. In SF-SVOSTM, if higher timestamp (TS) transaction $T_j$ has already been committed then lower TS transaction $T_i$ returns abort and retry with higher TS in the incarnation of $T_i$ and returns commit. So, conflicts edges in SF-SVOSTM follows increasing of their TS. Formally, $TS(T_i) < TS(T_j)$ $\implies$ $T_i$ $\rightarrow$ $T_j$.
	
\end{itemize}
So, all the edges of $\copg{H}{\ll}$ generated by SF-SVOSTM algorithm follows the increasing order of  TS of the transactions. Thus, $\copg{H}{\ll}$ graph is acyclic. Hence, with the help of \thmref{coopqhprev}, any legal history $H$ generated by SF-SVOSTM satisfies co-opacity.  

\end{proof}

\section{Graph Characterization of Local Opacity and SF-KOSTM Correctness}
\label{sec:gcoflo}

This section describes the graph characterization of local opacity for the history $H$ which maintains multiple versions corresponding to each key. Graph Characterization helps to prove the correctness of STMs for a given version order. Lets assume a history $H$ with given version order $\ll$. Following the  graph characterization by Chaudhary et al. \cite{Chaudhary+:KSFTM:Corr:2017} and modified it for sequential histories with high-level methods while maintaining multiple versions corresponding to each key and extend it for local opacity which helps to prove the correctness of SF-KOSTM.

Similar to \secref{gcofo}, SF-KOSTM executes a concurrent history $H$ which consists of multiple transactions. Each transaction calls high-level methods which internally invokes multiple read-write (or lower-level) operations including method invocation and response known as \emph{events} (or $evts$) as discussed in \secref{sm}. So, we need to ensure the atomicity on both the levels. First, we ensure method level atomicity using the \emph{linearization point (LP)} of respective method. After that we ensure the atomicity of transaction on the basis of graph characterization of local-opacity. 

High-level methods are interval instead of dots (atomic). In order to make it atomic, we order the high-level method on the basis of their \emph{linearization point (LP)}. We consider \emph{first unlocking point} of each successful method as the \emph{LP} of the respective method.
Now, we need to ensure the atomicity at transactional level, a transaction internally invokes multiple high-level methods and transactions are overlapping to each other in concurrent history $H$. So, with the help of \emph{graph characterization} of \emph{local-opacity}, SF-KOSTM ensures the atomicity of the transaction defined below.

We construct a opacity graph represented as $\lopg{H}{\ll} = (V, E)$ which consists of $V$ vertices and $E$ edges. Here, each committed transaction $T_i$ is consider as a vertex and edges are as follows:
\begin{itemize}
	\item \textit{real-time (or \rt) edge}: This is same as \rt{} edge defined in $\copg{H}{\ll_H}$. If transaction $T_i$ returns commit before the beginning of other transaction $T_j$ then real-time edge goes from $T_i$ to $T_j$. Formally, ($\tryc_i() \prec_{H}$ \emph{STM\_begin$_j$}()) $\implies$  $T_i$ $\rightarrow$ $T_j$.
	\item \textit{return value from (or \rvf) edge}: There exist a \emph{\rvf} edge between two transaction $T_i$ and $T_j$ such that (1) If $T_i$ is the latest transaction that has updated ($\upmt{()}$) the key $k$ of \tab{}, $ht$ and committed; (2) After that $T_j$ invokes a \rvmt{} $rvm_{j}$ on the same key $k$ of \tab{} $ht$ and returns the value updated by $T_i$, i.e., $\upmt{_i()} \prec_{H}^{\mr} \rvm_{j}$. As defined in \secref{sm}, \upmt{()} can either be \emph{STM\_insert()} or \emph{STM\_delete()}. If the $\upmt{_i()}$ is \emph{STM\_insert()} method on key $k$ then \rvmt{()} returns the value updated by $T_i$, i.e.,  $i_i(k, v) <_H c_i <_H \rvm_j(k, v)$. If the $\upmt{_i()}$ is \emph{STM\_delete()} method on key $k$ then \rvmt returns $null$, i.e.,  $d_i(k, null) <_H c_i <_H \rvm_j(k, null)$.
	\item \textit{multi-version (or \mv) edge}: It depends on the version order between two transactions $T_i$ and $T_j$. For the sake of understanding, consider a triplet of three transactions $T_i$, $T_j$ and $T_k$ with successful methods on key $k$ as  $\up_i(k,u)$, $\rvm_j(k,u)$, $\up_k(k,v)$ , where $u \neq v$ and $up_i$ stands for \upmt{$_i$()} of $T_i$. It can observe that a \emph{return value from edge} is going from $T_i$ to $T_j$ because of $\rvm_j(k,u)$. If the version order is $k_i \ll k_k$ then the multi-version edge is going from $T_j$ to $T_k$. Otherwise, multi-version edge is from $T_k$ to $T_i$ because of version order ($k_k \ll k_i$). 
\end{itemize}

\noindent
For better understanding, we consider a history $H$: $l_1(ht, k_5, nil), l_2(ht, k_7, nil),\\ d_1(ht, k_6, nil), C_1, i_2(ht, k_5, v_2), C_2, l_3(ht, k_5, v_2), i_3(ht, k_7, v_3), C_3$ and show the time line view of it in \figref{glo}.(a). We construct $\lopg{H}{\ll} = (V, E)$ shown in \figref{glo}.(b). There exist a \mv{} edge between $T_1$ and $T_2$ because $T_1$ lookups the value of key $k_5$ from $T_0$ and after that $T_2$ creates a version on $k_5$ with value $v_2$. $T_3$ begins after the commit of $T_1$ and $T_2$ so, \rt{} edges are going from $T_1$ to $T_3$ and $T_2$ to $T_3$. Here, $T_3$ lookups key $k_5$ from the version created by $T_2$ and returns the value $v_2$. So, \rvf{} edge  is going from $T_2$ to $T_3$. Hence, $H$ constructs an acyclic graph (OG) with equivalent serial schedule $T_1 T_2 T_3$.



\cmnt{\begin{figure}[H]
	\centerline{\scalebox{0.7}{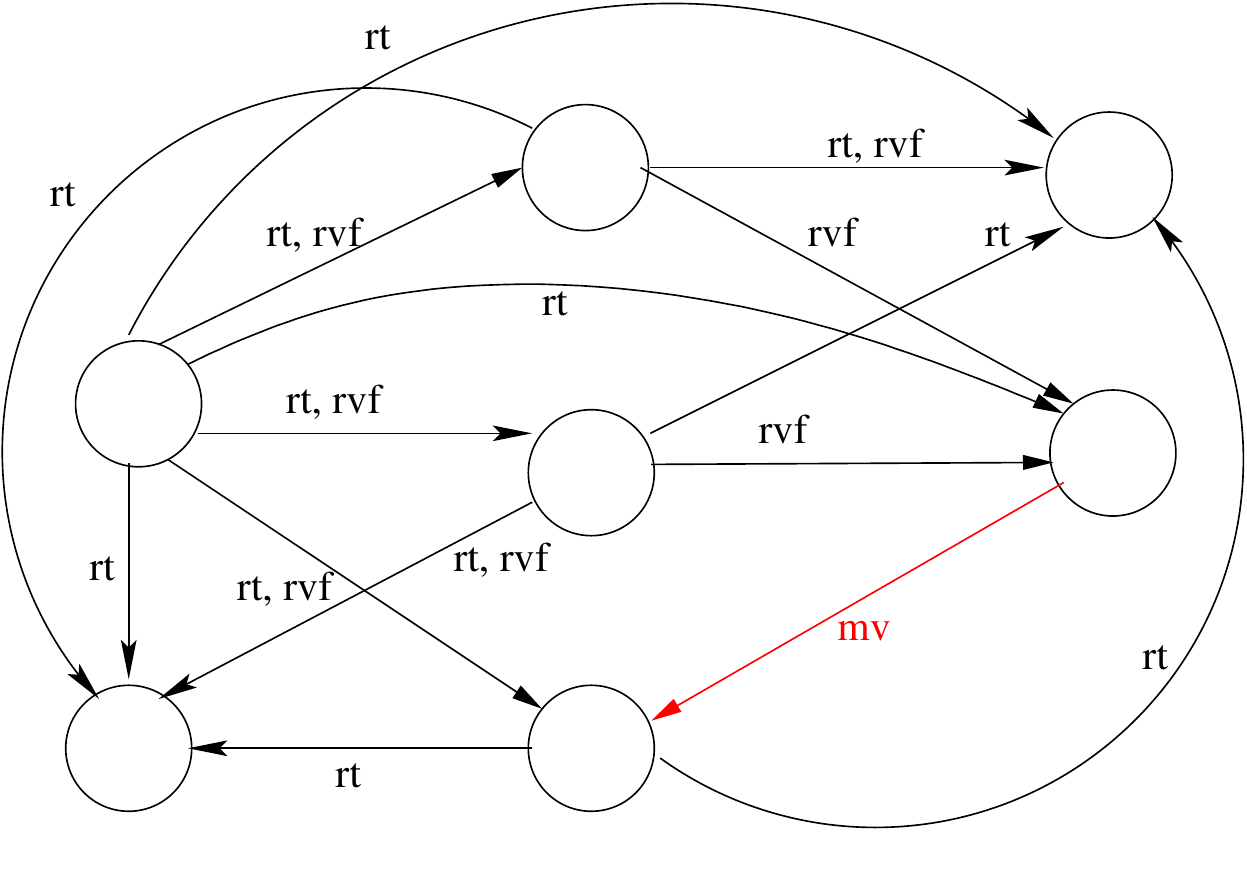}}
	\caption{$\opg{H3}{\ll_{H3}}$}
	\label{fig:mvostm1}
\end{figure}}

\begin{figure}
	\centerline{
		\scalebox{0.4}{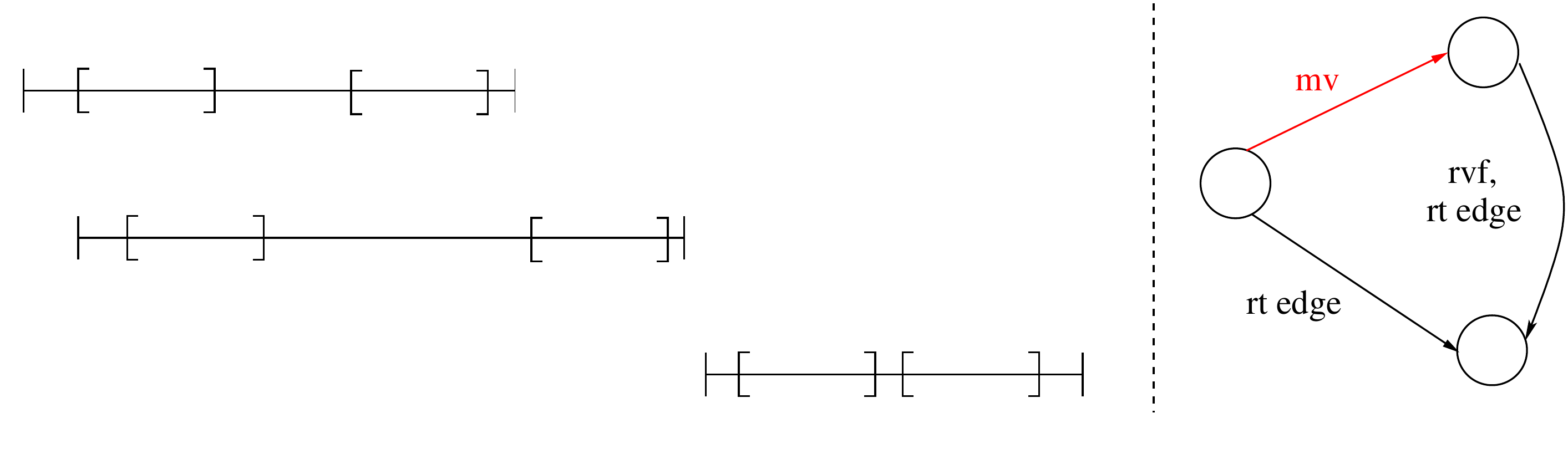}}
	\captionsetup{justification=centering}
	\caption{Illustration of Graph Characterization of Opacity}
	\label{fig:glo}
\end{figure}

For a given history $H$ and version order $\ll$, we consider a complete graph $\overline{H}$ instead of $H$ and construct the graph $\lopg{\overline{H}}{\ll}$. It can be seen that $\overline{H}$ has more $\rt{}$ edges than $H$, i.e., $\prec_H^{RT} \subseteq \prec_{\overline{H}}^{RT}$. But, for the graph construction, we consider only $\rt{}$ edges of $H$ with the assumption $\rt(H) = \rt(\overline{H})$ that satisfies the following property:

\begin{property}
	\label{prop:hoverh}
	For a given history $H$ and version order $\ll$, the opacity graphs $\lopg{H}{\ll}$ and $\lopg{\overline{H}}{\ll}$ are same.
\end{property}
\begin{definition}
	\label{def:seqver}
	We define a version order $\ll_S$ for \tseq{} history $S$ such that if two committed transactions $T_i$ and $T_j$ has created versions on key $k$ as $k_i$ and $k_j$ respectively with version order $k_i \ll_S k_j$ then $T_i$ committed before $T_j$ in $S$. Formally, $\langle k_i \ll_S k_j \Leftrightarrow T_i <_S T_j \rangle $.

\end{definition}
\noindent
This definition along with below defined lemmas and theorems will help us to prove the correctness of our graph characterization.  

\begin{lemma}
	\label{lem:seracycle1}
	The opacity graph for \legal{} \tseq{} history $S$ as $\lopg{S, \ll_S}$ is acyclic.
\end{lemma}

\begin{proof}
	The proof of this lemma is similar as \lemref{concycle} of \secref{gcofo}. We order all the transactions of $S$ into real-time order on the basis of their increasing order of timestamp (TS). For example, consider two transaction $T_i$ and $T_j$ with $TS(T_i)$ is less than $TS(T_j)$ then $T_i$ will occur before $T_j$ in $S$. Formally, $TS(T_i) < TS(T_j)$ $\Leftrightarrow T_i <_S T_j$. We consider all the types edges of $\lopg{S, \ll_S}$ and analyze it one by one as follows to show the acyclicity of it:
	\begin{itemize}
		\item \rt{} edge: It follow that any transaction begin after commit of previous transaction only. Hence, all the \rt{} edges go from a lower TS transaction $T_i$ to higher TS transaction $T_j$ and follow timestamp order. 
		
		\item \rvf{} edge: Any transaction $T_j$ lookups key $k$ from $T_i$ in $S$ then $T_i$ has to be committed before invoking of lookup of $T_j$. So, $TS{(T_i)} < TS{(T_j)}$. Hence, all the \rvf{} edges goes from a lower TS transaction to a higher TS transaction.
		
		\item \mv{} edge: Consider a triplet of three transactions $T_i$, $T_j$ and $T_k$ with successful methods on key $k$ as  $\up_i(k,u)$, $\rvm_j(k,u)$, $\up_k(k,v)$ , where $u \neq v$. Here, $\rvm_j(k,u)$ method is returning the latest value written by $T_i$ on key $k$ with value $u$ using $\up_i(k,u)$. So, there exist a \rvf{} edge with $TS(T_i) < TS(T_j)$. There are two cases for the version order of $k$ as follows: (1) If the version order is $T_k \ll_S T_i$ which implies that $TS(T_k) < TS(T_i)$ then multi-version edge goes from $T_k$ to $T_i$ which also follows the increasing order of $TS$. (2) If the version order is $T_i \ll_S T_k$ which implies that $TS(T_i) < TS(T_k)$. Since $S$ is a \legal{} \tseq{} history, so, $TS(T_j) < TS(T_k)$ and then multi-version edge goes from $T_j$ to $T_k$ which again follows the increasing order of $TS$. So, \mv{} edges also follow the increasing order of $TS$ order.

		
	\end{itemize}
	
	Therefore, all the types of edges follow the increasing order of transaction's $TS$ as defined above. All the edges of $S$ goes from lower $TS$ transaction to higher $TS$ transactions. This implies that the opacity graph generated by \legal{} \tseq{} history $S$ as $\lopg{S, \ll_S}$ is acyclic. 
\end{proof}

\begin{lemma}
	\label{lem:eqv_hist_mvorder}
	Consider a history $H$ with given version order $\ll_H$. Another history $H'$ is equivalent to $H$ then the mv edges $\mv(H, \ll_H)$ induced by $\ll_H$ in $H$ and $H'$ will be same.	
\end{lemma}

\begin{proof}
	Since history $H$ and $H'$ are equivalent, so, version order of  $\ll_H$ will be same as version order of  $\ll_{H'}$. We can observe that \mv{} edges depend on version order $\ll$ and the methods of the history. It is independent from the order of the methods in $H$. So, being equivalence $H'$ also contains the same version order $\ll_H$ and the methods as in $H$. Thus, multi-version \mv{} edges are same in $H$ and $H'$. 
\end{proof}


\begin{theorem}
	\label{thm:opg}
	A \valid{} history $H$ is opaque with a version order $\ll_H$ iff $\lopg{H}{\ll_H}$ is acyclic.
\end{theorem}

\begin{proof}
	\textbf{(if part):} First, we consider $\lopg{H}{\ll_H}$ is acyclic and we need to prove that history $H$ is opaque. Since $\lopg{H}{\ll_H}$ is acyclic, we apply topological sort on $\lopg{H}{\ll_H}$ and generate a \tseq{} history $S$ such that $S$ is equivalent to $\overline{H}$.
	$\lopg{H}{\ll_H}$ maintains real-time edges as well and $S$ has been generated from it. So, $S$ also respect real-time order (or \rt{}) as $H$. Formally, $\prec_H^{RT} \subseteq \prec_S^{RT}$.
	
	$\lopg{H}{\ll_H}$ maintains the return value from (or \rvf{}) edges for \rvmt{()} on any key $k$ by transaction $T_i$ returns the value written on $k$ by previously committed transaction $T_j$. So, $S$ is \emph{valid}. Now, we need to prove that $S$ is \emph{legal}. We prove it by contradiction, so we assume that $S$ is not \emph{legal}. That means, a \rvmt{()} $\rvm_j(k, u)$ lookups on key $k$ from a committed transaction $T_i$ which wrote the value of $k$ as $u$. But between these two transaction $T_i$ and $T_j$, an another committed transaction $T_k$ exist in $S$ which wrote to $k$ with value $v$ and ($u \neq v$). Formally, $T_i \prec_{S}^{RT} T_k \prec_{S}^{RT} T_j$. Consider a given version order in $H$ as $\ll_H$, if the version order is $T_k \ll_S T_i$ then the multi-version edge goes from $T_k$ to $T_i$. Consider the other version order is $T_i \ll_S T_k$ then the multi-version edge goes from $T_j$ to $T_k$. So, in both the cases $T_k$ is not coming between $T_i$ and $T_j$. So, our assumption is wrong. Hence,  $S$ is \emph{legal}.



	
	$S$ satisfies all the properties of \emph{opacity} and equivalent to $H$ because $S$ has been generated from the topological sort on $\lopg{H}{\ll_H}$. Hence, history H is opaque.

	\textbf{(Only if part):} Now, we consider $H$ is opaque and have to prove that $\lopg{H}{\ll_H}$ is acyclic.
	Since $H$ is opaque there exists an equivalent  \legal{} \tseq{} history $S$ to $\overline{H}$ which maintains real-time (\rt{}) and conflict ($Conf$) order of $H$. From the \lemref{seracycle1}, we can say that opacity graph $\lopg{S, \ll_S}$ is acyclic. As we know, $\lopg{H}{\ll_H}$ is the subgraph of $\lopg{S, \ll_S}$.
	Hence, $\lopg{H}{\ll_H}$ is acyclic.
	
\end{proof}

\noindent The above defined lemmas and theorems of \emph{opacity} can be extended the proof of \emph{local-opacity}.
\begin{theorem}
	\label{thm:log}
	A \valid{} history $H$ is \lopq iff all the sub-histories $\shset{H}$ (defined in \secref{sm}) for a history $H$ are opaque, i.e., A \valid{} history $H$ is \lopq iff the opacity graph $\lopg{sh}{\ll_{sh}}$ generated for each sub-history $sh$ of $\shset{H}$ with given version order $\ll_{sh}$ is acyclic. Formally,$\langle (H \text{ is \lopq}) \Leftrightarrow (\forall sh \in \shset{H}, \exists \ll_{sh}: \lopg{sh}{\ll_{sh}} \text{ is acyclic}) \rangle$.

\end{theorem}

\begin{proof}
	
	In order to prove it, we need to show that each sub-history $sh$ from the sub-histories of $H$, $\shset{H}$ is \emph{valid}. After that the remaining proof follows \thmref{opg}.
	
	To prove sub-history $sh$ is \emph{valid}, consider a $sh$ with any \rvmt{()} $\rvm_j(k,u)$ of a transaction $T_j$. We can easily observe that \rvm method returns the value of key $k$ that has been written by a committed transaction $T_i$ on $k$ with value $u$. So, it can be seen that $sh$ has all the transactions that has committed before $\rvm_j(k,u)$. Similarly, all the \rvm methods of $sh$ return the value from previously committed transactions. Thus, sub-history $sh$ is \emph{valid}.
	
	Now, we need to show each sub-history $sh$ is opaque with a version order $\ll_{sh}$ iff $\lopg{sh}{\ll_{sh}}$ is acyclic. The proof of this is directly coming \thmref{opg} while replacing $H$ with $sh$.
	Similarly, all the sub-histories of $\shset{H}$ are \emph{valid} and satisfying  \thmref{opg}. So, all the sub-histories of history $H$ are \emph{opaque}. Hence, a \valid{} history $H$ is \lopq.
	
	
\end{proof}

\noindent
Now, we consider the algorithms defined for each method of \emph{SF-KOSTM} in \secref{pm} and prove the correctness of SF-KOSTM. Consider a history $H$ generated by SF-KOSTM with two transaction $T_i$ and $T_j$ with status either \emph{live or committed} using $status$ flags as true. Then the edges between $T_i$ and $T_j$ follow the $\tltl$ order. SF-KOSTM algorithm ensures that $\tltl$ are keep on increasing order of the transaction timestamp (TS) order using atomic counter $\gtcnt$ in \emph{STM\_begin()}. Though the value of $\tltl$ is increasing in the other methods of SF-KOSTM still its maintaining the increasing order of transactions TS. We assume all the histories generate (or gen) by SF-KOSTM algorithm as $\gen{\emph{SF-KOSTM}}.$

\begin{lemma}
	\label{lem:tltl-edge}
	Any history $H$ generated by SF-KOSTM as $\gen{\emph{SF-KOSTM}}$ with two transactions $T_i$ and 
	$T_j$ such that status flags of both the transactions are true.  If there is an edge from $T_i$ to $T_j$ then $\tltl_i$ is less than $\tltl_j$. Formally, $T_i$ $\rightarrow$ $T_j$ $\implies$ $\tltl_i$ $<$ $\tltl_j$.		
\end{lemma}

\begin{proof}
	We consider all types of edges in $\lopg{H}{\ll_H}$ and analyze it as follows:
	
	\begin{itemize}
		\item \textit{real-time (or \rt) edge}: Here, the transaction $T_i$ returns commit before the begin of other transaction $T_j$ then real-time edge goes from $T_i$ to $T_j$. Hence, $\tltl_i$ gets the value from $\gtcnt$ earlier than begin of $T_j$. So, $T_i$ $\rightarrow$ $T_j$ $\implies$ $\tltl_i$ $<$ $\tltl_j$.
		
		\item \textit{return value from (or \rvf) edge}: The transaction $T_i$ has updated ($\upmt{()}$) the key $k$ of \tab{}, $ht$ and committed. After that transaction $T_j$ invokes a $rvm_{j}$ on the same key $k$ and returns the value updated by $T_i$. SF-KOSTM ensures that $T_j$ returns the value of $k$ from the transaction which has lesser TS than $T_j$ i.e., $T_i$ $\rightarrow$ $T_j$ $\implies$ $\tltl_i$ $<$ $\tltl_j$. 
		\item \textit{multi-version (or \mv) edge}: SF-KOSTM ensures that version order between two transactions $T_i$ and $T_j$ are also following TS order using their $\tltl$. Consider a triplet generated by SF-KOSTM with three transactions $T_i$, $T_j$ and $T_k$ with successful methods on key $k$ as  $\up_i(k,u)$, $\rvm_j(k,u)$, $\up_k(k,v)$ , where $u \neq v$. It can observe that a return value from edge is going from $T_i$ to $T_j$ because of $\rvm_j(k,u)$, so,  $T_i$ $\rightarrow$ $T_j$ $\implies$ $\tltl_i$ $<$ $\tltl_j$. If the version order is $k_i \ll k_k$ then the multi-version edge is going from $T_j$ to $T_k$. Hence, the order among the transactions are  ($T_i$ $\rightarrow$ $T_j$ $\rightarrow$ $T_k$) $\implies$ ($\tltl_i$ $<$ $\tltl_j$ $<$ $\tltl_k$). Otherwise, multi-version edge is from $T_k$ to $T_i$ because of version order ($k_k \ll k_i$). Then the relation is ($T_k$ $\rightarrow$ $T_i$ $\rightarrow$ $T_j$) $\implies$ ($\tltl_k$ $<$ $\tltl_i$ $<$ $\tltl_j$). 
	\end{itemize}
	So, all the edges of $\lopg{H}{\ll_H}$  are following increasing order of transactions $\tltl$. Hence, $T_i$ $\rightarrow$ $T_j$ $\implies$ $\tltl_i$ $<$ $\tltl_j$.
		
\end{proof}

\begin{theorem}
	\label{thm:coopqh}
	A valid SF-KOSTM history $H$ is locally-opaque iff $\lopg{H}{\ll_H}$ is acyclic.
\end{theorem}
\begin{proof}
	\textbf{(if part):} First, we consider $H$ is \emph{valid} and $\lopg{H}{\ll_H}$ generated by SF-KOSTM is acyclic then we need to prove that history $H$ is locally-opaque. Since $\lopg{H}{\ll_H}$ is acyclic, we apply topological sort on $\lopg{H}{\ll_H}$ and obtained a \tseq{} history $S$ which is equivalent to $\overline{H}$.
	$S$ also respect real-time (or $RT$), return value from (or $RVF$) and multi-version (or $MV$)  edges as $\overline{H}$. Formally, $S$ respects $\prec$$^{RT}_S$ = $\prec$$^{RT}_{\overline{H}}$, $\prec$$^{RVF}_S$ = $\prec$$^{RVF}_{\overline{H}}$ and $\prec$$^{MV}_S$ = $\prec$$^{MV}_{\overline{H}}$. 
	
	Since \rvf{} and \rt{} relation between the methods of SF-KOSTM for a given version in $S$ are also present in $\overline{H}$. Formally, $\prec_{\overline{H}}^{RT} \subseteq \prec_S^{RT}$ and $\prec_{\overline{H}}^{RVF} \subseteq \prec_S^{RVF}$. Given that $H$ is \emph{valid} which implies that  $\overline{H}$ is also \emph{valid}. So, we can say that $S$ is \emph{legal} for a given version order. Similarly, we can prove that all the sub-histories $\shset{H}$ for a history $H$ are opaque. Hence with the help of \thmref{log}, collectively, $H$ satisfies all the necessary conditions of \emph{local-opacity}. Hence, history $H$ is locally-opaque.
	
	\textbf{(Only if part):} Now, we consider $H$ is locally-opaque and valid then we have to prove that $\lopg{H}{\ll_H}$ is acyclic. We prove it through contradiction, so we assume there exist a cycle in $\lopg{H}{\ll_H}$. From \lemref{tltl-edge}, any two transactions $T_i$ and 
	$T_j$ generated by SF-KOSTM such that both their status flags are true and $T_i$ $\rightarrow$ $T_j$ $\implies$ $\tltl_i$ $<$ $\tltl_j$. Consider the transitive case with $k$ transactions
 $T_1, T_2, T_3...T_k$ such	that status flags of all the $k$ transactions are true. If edges exist like ($T_1$ $\rightarrow$ $T_2$ $\rightarrow$ $T_3$ 
	$\rightarrow$....$\rightarrow$ $T_k$) $\implies$ ($\tltl_1 $ $<$ $\tltl_2$ $<$ 
	$\tltl_3$ $<$ ....$<$ $\tltl_k$).\\
	Now, we consider our assumption, there exist a cycle in $\lopg{H}{\ll_H}$. So, $T_1$ 
	$\rightarrow$ $T_2$ $\rightarrow$ $T_3$ $\rightarrow$....$\rightarrow$ 
	$T_k$ $\rightarrow$ $T_1$ that implies $\tltl_1 $ $<$ $\tltl_2$ $<$ 
	$\tltl_3$ $<$ ....$<$ $\tltl_k$ $<$ $\tltl_1$.\\
	Hence, above assumption says that, $\tltl_1$ $<$ $\tltl_1$ but this is not possible. So, our assumption there exist a cycle in $\lopg{H}{\ll_H}$ is wrong.\\
	Therefore, $\lopg{H}{\ll_H}$ produced by SF-KOSTM is acyclic.
\end{proof}

\begin{theorem}
	Any valid history $H$ generated by SF-KOSTM satisfies local-opacity.
\end{theorem}

\begin{proof}
	With the help of \lemref{tltl-edge}, we can say that any history $H$  $\gen{\emph{SF-KOSTM}}$ with two transactions $T_i$ and 
	$T_j$ such that status flags of both the transactions are true.  If there is an edge from $T_i$ to $T_j$ then $\tltl_i$ is less than $\tltl_j$. Formally, $T_i$ $\rightarrow$ $T_j$ $\implies$ $\tltl_i$ $<$ $\tltl_j$. So, we can infer that any valid history $H$ generated by SF-KOSTM following the edges $T_i$ $\rightarrow$ $T_j$ in increasing order of $\tltl_i$ $<$ $\tltl_j$ with the help of atomic \gtcnt. Hence, we can conclude that SF-KOSTM always produce an acyclic $\lopg{H}{\ll_H}$ graph.
	
	Now, using \thmref{coopqh}, we can infer that if a valid history $H$ generated by SF-KOSTM always produces an acyclic $\lopg{H}{\ll_H}$  graph then $H$ is locally-opaque. Hence, any valid history $H$ generated by SF-KOSTM satisfies local-opacity.
	
\end{proof}

\section{Liveness Proof of SF-SVOSTM and SF-KOSTM}
\label{apn:ap-cc}
This section describes the liveness proof of SF-SVOSTM and SF-KOSTM. Both the algorithms ensures starvation-freedom as the progress guarantee.
\subsection{Liveness Proof of SF-SVOSTM}
\label{subsec:LP}

\paragraph{Proof Notations:} Following the notion derived for SF-SVOSTM algorithm, we assume that all the histories accepted by SF-SVOSTM algorithm as \emph{\gen{SF-SVOSTM}}. This subsection considers only histories that are generated by SF-SVOSTM unless explicitly stated otherwise. For simplicity, we consider the sequential histories for our discussion and we can get the sequential history using the linearization points (or LPs) as \emph{first unlocking point of each successful method}.

Let us consider a transaction $T_i$ from the history $H$ \emph{\gen{SF-SVOSTM}}. Each transaction $T_i$ maintains $\langle its_i, cts_i\rangle$. 
 The value of $cts$ is assigned atomically with the help of atomic global counter $gcounter$. So, we use $gcounter$ to approximate the system time.

Apart from these $cts_i$ and $its_i$ transaction $T_i$ maintains $lock$ and $status$. $T_i$ acquires the $lock$ on the keys before accessing it. $status$ can be $\langle live, false, commit \rangle$.  The value of $lock$ and $status$ field change as the execution proceeds. For the sake of understanding, we represent
 the timestamps of a transaction $T_i$ corresponding to history $H$ as $\htits{i}{H}$ and $\htcts{i}{H}$.
 
 To satisfy the \stfdm for SF-SVOSTM, We assumed bounded termination for the fair scheduler as described \asmref{bdtm} in \subsecref{prelim}. In the proposed algorithms, we have considered \emph{TB} as the maximum time-bound of a transaction $T_i$ within this either $T_i$ will return commit or abort in the absence of deadlock. We consider an assumption about the transactions of the system as described \asmref{self} in \subsecref{prelim} which will help to achieve and prove about the starvation-freedom of SF-SVOSTM.

\begin{theorem}
	SF-SVOSTM ensures starvation-freedom in presence of a fair scheduler that satisfies \asmref{bdtm}(bounded-termination) and in the absence of parasitic transactions that satisfies \asmref{self}.
\end{theorem}

\begin{proof}
	Consider any history $H$ generated by SF-SVOSTM algorithm with transaction $T_i$.  Initially, thread $Th_i$ calls \tbeg{()} for $T_i$ which maintains $\langle \emph{its, cts}\rangle$ and set the $status$ as $live$. If $T_i$ is the first incarnation then its \emph{Initial Timestamp (its)} and \emph{Current Timestamp (cts)} are same as $i$. We represent the $\langle its, cts\rangle$ as $\langle its_i, cts_i\rangle$ for transaction $T_i$. If $T_i$ is aborted then thread $Th_i$ executes it again with new incarnation of $T_i$ until $T_i$ commits. Let the new incarnation of $T_i$ say $T_j$ then thread $Th_i$ maintain its $\langle its, cts \rangle$ as $\langle its_i, cts_j \rangle$. $Th_i$ stores the first incarnation $its_i$ of $T_i$ to set the reincarnation $its_j$ of $T_j$ is same as $its_i$. The value of $cts$ is incremented atomically with the help of atomic global counter $gcounter$. So, we use $gcounter$ to approximate the system time. 
	
	Through  \asmref{bdtm}, we can say that $T_i$ will terminate ($\commit$ or $\abort$) in bounded time. If $T_i$ returns abort then $T_i$ will retry again with new incarnation of $T_i$, say $T_j$ while satisfying the \asmref{self}. The incarnation of $T_i$ transaction as $\langle its_i, cts_j \rangle$. So, it can be seen that, $T_i$ will get the lowest $its$ in the system and achieve the highest priority. Eventually, $T_i$ returns commit. Similarly, all the transactions of the $H$ generated by SF-SVOSTM will eventually commit. Hence, SF-SVOSTM ensures starvation-freedom.
\end{proof}


\subsection{Liveness Proof of SF-KOSTM}
\label{subsec:LPK}

\paragraph{Proof Notations:}  It follows the notion derived for SF-KOSTM algorithm, we assume that all the histories accepted by SF-KOSTM algorithm as \emph{\gen{SF-KOSTM}}. This sub-section considers only histories that are generated by SF-KOSTM unless explicitly stated otherwise. For simplicity, we consider the sequential histories for our discussion and we can get the sequential history using the linearization points (or LPs) as \emph{first unlocking point of each successful method}.

The notions of SF-KOSTM is the extension of proof notion defined above for SF-SVOSTM. A transaction $T_i$ of a history $H$ generated by SF-KOSTM execute \tbeg{()} and obtained $wts$ along with $cts$, $its$ used in SF-SVOSTM. So, a transaction $T_i$ of SF-KOSTM maintains three timestamps as $its_i$, $cts_i$, and $wts_i$. The value of $cts_i$ is assigned atomically with the help of atomic global counter $gcounter$. So, we use $gcounter$ to approximate the system time. To achieve the highest $cts_i$, we use $wts_i$ which is significantly larger than $cts_i$ as defined in \subsecref{working}. Along with this SF-KOSTM maintains $tltl$ and $tutl$ corresponding to each transaction $T_i$ as $tltl_i$ and $tutl_i$ to maintain the real-time order.

Apart from these timestamps, transaction $T_i$ maintains $lock$, $vrt$ and $status$. $T_i$ acquires the $lock$ on the keys before accessing it. $vrt$ stands for \emph{version real time} order which helps to maintain the real-time order among the transactions. $status$ can be $\langle live, false, commit \rangle$ same as SF-SVOSTM.  The value of $lock$, $vrt$ and $status$ field change as the execution proceeds. So, we denote them as $\htlock{i}{H}, \htval{i}{H}, \htstat{i}{H}$. Sometimes, we also denote them based on the context of transaction $T_i$ while ignoring $H$ as $\tlock{i}, \tval{i}, \tstat{i}$. For the sake of understanding, we represent
the timestamps of a transaction $T_i$ corresponding to history $H$ as $\htits{i}{H}$, $\htcts{i}{H}$, and $\htwts{i}{H}$ along with  $\htltl{i}{H}, \htutl{i}{H}$.

We use $gcounter$ to approximate the system time. We denote the \syst of history $H$ as the value of $gcounter$ immediately after the last event of $H$. Here, we assume the value of $C$, used to derive the $wts$ as 0.1 in the arguments because on results $C$ as 0.1 performs best. But its not compulsory to keep it 0.1, it can be any value greater than 0 to prove the correctness of SF-KOSTM.



The counter application invokes transactions in such a way that if the current $T_i$ transaction aborts, it invokes a new transaction $T_j$ with the same $its$. We say that $T_i$ is an \emph{\inc} of $T_j$ in a history $H$ if $\htits{i}{H} = \htits{j}{H}$. Thus the multiple \inc{s} of a transaction $T_i$ get invoked by the application until an \inc finally commits. 

To capture this notion of multiple transactions with the same $its$, we define \emph{\incset} (incarnation set) of $T_i$ in $H$ as the set of all the transactions in $H$ which have the same $its$ as $T_i$ and includes $T_i$ as well. Formally, 
\begin{equation*}
\incs{i}{H} = \{T_j|(T_i = T_j) \lor (\htits{i}{H} = \htits{j}{H})\}  
\end{equation*}

Note that from this definition of  \incset, we implicitly get that $T_i$ and all the transactions in its \incset of $H$ also belong to $H$. Formally, $\incs{i}{H} \in \txns{H}$. 

The application invokes different incarnations of a transaction $T_i$ in such a way that as long as an \inc is live, it does not invoke the next \inc. It invokes the next \inc after the current \inc has got aborted. Once an \inc of $T_i$ has committed, it can not have any future \inc{s}. Thus, the application views all the \inc{s} of a transaction as a single \emph{\aptr}. 

We assign \emph{\incn{s}} to all the transactions that have the same $its$. We say that a transaction $T_i$ starts \emph{afresh}, if $\inum{i}$ is 1. We say that $T_i$ is the \ninc of $T_i$ if $T_j$ and $T_i$ have the same $its$ and $T_i$'s \incn is $T_j$'s \incn + 1. Formally, $\langle (\nexti{i} = T_j) \equiv (\tits{i} = \tits{j}) \land (\inum{i} = \inum{j} + 1)\rangle$

As mentioned the objective of the application is to ensure that every \aptr eventually commits. Thus, the applications views the entire \incset as a single \aptr (with all the transactions in the \incset having the same $its$). We can say that an \aptr has committed if in the corresponding \incset a transaction in eventually commits. For $T_i$ in a history $H$, we denote this by a boolean value \incct (incarnation set committed) which implies that either $T_i$ or an \inc of $T_i$ has committed. Formally, we define it as $\inct{i}{H}$

\begin{equation*}
\inct{i}{H} = \begin{cases}
True  & (\exists T_j: (T_j \in \incs{i}{H}) \\ & \land (T_j \in \comm{H}))\\
False & \text{otherwise}
\end{cases}
\end{equation*}

\noindent From the definition of \incct we get the following observations \& lemmas about a transaction $T_i$

\begin{observation}
	\label{obs:inct-term}
	Consider a transaction $T_i$ in a history $H$ with its \incct being true in $H$. Then $T_i$ is terminated (either committed or aborted) in $H$. Formally, $\langle H, T_i: (T_i \in \txns{H}) \land (\inct{i}{H}) \implies (T_i \in \termed{H}) \rangle$. 
\end{observation}

\begin{observation}
	\label{obs:inct-fut}
	Consider a transaction $T_i$ in a history $H$ with its \incct being true in $H1$. Let $H2$ be a extension of $H1$ with a transaction $T_j$ in it. Suppose $T_j$ is an \inc of $T_i$. Then $T_j$'s \incct is true in $H2$. Formally, $\langle H1, H2, T_i, T_j: (H1 \sqsubseteq H2) \land (\inct{i}{H1}) \land (T_j \in \txns{H2}) \land (T_i \in \incs{j}{H2})\implies (\inct{j}{H2}) \rangle$. 
\end{observation}

\begin{lemma}
	\label{lem:inct-diff}
	Consider a history $H1$ with a strict extension $H2$. Let $T_i$ \& $T_j$ be two transactions in $H1$ \& $H2$ respectively. Let $T_j$ not be in $H1$. Suppose $T_i$'s \incct is true. Then \its of $T_i$ cannot be the same as \its of $T_j$. Formally, $\langle H1, H2, T_i, T_j: (H1 \sqsubset H2) \land (\inct{i}{H1}) \land (T_j \in \txns{H2}) \land (T_j \notin \txns{H1}) \implies (\htits{i}{H1} \neq \htits{j}{H2}) \rangle$.
\end{lemma}

\begin{proof}
	Here, we have that $T_i$'s \incct is true in $H1$. Suppose $T_j$ is an \inc of $T_i$, i.e., their \its{s} are the same. We are given that $T_j$ is not in $H1$. This implies that $T_j$ must have started after the last event of $H1$. 
	
	We are also given that $T_i$'s \incct is true in $H1$. This implies that an \inc of $T_i$ or $T_i$ itself has committed in $H1$. After this commit, the application will not invoke another transaction with the same \its as $T_i$. Thus, there cannot be a transaction after the last event of $H1$ and in any extension of $H1$ with the same \its of $T_1$. Hence, $\htits{i}{H1}$ cannot be same as $\htits{j}{H2}$. 
\end{proof}

Now we show the liveness with the following observations, lemmas \& theorems. We start with two observations about that histories of which one is an extension of the other. The following states that for any history, there exists an extension. In other words, we assume that the STM system runs forever and does not terminate. This is required for showing that every transaction eventually commits. 

\begin{observation}
	\label{obs:hist-future}
	Consider a history $H1$ generated by \gen{\ksftm}. Then there is a history $H2$ in \gen{\ksftm} such that $H2$ is a strict extension of $H1$. Formally, $\langle \forall H1: (H1 \in \gen{\ksftm}) \implies (\exists H2: (H2 \in \gen{SF-KOSTM}) \land (H1 \sqsubset H2) \rangle$. 
\end{observation}

\noindent The follow observation is about the transaction in a history and any of its extensions. 

\begin{observation}
	\label{obs:hist-subset}
	Given two histories $H1$ \& $H2$ such that $H2$ is an extension of $H1$. Then, the set of transactions in $H1$ are a subset equal to  the set of transaction in $H2$. Formally, $\langle \forall H1, H2: (H1 \sqsubseteq H2) \implies (\txns{H1} \subseteq \txns{H2}) \rangle$. 
\end{observation}

In order for a transaction $T_i$ to commit in a history $H$, it has to compete with all the live transactions and all the aborted that can become live again as a different \inc. Once a transaction $T_j$ aborts, another \inc of $T_j$ can start and become live again. Thus $T_i$ will have to compete with this \inc of $T_j$ later. Thus, we have the following observation about aborted \& committed transactions. 

\begin{observation}
	\label{obs:abort-retry}
	Consider an aborted transaction $T_i$ in a history $H1$. Then there is an extension of $H1$, $H2$ in which an \inc of $T_i$, $T_j$ is live and has $\tcts{j}$ is greater than $\tcts{i}$. Formally, $\langle H1, T_i: (T_i \in \aborted{H1}) \implies(\exists T_j, H2: (H1 \sqsubseteq H2) \land (T_j \in \live{H2}) \land (\htits{i}{H2} = \htits{j}{H2}) \land (\htcts{i}{H2} < \htcts{j}{H2})) \rangle$. 
\end{observation}

\begin{observation}
	\label{obs:cmt-noinc}
	Consider an committed transaction $T_i$ in a history $H1$. Then there is no extension of $H1$, in which an \inc of $T_i$, $T_j$ is live. Formally, $\langle H1, T_i: (T_i \in \comm{H1}) \implies(\nexists T_j, H2: (H1 \sqsubseteq H2) \land (T_j \in \live{H2}) \land (\htits{i}{H2} = \htits{j}{H2})) \rangle$. 
\end{observation}

\begin{lemma}
	\label{lem:cts-wts}
	Consider a history $H1$ and its extension $H2$. Let $T_i, T_j$ be in $H1, H2$ respectively such that they are \inc{s} of each other. If \wts of $T_i$ is less than \wts of $T_j$ then \cts of $T_i$ is less than \cts $T_j$.
	Formally, $\langle H1, H2, T_i, T_j: (H1 \sqsubset H2) \land (T_i \in \txns{H1}) \land (T_j \in \txns{H2}) \land (T_i \in \incs{j}{H2}) \land (\htwts{i}{H1} < \htwts{j}{H2})\implies (\htcts{i}{H1} < \htcts{j}{H2}) \rangle$
\end{lemma}

\begin{proof}
	Here we are given that 
	\begin{equation}
	\label{eq:wts-ij}
	\htwts{i}{H1} < \htwts{j}{H2}
	\end{equation}
	
	The definition of \wts of $T_i$ is: $\htwts{i}{H1} = \htcts{i}{H1} + C * (\htcts{i}{H1} - \htits{i}{H1})$. Combining this \eqnref{wts-ij}, we get that 
	
	
	$(C + 1) * \htcts{i}{H1} - C * \htits{i}{H1} < (C + 1) * \htcts{j}{H2} - C * \htits{j}{H2} \xrightarrow[\htits{i}{H1} = \htits{j}{H2}]{T_i \in \incs{j}{H2}} \htcts{i}{H1} < \htcts{j}{H2}$. 
\end{proof}

\begin{lemma}
	\label{lem:wts-great}
	Consider a live transaction $T_i$ in a history $H1$ with its $\twts{i}$ less than a constant $\alpha$. Then there is a strict extension of $H1$, $H2$ in which an \inc of $T_i$, $T_j$ is live with \wts greater than $\alpha$. Formally, $\langle H1, T_i: (T_i \in \live{H1}) \land (\htwts{i}{H1} < \alpha) \implies(\exists T_j, H2: (H1 \sqsubseteq H2) \land (T_i \in \incs{j}{H2}) \land ((T_j \in \comm{H2}) \lor ((T_j \in \live{H2}) \land (\htwts{j}{H2} > \alpha)))) \rangle$. 
\end{lemma}

\begin{proof}
	The proof comes the behavior of an \aptr. The application keeps invoking a transaction with the same \its until it commits. Thus the transaction $T_i$ which is live in $H1$ will eventually terminate with an abort or commit. If it commits, $H2$ could be any history after the commit of $T_2$. 
	
	On the other hand if $T_i$ is aborted, as seen in \obsref{abort-retry} it will be invoked again or reincarnated with another \cts and \wts. It can be seen that \cts is always increasing. As a result, the \wts is also increasing. Thus eventually the \wts will become greater $\alpha$. Hence, we have that either an \inc of $T_i$ will get committed or will eventually have \wts greater than or equal to $\alpha$.
\end{proof}

\noindent Next we have a lemma about \cts of a transaction and the \syst of a history. 

\begin{lemma}
	\label{lem:cts-syst}
	Consider a transaction $T_i$ in a history $H$. Then, we have that \cts of $T_i$ will be less than or equal to \syst of $H$. Formally, $\langle T_i, H1: (T_i \in \txns{H}) \implies (\htcts{i}{H} \leq \hsyst{H}) \rangle$. 
\end{lemma}

\begin{proof}
	We get this lemma by observing the \mth{s} of the STM System that increment the \tcntr which are \begt and \tryc. It can be seen that \cts of $T_i$ gets assigned in the \begt \mth. So if the last \mth of $H$ is the \begt of $T_i$ then we get that \cts of $T_i$ is same as \syst of $H$. On the other hand if some other \mth got executed in $H$ after \begt of $T_i$ then we have that \cts of $T_i$ is less than \syst of $H$. Thus combining both the cases, we get that \cts of $T_i$ is less than or equal to as \syst of $H$, i.e., $(\htcts{i}{H} \leq \hsyst{H})$
\end{proof}

\noindent From this lemma, we get the following corollary which is the converse of the lemma statement

\begin{corollary}
	\label{cor:cts-syst}
	Consider a transaction $T_i$ which is not in a history $H1$ but in an strict extension of $H1$, $H2$. Then, we have that \cts of $T_i$ is greater than the \syst of $H$. Formally, $\langle T_i, H1, H2: (H1 \sqsubset H2) \land (T_i \notin \txns{H1}) \land (T_i \in \txns{H2}) \implies (\htcts{i}{H2} > \hsyst{H1}) \rangle$. 
\end{corollary}

\noindent Now, we have lemma about the \mth{s} of \ksftm completing in finite time. 

\begin{lemma}
	\label{lem:mth-fdm}
	If all the locks are fair and the underlying system scheduler is fair then all the \mth{s} of \ksftm will eventually complete.
\end{lemma}

\begin{proof}
	It can be seen that in any \mth, whenever a transaction $T_i$ obtains multiple locks, it obtains locks in the same order: first lock relevant \tobj{s} in a pre-defined order and then lock relevant \glock{s} again in a predefined order. Since all the locks are obtained in the same order, it can be seen that the \mth{s} of \ksftm will not deadlock. 
	
	It can also be seen that none of the \mth{s} have any unbounded while loops. All the loops in \tryc \mth iterate through all the \tobj{s} in the write-set of $T_i$. Moreover, since we assume that the underlying scheduler is fair, we can see that no thread gets swapped out infinitely. Finally, since we assume that all the locks are fair, it can be seen all the \mth{s} terminate in finite time.
\end{proof}

\begin{theorem}
	\label{thm:trans-com|abt}
	Every transaction either commits or aborts in finite time.
\end{theorem}

\begin{proof}
	This theorem comes directly from the \lemref{mth-fdm}. Since every \mth of \ksftm will eventually complete, all the transactions will either commit or abort in finite time.
\end{proof}

\noindent From this theorem, we get the following corollary which states that the maximum \emph{lifetime} of any transaction is $L$. 

\begin{corollary}
	\label{cor:cts-L}
	Any transaction $T_i$ in a history $H$ will either commit or abort before the \syst of $H$ crosses $\tcts{i} + L$. 
\end{corollary}

\noindent The following lemma connects \wts and \its of two transactions, $T_i, T_j$. 

\begin{lemma}
	\label{lem:wts-its}
	Consider a history $H1$ with two transactions $T_i, T_j$. Let $T_i$ be in $\live{H1}$. Suppose $T_j$'s \wts is greater or equal to $T_i$' s \wts. Then \its of $T_j$ is less than $\tits{i} + 2*L$. Formally, $\langle H, T_i, T_j : (\{ T_i, T_j\} \subseteq \txns{H}) \land ( T_i \in \live{H}) \land (\htwts{j}{H} \geq \htwts{i}{H}) \Longrightarrow (\htits{i}{H} + 2L \geq \htits{j}{H}) \rangle$.
\end{lemma}

\begin{proof}
	Since $T_i$ is live in $H1$, from \corref{cts-L}, we get that it terminates before the system time, $\tcntr$ becomes $\tcts{i} + L$. Thus, \syst of history $H1$ did not progress beyond $\tcts{i} + L$. Hence, for any other transaction $T_j$ (which is either live or terminated) in $H1$, it must have started before \syst has crossed $\tcts{i} + L$. Formally $\langle \tcts{j} \leq \tcts{i} + L \rangle$.
	
	Note that we have defined \wts of a transaction $T_j$ as: $\twts{j} = (\tcts{j} + C * (\tcts{j} - \tits{j}))$. Now, let us consider the difference of the \wts{s} of both the transactions. 
	\noindent 
	\begin{math}
		\twts{j} - \twts{i}  = (\tcts{j} + C * (\tcts{j} - \tits{j})) - (\tcts{i} + C * (\tcts{i} - \tits{i})) \\ 
		= (C + 1)(\tcts{j} - \tcts{i}) - C(\tits{j} - \tits{i}) \\
		\leq (C + 1)L - C(\tits{j} - \tits{i}) \qquad [\because \tcts{j} \leq \tcts{i} + L] \\
		= 2*L + \tits{i} - \tits{j}  \qquad [\because C = 1] \\
	\end{math}
	
	\noindent Thus, we have that: $ \langle (\tits{i} + 2L - \tits{j}) \geq (\twts{j} - \twts{i}) \rangle$. This gives us that \\
	$((\twts{j} - \twts{i}) \geq 0) \Longrightarrow ((\tits{i} + 2L - \tits{j}) \geq 0)$. 
	
	\noindent From the above implication we get that, 
	$(\twts{j} \geq \twts{i}) \Longrightarrow (\tits{i} + 2L \geq \tits{j})$.
	
\end{proof}

It can be seen that \ksftm algorithm gives preference to transactions with lower \its to commit. To understand this notion of preference, we define a few notions of enablement of a transaction $T_i$ in a history $H$. We start with the definition of \emph{\itsen} as:

\begin{definition}
	\label{defn:itsen}
	We say $T_i$ is \emph{\itsen} in $H$ if for all transactions $T_j$ with \its lower than \its of $T_i$ in $H$ have \incct to be true. Formally, 
	\begin{equation*}
		\itsenb{i}{H} = \begin{cases}
			True    & (T_i \in \live{H}) \land (\forall T_j \in \txns{H} : (\htits{j}{H} < \htits{i}{H}) \\& \implies (\inct{j}{H})) \\
			False	& \text{otherwise}
		\end{cases}
	\end{equation*}
\end{definition}

\noindent The follow lemma states that once a transaction $T_i$ becomes \itsen it continues to remain so until it terminates. 

\begin{lemma}
	\label{lem:itsen-future}
	Consider two histories $H1$ and $H2$ with $H2$ being a extension of $H1$. Let a transaction $T_i$ being live in both of them. Suppose $T_i$ is \itsen in $H1$. Then $T_i$ is \itsen in $H2$ as well. Formally, $\langle H1, H2, T_i: (H1 \sqsubseteq H2) \land (T_i \in \live{H1}) \land (T_i \in \live{H2}) \land (\itsenb{i}{H1}) \implies (\itsenb{i}{H2}) \rangle$. 
\end{lemma}

\begin{proof}
	When $T_i$ begins in a history $H3$ let the set of transactions with \its less than $\tits{i}$ be $smIts$. Then in any extension of $H3$, $H4$ the set of transactions with \its less than $\tits{i}$ remains as $smIts$. 
	
	Suppose $H1, H2$ are extensions of $H3$. Thus in $H1, H2$ the set of transactions with \its less than $\tits{i}$ will be $smIts$. Hence, if $T_i$ is \itsen in $H1$ then all the transactions $T_j$ in $smIts$ are $\inct{j}{H1}$. It can be seen that this continues to remain true in $H2$. Hence in $H2$, $T_i$ is also \itsen which proves the lemma.
\end{proof}

The following lemma deals with a committed transaction $T_i$ and any transaction $T_j$ that terminates later. In the following lemma, $\incv$ is any constant greater than or equal to 1. 

\begin{lemma}
	\label{lem:tryci-j}
	Consider a history $H$ with two transactions $T_i, T_j$ in it. Suppose transaction $T_i$ commits before $T_j$ terminates (either by commit or abort) in $H$. Then $\ct_i$ is less than $\ct_j$ by at least $\incv$. Formally, $\langle H, \{T_i, T_j\} \in \txns{H}: (\tryc_i <_H \term_j) \implies (\ct_i + \incv \leq \ct_j)\rangle$. 
\end{lemma}

\begin{proof}
	When $T_i$ commits, let the value of the global $\tcntr$ be $\alpha$. It can be seen that in \begt \mth, $\ct_j$ get initialized to $\infty$. The only place where $\ct_j$ gets modified is at \Lineref{tcv201} of \tryc. Thus if $T_j$ gets aborted before executing \tryc \mth or before this line of \tryc we have that $\ct_j$ remains at $\infty$. Hence in this case we have that $\langle \ct_i + \incv < \ct_j \rangle$.
	
	If $T_j$ terminates after executing \Lineref{tcv201} of \tryc \mth then $\ct_j$ is assigned a value, say $\beta$. It can be seen that $\beta$ will be greater than $\alpha$ by at least $\incv$ due to the execution of this line. Thus, we have that $\langle \alpha + \incv \leq \beta \rangle$
\end{proof}

\noindent The following lemma connects the \tltl and \ct of a transaction $T_i$. 

\begin{lemma}
	\label{lem:ti|tltl-comt}
	Consider a history $H$ with a transaction $T_i$ in it. Then in $H$, $\ttltl{i}$ will be less than or equal to $\ct_i$. Formally, $\langle H, \{T_i\} \in \txns{H}: (\htltl{i}{H} \leq H.\ct_i) \rangle$.
\end{lemma}

\begin{proof}
	Consider the transaction $T_i$. In \begt \mth, $\ct_i$ get initialized to $\infty$. The only place where $\ct_i$ gets modified is at \Lineref{tcv201} of \tryc. Thus if $T_i$ gets aborted before this line or if $T_i$ is live we have that $(\ttltl{i} \leq \ct_i)$. On executing \Lineref{tcv201}, $\ct_i$ gets assigned to some finite value and it does not change after that. 
	
	It can be seen that $\ttltl{i}$ gets initialized to $\tcts{i}$ in \Lineref{begin15} of \begt \mth. In that line, $\tcts{i}$ reads $\tcntr$ and increments it atomically. Then in \Lineref{tcv201}, $\ct_i$ gets assigned the value of $\tcntr$ after incrementing it. Thus, we clearly get that $\tcts{i} (= \ttltl{i}\text{ initially}) < \ct_i$. Then $\ttltl{i}$ gets updated on \Lineref{rvm25} of read, \Lineref{tcv20} and \Lineref{tcv34} of \tryc \mth{s}. Let us analyze them case by case assuming that $\ttltl{i}$ was last updated in each of these \mth{s} before the termination of $T_i$:
	
	\begin{enumerate}
		\item \label{case:read} \Lineref{rvm25} of read \mth: Suppose this is the last line where $\ttltl{i}$ updated. Here $\ttltl{i}$ gets assigned to 1 + \vt of the previously committed version which say was created by a transaction $T_j$. Thus, we have the following equation, 
		\begin{equation}
			\label{eq:tltl-vt}
			\ttltl{i} = 1 + x[j].\vt
		\end{equation}
		
		It can be seen that $x[j].\vt$ is same as $\ttltl{j}$ when $T_j$ executed \Lineref{tcv20} of \tryc. Further, $\ttltl{j}$ in turn is same as $\ttutl{j}$ due to \Lineref{tcv20} of \tryc. From \Lineref{tcv202}, it can be seen that $\ttutl{j}$ is less than or equal to $\ct_j$ when $T_j$ committed. Thus we have that 
		\begin{equation}
			\label{eq:tltl-ct}
			x[j].\vt = \ttltl{j} = \ttutl{j} \leq \ct_j
		\end{equation}
		
		It is clear that from the above discussion that  $T_j$ executed \tryc \mth before $T_i$ terminated (i.e. $\tryc_j <_{H1} \term_i$). From \eqnref{tltl-vt} and \eqnref{tltl-ct}, we get \\
		\begin{math}
			\ttltl{i} \leq 1 + \ct_j \xrightarrow[]{\incv \geq 1} \ttltl{i} \leq \incv + \ct_j \xrightarrow[]{\lemref{tryci-j}} \ttltl{i} \leq \ct_i
		\end{math}
		
		\item \label{case:tryc-short} \Lineref{tcv20} of \tryc \mth: The reasoning in this case is very similar to the above case. 
		
		\item \label{case:tryc-long} \Lineref{tcv34} of \tryc \mth: In this line, $\ttltl{i}$ is made equal to $\ttutl{i}$. Further, in \Lineref{tcv202}, $\ttutl{i}$ is made lesser than or equal to $\ct_{i}$. Thus combing these, we get that $\ttltl{i} \leq \ct_{i}$. It can be seen that the reasoning here is similar in part to \csref{read}.  
	\end{enumerate}
	
	Hence, in all the three cases we get that $\langle \ttltl{i} \leq \ct_i \rangle$. 
\end{proof}

\noindent The following lemma connects the \tutl,\ct of a transaction $T_i$ with \wts of a transaction $T_j$ that has already committed. 

\begin{lemma}
	\label{lem:ti|tutl-comt}
	Consider a history $H$ with a transaction $T_i$ in it. Suppose $\ttutl{i}$ is less than $\ct_i$. Then, there is a committed transaction $T_j$ in $H$ such that $\twts{j}$ is greater than $\twts{i}$. Formally, $\langle H \in \gen{\ksftm}, \{T_i\} \in \txns{H}: (\htutl{i}{H} < H.\ct_i) \implies (\exists T_j \in \comm{H}: \htwts{j}{H} > \htwts{i}{H}) \rangle$.
\end{lemma}

\begin{proof}
	It can be seen that $\tutl_i$ initialized in \begt \mth to $\infty$. $\ttutl{i}$ is updated in \Lineref{rvm22} of read \mth, \Lineref{tcv22} \& \Lineref{tcv202} of \tryc \mth. If $T_i$ executes \Lineref{rvm22} of read \mth and/or \Lineref{tcv22} of \tryc \mth then $\ttutl{i}$ gets decremented to some value less than $\infty$, say $\alpha$. Further, it can be seen that in both these lines the value of $\ttutl{i}$ is possibly decremented from $\infty$ because of $nextVer$ (or $ver$), a version of $x$ whose \ts is greater than $T_i$'s \wts. This implies that some transaction $T_j$, which is committed in $H$, must have created $nextVer$ (or $ver$) and $\twts{j} > \twts{i}$. 
	
	Next, let us analyze the value of $\alpha$. It can be seen that $\alpha = x[nextVer/ver].vrt - 1$ where $nextVer/ver$ was created by $T_j$. Further, we can see when $T_j$ executed \tryc, we have that $x[nextVer].vrt = \ttltl{j}$ (from \Lineref{tc19}). From \lemref{ti|tltl-comt}, we get that $\ttltl{j} \leq \ct_j$. This implies that $\alpha < \ct_j$. Now, we have that $T_j$ has already committed before the termination of $T_i$. Thus from \lemref{tryci-j}, we get that $\ct_j < \ct_i$. Hence, we have that, 
	
	\begin{equation}
		\label{eq:alph-ct}
		\alpha < \ct_i 
	\end{equation}
	
	Now let us consider \Lineref{tcv202} executed by $T_i$ which causes $\ttutl{i}$ to change. This line will get executed only after both \Lineref{rvm22} of read \mth, \Lineref{tcv22} of \tryc \mth. This is because every transaction executes \tryc \mth only after read \mth. Further within \tryc \mth, \Lineref{tcv202} follows \Lineref{tcv22}. 
	
	There are two sub-cases depending on the value of $\ttutl{i}$ before the execution of \Lineref{tcv202}: (i) If $\ttutl{i}$ was $\infty$  and then get decremented to $\ct_i$ upon executing this line, then we get $\ct_i =  \ttutl{i}$. From \eqnref{alph-ct},  we can ignore this case. (ii) Suppose the value of $\ttutl{i}$ before executing \Lineref{tcv202} was $\alpha$. Then from \eqnref{alph-ct} we get that $\ttutl{i}$ remains at $\alpha$ on execution of \Lineref{tcv202}. This implies that a transaction $T_j$ committed such that $\twts{j} > \twts{i}$.
\end{proof}

\noindent The following lemma connects the \tltl of a committed transaction $T_j$ and \ct of a transaction $T_i$ that commits later. 

\begin{lemma}
	\label{lem:tltlj-comti}
	Consider a history $H1$ with transactions $T_i, T_j$ in it. Suppose $T_j$ is committed and $T_i$ is live in $H1$. Then in any extension of $H1$, say $H2$, $\ttltl{j}$ is less than or equal to $\ct_i$. Formally, $\langle {H1, H2} \in \gen{\ksftm}, \{T_i, T_j\} \subseteq \txns{H1, H2}: (H1 \sqsubseteq H2) \land (T_j \in \comm{H1}) \land (T_i \in \live{H1}) \implies (\htltl{j}{H2} < H2.\ct_i) \rangle$.
\end{lemma}

\begin{proof}
	As observed in the previous proof of \lemref{ti|tltl-comt}, if $T_i$ is live  or aborted in $H2$, then its \ct is $\infty$. In both these cases, the result follows.
	
	If $T_i$ is committed in $H2$ then, one can see that \ct of $T_i$ is not $\infty$. In this case, it can be seen that $T_j$ committed before $T_i$. Hence, we have that $\ct_j < \ct_i$. From \lemref{ti|tltl-comt}, we get that $\ttltl{j} \leq \ct_j$. This implies that $\ttltl{j} < \ct_i$. 
\end{proof}

\noindent In the following sequence of lemmas, we identify the condition by when a transaction will commit. 

\begin{lemma}
	\label{lem:its-wts}
	Consider two histories $H1, H3$ such that $H3$ is a strict extension of $H1$. Let $T_i$ be a transaction in  $\live{H1}$ such that $T_i$ \itsen in $H1$ and $\gval_i$ flag is true in $H1$. Suppose $T_i$ is aborted in $H3$. Then there is a history $H2$ which is an extension of $H1$ (and could be same as $H1$) such that (1) Transaction $T_i$ is live in $H2$; (2) there is a transaction $T_j$ that is live in ${H2}$; (3) $\htwts{j}{H2}$ is greater than $\htwts{i}{H2}$; (4) $T_j$ is committed in $H3$. Formally, $ \langle H1, H3, T_i: (H1 \sqsubset H3) \land (T_i \in \live{H1}) \land (\htval{i}{H1} = True) \land (\itsenb{i}{H1}) \land (T_i \in \aborted{H3})) \implies (\exists H2, T_j: (H1 \sqsubseteq H2 \sqsubset H3) \land (T_i \in \live{H2}) \land (T_j \in \txns{H2}) \land (\htwts{i}{H2} < \htwts{j}{H2}) \land (T_j \in \comm{H3})) \rangle$. 
\end{lemma}

\begin{proof}
	To show this lemma, w.l.o.g we assume that $T_i$ on executing either read or \tryc in $H2$ (which could be same as $H1$) gets aborted resulting in $H3$. Thus, we have that $T_i$ is live in $H2$. Here $T_i$ is \itsen in $H1$. From \lemref{itsen-future}, we get that $T_i$ is \itsen in $H2$ as well.
	
	
	Let us sequentially consider all the lines where a $T_i$ could abort. In $H2$, $T_i$ executes one of the following lines and is aborted in $H3$. We start with \tryc method. 
	
	\begin{enumerate}
		
		\item \tryc: 
		
		\begin{enumerate}
			\item \Lineref{tc3} \label{case:init-tc-chk}: This line invokes abort() method on $T_i$ which releases all the locks and returns $\mathcal{A}$ to the invoking thread. Here $T_i$ is aborted because its \val flag, is set to false by some other transaction, say $T_j$, in its \tryc algorithm. This can occur in Lines: \ref{lin:tcv14}, \ref{lin:tcv29} where $T_i$ is added to $T_j$'s \abl set. Later in \Lineref{tcv39}, $T_i$'s \val flag is set to false. Note that $T_i$'s \val is true (after the execution of the last event) in $H1$. Thus, $T_i$'s \val flag must have been set to false in an extension of $H1$, which we again denote as $H2$.
			
			This can happen only if in both the above cases, $T_j$ is live in $H2$ and its \its is less than $T_i$'s \its. But we have that $T_i$'s \itsen in $H2$. As a result, it has the smallest among all live and aborted transactions of $H2$. Hence, there cannot exist such a $T_j$ which is live and $\htits{j}{H2} < \htits{i}{H2}$. Thus, this case is not possible. 
			
			\item \Lineref{tcv111}: This line is executed in $H2$ if there exists no version of $x$ whose \ts is less than $T_i$'s \wts. This implies that all the versions of $x$ have \ts{s} greater than $\twts{i}$. Thus the transactions that created these versions have \wts greater than $\twts{i}$ and have already committed in $H2$. Let $T_j$ create one such version. Hence, we have that $\langle (T_j \in \comm{H2}) \implies (T_j \in \comm{H3}) \rangle$ since $H3$ is an extension of $H2$. 
			
			\item \Lineref{tcv10} \label{case:mid-tc-chk}: This case is similar to \csref{init-tc-chk}, i.e., \Lineref{tcv39}. 
			
			\item \Lineref{tcv15} \label{case:its-chk1}: In this line, $T_i$ is aborted as some other transaction $T_j$ in $T_i$'s \lrl has committed. Any transaction in $T_i$'s \lrl has \wts greater than $T_i$'s \wts. This implies that $T_j$ is already committed in $H2$ and hence committed in $H3$ as well. 
			
			\item \Lineref{tcv24} \label{case:tc-lts-cross}: In this line, $T_i$ is aborted because its lower limit has crossed its upper limit. First, let us consider $\ttutl{i}$. It is initialized in \begt \mth to $\infty$. As long as it is $\infty$, these limits cannot cross each other. Later, $\ttutl{i}$ is updated in \Lineref{rvm22} of read \mth, \Lineref{tcv22} \& \Lineref{tcv202} of \tryc \mth. Suppose $\ttutl{i}$ gets decremented to some value $\alpha$ by one of these lines. 
			
			Now there are two cases here: (1) Suppose $\ttutl{i}$ gets decremented to $\ct_i$ due to \Lineref{tcv202} of \tryc \mth. Then from \lemref{ti|tltl-comt}, we have $\ttltl{i} \leq \ct_i =  \ttutl{i}$. Thus in this case, $T_i$ will not abort. (2) $\ttutl{i}$ gets decremented to $\alpha$ which is less than $\ct_i$. Then from \lemref{ti|tutl-comt}, we get that there is a committed transaction $T_j$ in $\comm{H2}$ such that $\twts{j} > \twts{i}$. This implies that $T_j$ is in $\comm{H3}$.

			\item \Lineref{tcv30} \label{case:its-chk2}: In this case, $T_k$ is in $T_i$'s \srl and is committed in $H1$. And, from this case, we have that
			
			\begin{equation}
				\label{eq:tltl-k_i}
				\htutl{i}{H2} \leq \htltl{k}{H2}
			\end{equation}
			
			From the assumption of this case, we have that $T_k$ commits before $T_i$. Thus, from \lemref{tltlj-comti}, we get that $\ct_k < \ct_i$. From \lemref{ti|tltl-comt}, we have that $\ttltl{k} \leq \ct_k$. Thus, we get that $\ttltl{k} < \ct_i$. Combining this with the inequality of this case \eqnref{tltl-k_i}, we get that $\ttutl{i} < \ct_i$.
			
			Combining this inequality with \lemref{ti|tutl-comt}, we get that there is a transaction $T_j$ in $\comm{H2}$ and $\htwts{j}{H2} > \htwts{i}{H2}$. This implies that $T_j$ is in $\comm{H3}$ as well.
			
		\end{enumerate}
	\item STM read: 
	
	\begin{enumerate}
		\item \Lineref{rvm5}: This case is similar to \csref{init-tc-chk}, i.e., \Lineref{tc3}
		
		\item \Lineref{rvm27}: The reasoning here is similar to \csref{tc-lts-cross}, i.e., \Lineref{tcv24}.
	\end{enumerate}
	
\end{enumerate}

\end{proof}

The interesting aspect of the above lemma is that it gives us a insight as to when a $T_i$ will get commit. If an \itsen transaction $T_i$ aborts then it is because of another transaction $T_j$ with \wts higher than $T_i$ has committed. To precisely capture this, we define two more notions of a transaction being enabled \emph{\cdsen} and \emph{\finen}. To define these notions of enabled, we in turn define a few other auxiliary notions. We start with \emph{\affset},
\begin{equation*}
	\haffset{i}{H} = \{T_j|(T_j \in \txns{H}) \land (\htits{j}{H} < \htits{i}{H} + 2*L)\}
\end{equation*}

From the description of \ksftm algorithm and \lemref{wts-its}, it can be seen that a transaction $T_i$'s commit can depend on committing of transactions (or their \inc{s}) which have their \its less than \its of $T_i$ + $2*L$, which is $T_i$'s \affset. We capture this notion of dependency for a transaction $T_i$ in a history $H$ as \emph{commit dependent set} or \emph{\cdset} as: the set of all transactions $T_j$ in $T_i$'s \affset that do not any \inc that is committed yet, i.e., not yet have their \incct flag set as true. Formally, 

\begin{equation*}
	\hcds{i}{H} = \{T_j| (T_j \in \haffset{i}{H}) \land (\neg\inct{j}{H}) \}
\end{equation*}

\noindent Based on this definition of \cdset, we next define the notion of \cdsen. 

\begin{definition}
	\label{defn:cdsen}
	We say that transaction $T_i$ is \emph{\cdsen} if the following conditions hold true (1) $T_i$ is live in $H$; (2) \cts of $T_i$ is greater than or equal to \its of $T_i$ + $2*L$; (3) \cdset of $T_i$ is empty, i.e., for all transactions $T_j$ in $H$ with \its lower than \its of  $T_i$ + $2*L$ in $H$ have their \incct to be true. Formally, 
	
	\begin{equation*}
		\cdsenb{i}{H} = \begin{cases}
			True    & (T_i \in \live{H}) \land (\htcts{i}{H} \geq \htits{i}{H} + 2*L) \land 
			\\ &(\hcds{i}{H} = \phi) \\
			False	& \text{otherwise}
		\end{cases}
	\end{equation*}
\end{definition}

\noindent The meaning and usefulness of these definitions will become clear in the course of the proof. In fact, we later show that once the transaction $T_i$ is \cdsen, it will eventually commit. We will start with a few lemmas about these definitions. 

\begin{lemma}
	\label{lem:its-enb}	Consider a transaction $T_i$ in a history $H$. If $T_i$ is \cdsen then $T_i$ is also \itsen. Formally, $\langle H, T_i: (T_i \in \txns{H}) \land (\cdsenb{i}{H}) \implies (\itsenb{i}{H}) \rangle$. 
\end{lemma}

\begin{proof}
	If $T_i$ is \cdsen in $H$ then it implies that $T_i$ is live in $H$. From the definition of \cdsen, we get that $\hcds{i}{H}$ is $\phi$ implying that any transaction $T_j$ with $\tits{k}$ less than $\tits{i} + 2*L$ has its \incct flag as true in $H$. Hence, for any transaction $T_k$ having $\tits{k}$ less than $\tits{i}$, $\inct{k}{H}$ is also true. This shows that $T_i$ is \itsen in $H$.
\end{proof}

\begin{lemma}
	\label{lem:cds-tk-h1}
	Consider a transaction $T_i$ which is \cdsen in a history $H1$. Consider an extension of $H1$, $H2$ with a transaction $T_j$ in it such that $T_i$ is an \inc of $T_j$. Let $T_k$ be a transaction in the \affset of $T_j$ in $H2$ Then $T_k$ is also in the set of transaction of $H1$. Formally, $\langle H1, H2, T_i, T_j, T_k: (H1 \sqsubseteq H2) \land  (\cdsenb{i}{H1}) \land (T_i \in \incs{j}{H2}) \land (T_k \in \haffset{j}{H2}) \implies (T_k \in \txns{H1}) \rangle$
\end{lemma}

\begin{proof} 
	Since $T_i$ is \cdsen in $H1$, we get (from the definition of \cdsen) that 
	\begin{equation}
		\label{eq:ti-cts-its}
		\htcts{i}{H1} \geq \htits{i}{H1} + 2*L
	\end{equation}
	
	Here, we have that $T_k$ is in $\haffset{j}{H2}$. Thus from the definition of \affset, we get that 
	\begin{equation}
		\label{eq:tk-tj-aff}
		\htits{k}{H2} < \htits{j}{H2} + 2*L
	\end{equation}
	
	Since $T_i$ and $T_j$ are \inc{s} of each other, their \its are the same. Combining this with \eqnref{tk-tj-aff}, we get that 
	\begin{equation}
		\label{eq:tk-ti-h12}
		\htits{k}{H2} < \htits{i}{H1} + 2*L
	\end{equation}
	
	We now show this proof through contradiction. Suppose $T_k$ is not in $\txns{H1}$. Then there are two cases:
	
	\begin{itemize}
		\item No \inc of $T_k$ is in $H1$: This implies that $T_k$ starts afresh after $H1$. Since $T_k$ is not in $H1$, from \corref{cts-syst} we get that
		
		$\htcts{k}{H2} > \hsyst{H1} \xrightarrow [\htcts{k}{H2} = \htits{k}{H2}] {T_k \text{ starts afresh}}\htits{k}{H2} > \hsyst{H1} \xrightarrow [\hsyst{H1} \geq \htcts{i}{H1}]{(T_i \in H1) \land \lemref{cts-syst}} \htits{k}{H2} > \htcts{i}{H1} \xrightarrow {\eqnref{ti-cts-its}} \htits{k}{H2} > \htits{i}{H1} + 2*L \xrightarrow {\htits{i}{H1} = \htits{j}{H2}} \htits{k}{H2} > \htits{j}{H2} + 2*L$
		
		But this result contradicts with \eqnref{tk-tj-aff}. Hence, this case is not possible. 
		
		\item There is an \inc of $T_k$, $T_l$ in $H1$: In this case, we have that 
		
		\begin{equation}
			\label{eq:tl-h1}
			\htits{l}{H1} = \htits{k}{H2}
		\end{equation}
		
		
		Now combing this result with \eqnref{tk-ti-h12}, we get that $\htits{l}{H1} < \htits{i}{H1} + 2*L$. This implies that $T_l$ is in \affset of $T_i$ in $H1$. Since $T_i$ is \cdsen, we get that $T_l$'s \incct must be true. 
		
		We also have that $T_k$ is not in $H1$ but in $H2$ where $H2$ is an extension of $H1$. Since $H2$ has some events more than $H1$, we get that $H2$ is a strict extension of $H1$.
		
		Thus, we have that, $(H1 \sqsubset H2) \land (\inct{l}{H1}) \land (T_k \in \txns{H2}) \land (T_k \notin \txns{H1})$. Combining these with \lemref{inct-diff}, we get that $(\htits{l}{H1} \neq \htits{k}{H2})$. But this result contradicts \eqnref{tl-h1}. Hence, this case is also not possible.
	\end{itemize}
	Thus from both the cases we get that $T_k$ should be in $H1$. Hence proved.
\end{proof}

\begin{lemma}
	\label{lem:aff-tkinc-h1}
	Consider two histories $H1, H2$ where $H2$ is an extension of $H1$. Let $T_i, T_j, T_k$ be three transactions such that $T_i$ is in $\txns{H1}$ while $T_j, T_k$ are in $\txns{H2}$. Suppose we have that (1) $\tcts{i}$ is greater than $\tits{i} + 2*L$ in $H1$; (2) $T_i$ is an \inc of $T_j$; (3) $T_k$ is in \affset of $T_j$ in $H2$. Then an \inc of $T_k$, say $T_l$ (which could be same as $T_k$) is in $\txns{H1}$. Formally, $\langle H1, H2, T_i, T_j, T_k: (H1 \sqsubseteq H2) \land (T_i \in \txns{H1})  \land (\{T_j, T_k\} \in \txns{H2}) \land (\htcts{i}{H1} > \htits{i}{H1} + 2*L) \land (T_i \in \incs{j}{H2}) \land (T_k \in \haffset{j}{H2}) \implies (\exists T_l: (T_l \in \incs{k}{H2}) \land (T_l \in \txns{H1})) \rangle$
\end{lemma}

\begin{proof} 
	
	\noindent This proof is similar to the proof of \lemref{cds-tk-h1}. We are given that 
	\begin{equation}
		\label{eq:given-ti-ctsits}
		\htcts{i}{H1} \geq \htits{i}{H1} + 2*L
	\end{equation}
	
	We now show this proof through contradiction. Suppose no \inc of $T_k$ is in $\txns{H1}$. This implies that $T_k$ must have started afresh in some history $H3$ which is an extension of $H1$. Also note that $H3$ could be same as $H2$ or a prefix of it, i.e., $H3 \sqsubseteq H2$. Thus, we have that 
	
	\noindent
	\begin{math}
		\htits{k}{H3} > \hsyst{H1} \xrightarrow{\lemref{cts-syst}} \htits{k}{H3} > \htcts{i}{H1} \xrightarrow{\eqnref{given-ti-ctsits}} \htits{k}{H3} > \htits{i}{H1} + 2*L \xrightarrow{\htits{i}{H1} = \htits{j}{H2}} \htits{k}{H3} > \htits{j}{H2} + 2*L \xrightarrow[\obsref{hist-subset}]{H3 \sqsubseteq H2} \htits{k}{H2} > \htits{j}{H2} + 2*L \xrightarrow[definition]{\affset} T_k \notin \haffset{j}{H2}
	\end{math}
	
	But we are given that $T_k$ is in \affset of $T_j$ in $H2$. Hence, it is not possible that $T_k$ started afresh after $H1$. Thus, $T_k$ must have a \inc in $H1$.
\end{proof}

\begin{lemma}
	\label{lem:aff-same}
	Consider a transaction $T_i$ which is \cdsen in a history $H1$. Consider an extension of $H1$, $H2$ with a transaction $T_j$ in it such that $T_j$ is an \inc of $T_i$ in $H2$. Then \affset of $T_i$ in $H1$ is same as the \affset of $T_j$ in $H2$. Formally, $\langle H1, H2, T_i, T_j: (H1 \sqsubseteq H2) \land  (\cdsenb{i}{H1}) \land (T_j \in \txns{H2}) \land (T_i \in \incs{j}{H2}) \implies ((\haffset{i}{H1} = \haffset{j}{H2})) \rangle$
\end{lemma}

\begin{proof}
	From the definition of \cdsen, we get that $T_i$ is in $\txns{H1}$. Now to prove that \affset{s} are the same, we have to show that $(\haffset{i}{H1} \subseteq \haffset{j}{H2})$ and $(\haffset{j}{H1} \subseteq \haffset{i}{H2})$. We show them one by one:
	
	\paragraph{$(\haffset{i}{H1} \subseteq \haffset{j}{H2})$:} Consider a transaction $T_k$ in $\haffset{i}{H1}$. We have to show that $T_k$ is also in $\haffset{j}{H2}$. From the definition of \affset, we get that 
	\begin{equation}
		\label{eq:tk-h1}
		T_k \in \txns{H1}
	\end{equation}
	
	\noindent Combining \eqnref{tk-h1} with \obsref{hist-subset}, we get that 
	\begin{equation}
		\label{eq:tk-h2}
		T_k \in \txns{H2}
	\end{equation}
	
	\noindent From the definition of \its, we get that 
	\begin{equation}
		\label{eq:its-h1-h2}
		\htits{k}{H1} = \htits{k}{H2}
	\end{equation}
	
	\noindent Since $T_i, T_j$ are \inc{s} we have that . 
	\begin{equation}
		\label{eq:its-ij}
		\htits{i}{H1} = \htits{j}{H2}
	\end{equation}
	
	\noindent From the definition of \affset, we get that, \\ 
	$\htits{k}{H1} < \htits{i}{H1} + 2*L \xrightarrow{\eqnref{its-h1-h2}} \htits{k}{H2} < \htits{i}{H1} + 2*L \xrightarrow{\eqnref{its-ij}} \htits{k}{H2} < \htits{j}{H2} + 2*L$
	
	\noindent Combining this result with \eqnref{tk-h2}, we get that $T_k \in \haffset{j}{H2}$. 
	\paragraph{$(\haffset{i}{H1} \subseteq \haffset{j}{H2})$:} Consider a transaction $T_k$ in $\haffset{j}{H2}$. We have to show that $T_k$ is also in $\haffset{i}{H1}$. From the definition of \affset, we get that $T_k \in \txns{H2}$.
	
	Here, we have that $(H1 \sqsubseteq H2) \land  (\cdsenb{i}{H1}) \land (T_i \in \incs{j}{H2}) \land (T_k \in \haffset{j}{H2})$. Thus from \lemref{cds-tk-h1}, we get that $T_k \in \txns{H1}$. Now, this case is similar to the above case. It can be seen that Equations \ref{eq:tk-h1}, \ref{eq:tk-h2}, \ref{eq:its-h1-h2}, \ref{eq:its-ij} hold good in this case as well. 
	
	Since $T_k$ is in $\haffset{j}{H2}$, we get that \\
	$\htits{k}{H2} < \htits{i}{H2} + 2*L \xrightarrow{\eqnref{its-h1-h2}} \htits{k}{H1} < \htits{j}{H2} + 2*L \xrightarrow{\eqnref{its-ij}} \htits{k}{H1} < \htits{i}{H1} + 2*L $
	
	\noindent Combining this result with \eqnref{tk-h1}, we get that $T_k \in \haffset{i}{H1}$.
\end{proof}

\noindent Next we explore how a \cdsen transaction remains \cdsen in the future histories once it becomes true.

\begin{lemma}
	\label{lem:cds-fut}
	Consider two histories $H1$ and $H2$ with $H2$ being an extension of $H1$. Let  $T_i$ and $T_j$ be two transactions which are live in $H1$ and $H2$ respectively. Let $T_i$ be an \inc of $T_j$ and $\tcts{i}$ is less than $\tcts{j}$. Suppose $T_i$ is \cdsen in $H1$. Then $T_j$ is \cdsen in $H2$ as well. Formally, $\langle H1, H2, T_i, T_j: (H1 \sqsubseteq H2) \land (T_i \in \live{H1}) \land (T_j \in \live{H2}) \land (T_i \in \incs{j}{H2}) \land (\htcts{i}{H1} < \htcts{j}{H2}) \land (\cdsenb{i}{H1}) \implies (\cdsenb{j}{H2}) \rangle$. 
\end{lemma}

\begin{proof}
	We have that $T_i$ is live in $H1$ and $T_j$ is live in $H2$. Since $T_i$ is \cdsen in $H1$, we get (from the definition of \cdsen) that 
	\begin{equation}
		\label{eq:cts-its}
		\htcts{i}{H1} \geq \htits{i}{H2} + 2*L
	\end{equation}
	
	We are given that $\tcts{i}$ is less than $\tcts{j}$ and $T_i, T_j$ are incarnations of each other. Hence, we have that
	
	\ignore{
		\begin{align*}
			\htcts{j}{H2} & > \htcts{i}{H1} \\
			\htcts{j}{H2} & > \htits{i}{H1} + 2*L & [\text{From \eqnref{cts-its}}] \\
			\htcts{j}{H2} & > \htits{j}{H2} + 2*L & [\tits{i} = \tits{j}] \\
		\end{align*}
	}
	
	\begin{align*}
		\htcts{j}{H2} & > \htcts{i}{H1} \\
		& > \htits{i}{H1} + 2*L & [\text{From \eqnref{cts-its}}] \\
		& > \htits{j}{H2} + 2*L & [\tits{i} = \tits{j}] \\
	\end{align*}

	Thus we get that $\tcts{j} > \tits{j} + 2*L$. We have that $T_j$ is live in $H2$. In order to show that $T_j$ is \cdsen in $H2$, it only remains to show that \cdset of $T_j$ in $H2$ is empty, i.e., $\hcds{j}{H2} = \phi$. The \cdset becomes empty when all the transactions of $T_j$'s \affset in $H2$ have their \incct as true in $H2$.
	
	Since $T_j$ is live in $H2$, we get that $T_j$ is in $\txns{H2}$. Here, we have that $(H1 \sqsubseteq H2) \land (T_j \in \txns{H2}) \land (T_i \in \incs{j}{H2}) \land (\cdsenb{i}{H1})$. Combining this with \lemref{aff-same}, we get that $\haffset{i}{H1} = \haffset{j}{H2}$.
	
	Now, consider a transaction $T_k$ in $ \haffset{j}{H2}$. From the above result, we get that $T_k$ is also in $\haffset{i}{H1}$. Since $T_i$ is \cdsen in $H1$, i.e., $\cdsenb{i}{H1}$ is true, we get that $\inct{k}{H1}$ is true. Combining this with \obsref{inct-fut}, we get that $T_k$ must have its \incct as true in $H2$ as well, i.e. $\inct{k}{H2}$. This implies that all the transactions in $T_j$'s \affset have their \incct flags as true in $H2$. Hence the $\hcds{j}{H2}$ is empty. As a result, $T_j$ is \cdsen in $H2$, i.e., $\cdsenb{j}{H2}$. 
\end{proof}

Having defined the properties related to \cdsen, we start defining notions for \finen. Next, we define \emph{\maxwts} for a transaction $T_i$ in $H$ which is the transaction $T_j$ with the largest \wts in $T_i$'s \incset. Formally,
\begin{equation*}
	\hmaxwts{i}{H} = max\{\htwts{j}{H}|(T_j \in \incs{i}{H})\}
\end{equation*}

\noindent From this definition of \maxwts, we get the following simple observation. 

\begin{observation}
	\label{obs:max-wts} 
	For any transaction $T_i$ in $H$, we have that $\twts{i}$ is less than or equal to $\hmaxwts{i}{H}$. Formally, $\htwts{i}{H} \leq \hmaxwts{i}{H}$.
\end{observation}

Next, we combine the notions of \affset and \maxwts to define \emph{\affwts}. It is the maximum of \maxwts of all the transactions in its \affset. Formally, 
\begin{equation*}
	\haffwts{i}{H} = max\{\hmaxwts{j}{H}|(T_j \in \haffset{i}{H})\}
\end{equation*}

\noindent Having defined the notion of \affwts, we get the following lemma relating the \affset and \affwts of two transactions. 

\begin{lemma}
	\label{lem:affwts-same}
	Consider two histories $H1$ and $H2$ with $H2$ being an extension of $H1$. Let  $T_i$ and $T_j$ be two transactions which are live in $H1$ and $H2$ respectively. Suppose the \affset of $T_i$ in $H1$ is same as \affset of $T_j$ in $H2$. Then the \affwts of $T_i$ in $H1$ is same as \affwts of $T_j$ in $H2$. Formally, $\langle H1, H2, T_i, T_j: (H1 \sqsubseteq H2) \land (T_i \in \txns{H1}) \land (T_j \in \txns{H2}) \land (\haffset{i}{H1} = \haffset{j}{H2}) \implies (\haffwts{i}{H1} = \haffwts{j}{H2}) \rangle$. 
\end{lemma}

\begin{proof}
	
	From the definition of \affwts, we get the following equations
	\begin{equation}
		\label{eq:h1-ti-affwts}
		\haffwts{i}{H} = max\{\hmaxwts{k}{H}|(T_k \in \haffset{i}{H1})\}
	\end{equation}
	
	\begin{equation}
		\label{eq:h2-tj-affwts}
		\haffwts{j}{H} = max\{\hmaxwts{l}{H}|(T_l \in \haffset{j}{H2})\}
	\end{equation}
	
	From these definitions, let us suppose that $\haffwts{i}{H1}$ is $\hmaxwts{p}{H1}$ for some transaction $T_p$ in $\haffset{i}{H1}$. Similarly, suppose that $\haffwts{j}{H2}$ is $\hmaxwts{q}{H2}$ for some transaction $T_q$ in $\haffset{j}{H2}$.
	
	Here, we are given that $\haffset{i}{H1} = \haffset{j}{H2})$. Hence, we get that $T_p$ is also in $\haffset{i}{H1}$. Similarly, $T_q$ is in $\haffset{j}{H2}$ as well. Thus from Equations \eqref{eq:h1-ti-affwts} \& \eqref{eq:h2-tj-affwts}, we get that 
	
	\begin{equation}
		\label{eq:ti-tp-max}
		\hmaxwts{p}{H1} \geq  \hmaxwts{q}{H2}
	\end{equation}
	
	\begin{equation}
		\label{eq:tj-tq-max}
		\hmaxwts{q}{H2} \geq \hmaxwts{p}{H1}
	\end{equation}
	
	Combining these both equations, we get that $\hmaxwts{p}{H1} = \hmaxwts{q}{H2}$ which in turn implies that $\haffwts{i}{H1} = \haffwts{j}{H2}$.
	
\end{proof}

\noindent Finally, using the notion of \affwts and \cdsen, we define the notion of \emph{\finen}

\begin{definition}
	\label{defn:finen}
	We say that transaction $T_i$ is \emph{\finen} if the following conditions hold true (1) $T_i$ is live in $H$; (2) $T_i$ is \cdsen is $H$; (3) $\htwts{j}{H}$ is greater than $\haffwts{i}{H}$. Formally, 
	
	\begin{equation*}
		\finenb{i}{H} = \begin{cases}
			True    & (T_i \in \live{H}) \land (\cdsenb{i}{H}) \land (\htwts{j}{H} \\& > \haffwts{i}{H}) \\
			False	& \text{otherwise}
		\end{cases}
	\end{equation*}
\end{definition}
It can be seen from this definition, a transaction that is \finen is also \cdsen. We now show that just like \itsen and \cdsen, once a transaction is \finen, it remains \finen until it terminates. The following lemma captures it. 
\begin{lemma}
	\label{lem:fin-fut}
	Consider two histories $H1$ and $H2$ with $H2$ being an extension of $H1$. Let  $T_i$ and $T_j$ be two transactions which are live in $H1$ and $H2$ respectively. Suppose $T_i$ is \finen in $H1$. Let $T_i$ be an \inc of $T_j$ and $\tcts{i}$ is less than $\tcts{j}$. Then $T_j$ is \finen in $H2$ as well. Formally, $\langle H1, H2, T_i, T_j: (H1 \sqsubseteq H2) \land (T_i \in \live{H1}) \land (T_j \in \live{H2}) \land (T_i \in \incs{j}{H2}) \land (\htcts{i}{H1} < \htcts{j}{H2}) \land (\finenb{i}{H1}) \implies (\finenb{j}{H2}) \rangle$. 
\end{lemma}

\begin{proof}
	Here we are given that $T_j$ is live in $H2$. Since $T_i$ is \finen in $H1$, we get that it is \cdsen in $H1$ as well. Combining this with the conditions given in the lemma statement, we have that, \\ 
	
	\begin{equation}
		\label{eq:fin-given}
		\begin{split}
			\langle (H1 \sqsubseteq H2) \land (T_i \in \live{H1}) \land (T_j \in \live{H2}) \land (T_i \in \incs{j}{H2}) \\\land (\htcts{i}{H1} < \htcts{j}{H2}) 
			\land (\cdsenb{i}{H1}) \rangle
		\end{split}
	\end{equation}
	
	Combining \eqnref{fin-given} with \lemref{cds-fut}, we get that $T_j$ is \cdsen in $H2$, i.e., $\cdsenb{j}{H2}$. Now, in order to show that $T_j$ is \finen in $H2$ it remains for us to show that $\htwts{j}{H2} > \haffwts{j}{H2}$.
	
	We are given that $T_j$ is live in $H2$ which in turn implies that $T_j$ is in $\txns{H2}$. Thus changing this in \eqnref{fin-given}, we get the following 
	\begin{equation}
		\label{eq:mod-given}
		\begin{split}
			\langle (H1 \sqsubseteq H2) \land (T_j \in \txns{H2}) \land (T_i \in \incs{j}{H2}) \land (\htcts{i}{H1} < \htcts{j}{H2}) \\
			\land (\cdsenb{i}{H1}) \rangle
		\end{split}
	\end{equation}
	
	\noindent Combining \eqnref{mod-given} with \lemref{aff-same} we get that 
	\begin{equation}
		\label{eq:affs-eq}
		\haffwts{i}{H1} = \haffwts{j}{H2}
	\end{equation}
	
	\noindent We are given that $\htcts{i}{H1} < \htcts{j}{H2}$. Combining this with the definition of \wts, we get 
	\begin{equation}
		\label{eq:titj-wts}
		\htwts{i}{H1} < \htwts{j}{H2}
	\end{equation}
	
	\noindent Since $T_i$ is \finen in $H1$, we have that \\
	$\htwts{i}{H1} > \haffwts{i}{H1} \xrightarrow{\eqnref{titj-wts}} \htwts{j}{H2} > \haffwts{i}{H1} \xrightarrow{\eqnref{affs-eq}} \htwts{j}{H2} > \\
	\haffwts{j}{H2}$
	
\end{proof}

\noindent Now, we show that a transaction that is \finen will eventually commit. 


\begin{lemma}
	\label{lem:enbd-ct}
	Consider a live transaction $T_i$ in a history $H1$. Suppose $T_i$ is \finen in $H1$ and $\tval{i}$ is true in $H1$. Then there exists an extension of $H1$, $H3$ in which $T_i$ is committed. Formally, $\langle H1, T_i: (T_i \in \live{H1}) \land (\htval{i}{H1}) \land (\finenb{i}{H1}) \implies (\exists H3: (H1 \sqsubset H3) \land (T_i \in \comm{H3})) \rangle$. 
\end{lemma}

\begin{proof}
	Consider a history $H3$ such that its \syst being greater than $\tcts{i} + L$. We will prove this lemma using contradiction. Suppose $T_i$ is aborted in $H3$. 
	
	Now consider $T_i$ in $H1$: $T_i$ is live; its \val flag is true; and is \enbd. From the definition of \finen, we get that it is also \cdsen. From \lemref{its-enb}, we get that $T_i$ is \itsen in $H1$. Thus from \lemref{its-wts}, we get that there exists an extension of $H1$, $H2$ such that (1) Transaction $T_i$ is live in $H2$; (2) there is a transaction $T_j$ in ${H2}$; (3) $\htwts{j}{H2}$ is greater than $\htwts{i}{H2}$; (4) $T_j$ is committed in $H3$. Formally, 
	
	\begin{equation}
		\label{eq:its-wts-ant}
		\begin{split}
			\langle (\exists H2, T_j: (H1 \sqsubseteq H2 \sqsubset H3) \land (T_i \in \live{H2}) \land (T_j \in \txns{H2}) \\\land (\htwts{i}{H2} < \htwts{j}{H2})
			\land (T_j \in \comm{H3})) \rangle
		\end{split}
	\end{equation}
	
	Here, we have that $H2$ is an extension of $H1$ with $T_i$ being live in both of them and $T_i$ is \finen in $H1$. Thus from \lemref{fin-fut}, we get that $T_i$ is \finen in $H2$ as well. Now, let us consider $T_j$ in $H2$. From \eqnref{its-wts-ant}, we get that $(\htwts{i}{H2} < \htwts{j}{H2})$. Combining this with the observation that $T_i$ being live in $H2$, \lemref{wts-its} we get that $(\htits{j}{H2} \leq \htits{i}{H2} + 2*L)$.
	
	
	This implies that $T_j$ is in \affset of $T_i$ in $H2$, i.e., $(T_j \in \haffset{i}{H2})$. From the definition of \affwts, we get that 
	
	\begin{equation}
		\label{eq:max-affwts}
		(\haffwts{i}{H2} \geq \hmaxwts{j}{H2}) 
	\end{equation}
	
	Since $T_i$ is \finen in $H2$, we get that $\twts{i}$ is greater than \affwts of $T_i$ in $H2$. 
	\begin{equation}
		\label{eq:wts-affwts}
		(\htwts{i}{H2} > \haffwts{i}{H2}) 
	\end{equation}
	
	Now combining Equations \ref{eq:max-affwts}, \ref{eq:wts-affwts} we get,
	\begin{align*}
		\htwts{i}{H2} & > \haffwts{i}{H2} \geq \hmaxwts{j}{H2} & \\
		& > \haffwts{i}{H2} \geq \hmaxwts{j}{H2} \geq \htwts{j}{H2}  & [\text{From \obsref{max-wts}}] \\
		& > \htwts{j}{H2}
	\end{align*}
	
	But this equation contradicts with \eqnref{its-wts-ant}. Hence our assumption that $T_i$ will get aborted in $H3$ after getting \finen is not possible. Thus $T_i$ has to commit in $H3$.
\end{proof}
\noindent Next we show that once a transaction $T_i$ becomes \itsen, it will eventually become \finen as well and then committed. We show this change happens in a sequence of steps. We first show that Transaction $T_i$ which is \itsen first becomes \cdsen (or gets committed). We next show that $T_i$ which is \cdsen becomes \finen or get committed. On becoming \finen, we have already shown that $T_i$ will eventually commit. 

Now, we show that a transaction that is \itsen will become \cdsen or committed. To show this, we introduce a few more notations and definitions. We start with the notion of \emph{\depits (dependent-its)} which is the set of \its{s} that a transaction $T_i$ depends on to commit. It is the set of \its of all the transactions in $T_i$'s \cdset in a history $H$. Formally, 

\begin{equation*}
	\hdep{i}{H} = \{\htits{j}{H}|T_j \in \hcds{i}{H}\}
\end{equation*}

\noindent We have the following lemma on the \depits of a transaction $T_i$ and its future \inc $T_j$ which states that \depits of a $T_i$ either reduces or remains the same. 

\begin{lemma}
	\label{lem:depits-fut}
	Consider two histories $H1$ and $H2$ with $H2$ being an extension of $H1$. Let  $T_i$ and $T_j$ be two transactions which are live in $H1$ and $H2$ respectively and $T_i$ is an \inc of $T_j$. In addition, we also have that $\tcts{i}$ is greater than $\tits{i} + 2*L$ in $H1$. Then, we get that $\hdep{j}{H2}$ is a subset of $\hdep{i}{H1}$. Formally, $\langle H1, H2, T_i, T_j: (H1 \sqsubseteq H2) \land (T_i \in \live{H1}) \land (T_j \in \live{H2}) \land (T_i \in \incs{j}{H2}) \land (\htcts{i}{H1} \geq \htits{i}{H1} + 2*L) \implies (\hdep{j}{H2} \subseteq \hdep{i}{H1}) \rangle$. 
\end{lemma}

\begin{proof}
	Suppose $\hdep{j}{H2}$ is not a subset of $\hdep{i}{H1}$. This implies that there is a transaction $T_k$ such that $\htits{k}{H2} \in \hdep{j}{H2}$ but $\htits{k}{H1} \notin \hdep{j}{H1}$. This implies that $T_k$ starts afresh after $H1$ in some history say $H3$ such that $H1 \sqsubset H3 \sqsubseteq H2$. Hence, from \corref{cts-syst} we get the following
	
	\noindent
	\begin{math}
		\htits{k}{H3} > \hsyst{H1} \xrightarrow{\lemref{cts-syst}} \htits{k}{H3} > \htcts{i}{H1} \implies \htits{k}{H3} > \htits{i}{H1} + 2*L \xrightarrow{\htits{i}{H1} = \htits{j}{H2}} \htits{k}{H3} > \htits{j}{H2} + 2*L \xrightarrow[definitions]{\affset, \depits} \htits{k}{H2} \notin \hdep{j}{H2}
	\end{math}
	
	We started with $\tits{k}$ in $\hdep{j}{H2}$ and ended with $\tits{k}$ not in $\hdep{j}{H2}$. Thus, we have a contradiction. Hence, the lemma follows.

\end{proof}

\noindent Next we denote the set of committed transactions in $T_i$'s \affset in $H$ as \emph{\cis (commit independent set)}. Formally, 

\begin{equation*}
	\hcis{i}{H} = \{T_j| (T_j \in \haffset{i}{H}) \land (\inct{j}{H}) \}
\end{equation*}

\noindent In other words, we have that $\hcis{i}{H} = \haffset{i}{H} - \hcds{i}{H}$. Finally, using the notion of \cis we denote the maximum of \maxwts of all the transactions in $T_i$'s \cis as \emph{\pawts} (partly affecting \wts). It turns out that the value of \pawts affects the commit of $T_i$ which we show in the course of the proof. Formally, \pawts is defined as 

\begin{equation*}
	\hpawts{i}{H} = max\{\hmaxwts{j}{H}|(T_j \in \hcis{i}{H})\}
\end{equation*}

\noindent Having defined the required notations, we are now ready to show that a \itsen transaction will eventually become \cdsen. 

\begin{lemma}
	\label{lem:its-cds}
	Consider a transaction $T_i$ which is live in a history $H1$ and $\tcts{i}$ is greater than or equal to $\tits{i} + 2*L$. If $T_i$ is \itsen in $H1$ then there is an extension of $H1$, $H2$ in which an \inc $T_i$, $T_j$ (which could be same as $T_i$), is either committed or \cdsen. Formally, $\langle H1, T_i: (T_i \in \live{H1}) \land (\htcts{i}{H1} \geq \htits{i}{H1} + 2*L) \land (\itsenb{i}{H1}) \implies (\exists H2, T_j: (H1 \sqsubset H2) \land (T_j \in \incs{i}{H2}) \land ((T_j \in \comm{H2}) \lor (\cdsenb{j}{H2}))) \rangle$. 
\end{lemma}

\begin{proof}
	We prove this by inducting on the size of $\hdep{i}{H1}$, $n$. For showing this, we define a boolean function $P(k)$ as follows:
	
	\begin{math}
		P(k) = \begin{cases}
			True & \langle H1, T_i: (T_i \in \live{H1}) \land (\htcts{i}{H1} \geq \htits{i}{H1} + 2*L) \land\\ & (\itsenb{i}{H1}) \land 
			 (k \geq |\hdep{i}{H1}|) \implies \\ & (\exists H2, T_j: (H1 \sqsubset H2) \land (T_j \in \incs{i}{H2}) \land \\
			& ((T_j \in \comm{H2}) \lor (\cdsenb{j}{H2}))) \rangle \\
			False & \text{otherwise}
		\end{cases}
	\end{math}
	
	As can be seen, here $P(k)$ means that if (1) $T_i$ is live in $H1$; (2) $\tcts{i}$ is greater than or equal to $\tits{i} + 2*L$; (3) $T_i$ is \itsen in $H1$ (4) the size of $\hdep{i}{H1}$ is less than or equal to $k$;  then there exists a history $H2$ with a transaction $T_j$ in it which is an \inc of $T_i$ such that $T_j$ is either committed or \cdsen in $H2$. We show $P(k)$ is true for all (integer) values of $k$ using induction. 
	
	\vspace{1mm}
	\noindent
	\textbf{Base Case - $P(0)$:} Here, from the definition of $P(0)$, we get that $|\hdep{i}{H1}| = 0$. This in turn implies that $\hcds{i}{H1}$ is null. Further, we are already given that $T_i$ is live in $H1$ and $\htcts{i}{H1} \geq \htits{i}{H1} + 2*L$. Hence, all these imply that $T_i$ is \cdsen in $H1$. 
	
	\vspace{1mm}
	\noindent
	\textbf{Induction case - To prove $P(k+1)$ given that $P(k)$ is true:} If $|\hdep{i}{H1}| \leq k$, from the induction hypothesis $P(k)$, we get that $T_j$ is either committed or \cdsen in $H2$. Hence, we consider the case when 
	
	\begin{equation}
		\label{eq:hdep-assn}
		|\hdep{i}{H1}| = k + 1
	\end{equation}
	
	Let $\alpha$ be $\hpawts{i}{H1}$. Suppose $\htwts{i}{H1} < \alpha$. Then from \lemref{wts-great}, we get that there is an extension of $H1$, say $H3$ in which an \inc of $T_i$, $T_l$ (which could be same as $T_i$) is committed or is live in $H3$ and has \wts greater than $\alpha$. If $T_l$ is committed then $P(k+1)$ is trivially true. So we consider the latter case in which $T_l$ is live in $H3$. In case $\htwts{i}{H1} \geq \alpha$, then in the analysis below follow where we can replace $T_l$ with $T_i$.
	
	Next, suppose $T_l$ is aborted in an extension of $H3$, $H5$. Then from \lemref{its-wts}, we get that there exists an extension of $H3$, $H4$ in which (1) $T_l$ is live; (2) there is a transaction $T_m$ in $\txns{H4}$; (3) $\htwts{m}{H4} > \htwts{l}{H4}$ (4) $T_m$ is committed in $H5$. 
	
	Combining the above derived conditions (1), (2), (3) with \lemref{ti|tltl-comt} we get that in $H4$, 
	
	\begin{equation}
		\label{eq:ml-tits}
		\htits{m}{H4} \leq \htits{l}{H4} + 2*L
	\end{equation}
	
	\eqnref{ml-tits} implies that $T_m$ is in $T_l$'s \affset. Here, we have that $T_l$ is an \inc of $T_i$ and we are given that $\htcts{i}{H1} \geq \htits{i}{H1} + 2*L$. Thus from \lemref{aff-tkinc-h1}, we get that there exists an \inc of $T_m$, $T_n$ in $H1$.
	
	Combining \eqnref{ml-tits} with the observations (a) $T_n, T_m$ are \inc{s}; (b) $T_l, T_i$ are \inc{s}; (c) $T_i, T_n$ are in $\txns{H1}$, we get that $\htits{n}{H1} \leq \htits{i}{H1} + 2*L$. This implies that $T_n$ is in $\haffset{i}{H1}$. Since $T_n$ is not committed in $H1$ (otherwise, it is not possible for $T_m$ to be an \inc of $T_n$), we get that $T_n$ is in $\hcds{i}{H1}$. Hence, we get that $\htits{m}{H4} = \htits{n}{H1}$ is in $\hdep{i}{H1}$.  
	
	From \eqnref{hdep-assn}, we have that $\hdep{i}{H1}$ is $k+1$. From \lemref{depits-fut}, we get that $\hdep{i}{H4}$ is a subset of $\hdep{i}{H1}$. Further, we have that transaction $T_m$ has committed. Thus $\htits{m}{H4}$ which was in $\hdep{i}{H1}$ is no longer in $\hdep{i}{H4}$. This implies that $\hdep{i}{H4}$ is a strict subset of $\hdep{i}{H1}$ and hence\\ $|\hdep{i}{H4}| \leq k$. 
	
	\noindent Since $T_i$ and $T_l$ are \inc{s}, we get that $\hdep{i}{H4} = \hdep{l}{H1}$. Thus, we get that 
	
	\begin{equation}
		\label{eqn:hdep-ilh4}
		|\hdep{i}{H4}| \leq k \implies |\hdep{l}{H4}| \leq k
	\end{equation}
	
	\noindent Further, we have that $T_l$ is a later \inc of $T_i$. So, we get that
	
	\begin{equation}
		\label{eqn:cts-its}
		\htcts{l}{H4} > \htcts{i}{H4} \xrightarrow{given} \htcts{l}{H4} > \htits{i}{H4} + 2*L \xrightarrow{\htits{i}{H4} = \htits{l}{H4}} \htcts{l}{H4} > \htits{l}{H4} + 2*L
	\end{equation}
	
	We also have that $T_l$ is live in $H4$. Combining this with Equations \ref{eqn:hdep-ilh4}, \ref{eqn:cts-its} and given the induction hypothesis that $P(k)$ is true, we get that there exists a history extension of $H4$, $H6$ in which an \inc of $T_l$ (also $T_i$), $T_p$ is either committed or \cdsen. This proves the lemma.
\end{proof}

\begin{lemma}
	\label{lem:cds-fin}
	Consider a transaction $T_i$ in a history $H1$. If $T_i$ is \cdsen in $H1$ then there is an extension of $H1$, $H2$ in which an \inc $T_i$, $T_j$ (which could be same as $T_i$), is either committed or \finen. Formally, $\langle H1, T_i: (T_i \in \live{H}) \land (\cdsenb{i}{H1}) \implies (\exists H2, T_j: (H1 \sqsubset H2) \land (T_j \in \incs{i}{H2}) \land ((T_j \in \comm{H2}) \lor (\finenb{j}{H2})) \rangle$. 
\end{lemma}

\begin{proof}
	In $H1$, suppose $\haffwts{i}{H1}$ is $\alpha$. From \lemref{wts-great}, we get that there is a extension of $H1$, $H2$ with a transaction $T_j$ which is an \inc of $T_i$. Here there are two cases: (1) Either $T_j$ is committed in $H2$. This trivially proves the lemma; (2) Otherwise, $\twts{j}$ is greater than $\alpha$. 
	
	\noindent In the second case, we get that 
	
	\begin{equation}
		\label{eq:ext}
		\begin{split}
			(T_i \in \live{H1}) \land (T_j \in \live{H2}) \land (\cdsenb{i}{H}) \land (T_j \in \incs{i}{H2}) \land \\
			(\htwts{i}{H1} < \htwts{j}{H2})
		\end{split}
	\end{equation}
	
	\noindent Combining the above result with \lemref{cts-wts}, we get that $\htcts{i}{H1} < \htcts{j}{H2}$. Thus the modified equation is 
	
	\begin{equation}
		\label{eq:new-ext}
		\begin{split}
			(T_i \in \live{H1}) \land (T_j \in \live{H2}) \land (\cdsenb{i}{H1}) \land (T_j \in \incs{i}{H2}) \land \\
			(\htcts{i}{H1} < \htcts{j}{H2})
		\end{split}
	\end{equation}
	
	\noindent Next combining \eqnref{new-ext} with \lemref{aff-same}, we get that 
	\begin{equation}
		\label{eq:affs-ij}
		\haffset{i}{H1} = \haffset{j}{H2}
	\end{equation}
	
	\noindent Similarly, combining \eqnref{new-ext} with \lemref{cds-fut} we get that $T_j$ is \cdsen in $H2$ as well. Formally, 
	\begin{equation}
		\label{eq:th-cdsen}
		\cdsenb{j}{H2}
	\end{equation}
	
	Now combining \eqnref{affs-ij} with \lemref{affwts-same}, we get that 
	\begin{equation}
		\label{eq:affwts-same}
		\haffwts{i}{H1} = \haffwts{j}{H2}
	\end{equation}
	
	From our initial assumption we have that $\haffwts{i}{H1}$ is $\alpha$. From \eqnref{affwts-same}, we get that $\haffwts{j}{H2} = \alpha$. Further, we had earlier also seen that $\htwts{j}{H2}$ is greater than $\alpha$. Hence, we have that $\htwts{j}{H2} > \haffwts{j}{H2}$. 
	
	\noindent Combining the above result with \eqnref{th-cdsen}, $\cdsenb{j}{H2}$, we get that $T_j$ is \finen, i.e., $\finenb{j}{H2}$. 
\end{proof}

\noindent Next, we show that every live transaction eventually become \itsen. 

\begin{lemma}
	\label{lem:live-its}
	Consider a history $H1$ with $T_i$ be a transaction in $\live{H1}$. Then there is an extension of $H1$, $H2$ in which an \inc of $T_i$, $T_j$ (which could be same as $T_i$) is either committed or is \itsen. Formally, $\langle H1, T_i: (T_i\in \live{H}) \implies (\exists T_j, H2: (H1 \sqsubset H2) \land (T_j \in \incs{i}{H2}) \land (T_j \in \comm{H2}) \lor (\itsenb{i}{H})) \rangle$. 
\end{lemma}

\begin{proof}
	\noindent We prove this lemma by inducting on \its. 
	
	\vspace{1mm}
	\noindent
	\textbf{Base Case - $\tits{i} = 1$:} In this case, $T_i$ is the first transaction to be created. There are no transactions with smaller \its. Thus $T_i$ is trivially \itsen. 
	
	\vspace{1mm}
	\noindent
	\textbf{Induction Case:} Here we assume that for any transaction $\tits{i} \leq k$ the lemma is true. 
\end{proof}

Combining these lemmas gives us the result that for every live transaction $T_i$ there is an incarnation $T_j$ (which could be the same as $T_i$) that will commit. This implies that every \aptr eventually commits. The follow lemma captures this notion.

\begin{theorem}
	\label{thm:hwtm-com}
	Consider a history $H1$ with $T_i$ be a transaction in $\live{H1}$. Then there is an extension of $H1$, $H2$ in which an \inc of $T_i$, $T_j$ is committed. Formally, $\langle H1, T_i: (T_i\in \live{H}) \implies (\exists T_j, H2: (H1 \sqsubset H2) \land (T_j \in \incs{i}{H2}) \land (T_j \in \comm{H2})) \rangle$. 
\end{theorem}

\begin{proof}
	Here we show the states that a transaction $T_i$ (or one of it its \inc{s}) undergoes before it commits. In all these transitions, it is possible that an \inc of $T_i$ can commit. But to show the worst case, we assume that no \inc of $T_i$ commits. Continuing with this argument, we show that finally an \inc of $T_i$ commits. 
	
	Consider a live transaction $T_i$ in $H1$. Then from \lemref{live-its}, we get that there is a history $H2$, which is an extension of $H1$, in which $T_j$ an \inc of $T_i$ is either committed or \itsen. If $T_j$ is \itsen in $H2$, then from \lemref{its-cds}, we get that $T_k$, an \inc of $T_j$, will be \cdsen in a extension of $H2$, $H3$ (assuming that $T_k$ is not committed in $H3$). 
	
	From \lemref{cds-fin}, we get that there is an extension of $H3$, $H4$ in which an \inc of $T_k$, $T_l$ will be \finen assuming that it is not committed in $H4$. Finally, from \lemref{enbd-ct}, we get that there is an extension of $H4$ in which $T_m$, an \inc of $T_l$, will be committed. This proves our theorem.
\end{proof}

\noindent From this theorem, we get the following corollary which states that any history generated by \ksftm is $\stfdm$.

\begin{corollary}
	\ksftm algorithm ensures \stfdm. 
\end{corollary}

\vspace{-.3cm}
\section{Experimental Evaluation}
\label{sec:exp}
\vspace{-.2cm}
This section represents the experimental analysis of variants of the proposed Starvation-Free Object-based STMs (SF-SVOSTM, SF-MVOSTM, SF-MVOSTM-GC, and SF-KOSTM)\footnote{Code is available here: https://github.com/PDCRL/SF-MVOSTM.} for two data structure \emph{hash table} (HT-SF-SVOSTM, HT-SF-MVOSTM, HT-SF-MVOSTM-GC and HT-SF-KOSTM)  and \emph{linked-list} (list-SF-SVOSTM, list-SF-MVOSTM, list-SF-MVOSTM-GC and list-SF-KOSTM) implemented in C++. We analyzed that HT-SF-KOSTM and  list-SF-KOSTM perform best among all the proposed algorithms. So, we compared our HT-SF-KOSTM with hash table based state-of-the-art STMs HT-KOSTM \cite{Juyal+:SSS:2018}, HT-SVOSTM \cite{Peri+:OSTM:Netys:2018}, 
ESTM \cite{Felber+:ElasticTrans:2017:jpdc}, RWSTM \cite [Chap. 4]{ WeiVoss:TIS:2002:Morg}, HT-MVTO \cite{Kumar+:MVTO:ICDCN:2014} 
and our list-SF-KOSTM with list based state-of-the-art STMs list-KOSTM \cite{Juyal+:SSS:2018}, list-SVOSTM \cite{Peri+:OSTM:Netys:2018}, Trans-list \cite{ZhangDech:LockFreeTW:SPAA:2016}, Boosting-list \cite{HerlihyKosk:Boosting:PPoPP:2008}, NOrec-list \cite{Dalessandro+:NoRec:PPoPP:2010}, list-MVTO \cite{Kumar+:MVTO:ICDCN:2014}, list-KSFTM \cite{Chaudhary+:KSFTM:Corr:2017}.  

\noindent
\textbf{Experimental Setup:} The system configuration for experiments is 2 socket Intel(R) Xeon(R) CPU E5-2690 v4 @ 2.60GHz with 14 cores per socket and 2 hyper-threads per core, a total of 56 threads. A private 32KB L1 cache and 256 KB L2 cache is with each core. It has 32 GB RAM with Ubuntu 16.04.2 LTS running Operating System. Default scheduling algorithm of Linux with all threads have the same base priority is used in our experiments. This satisfies \asmref{bdtm}(bounded-termination) of the scheduler and we ensure the absence of parasitic transactions for our setup to satisfy \asmref{self}. 

\begin{figure}
	\includegraphics[width=12cm, height=5cm]{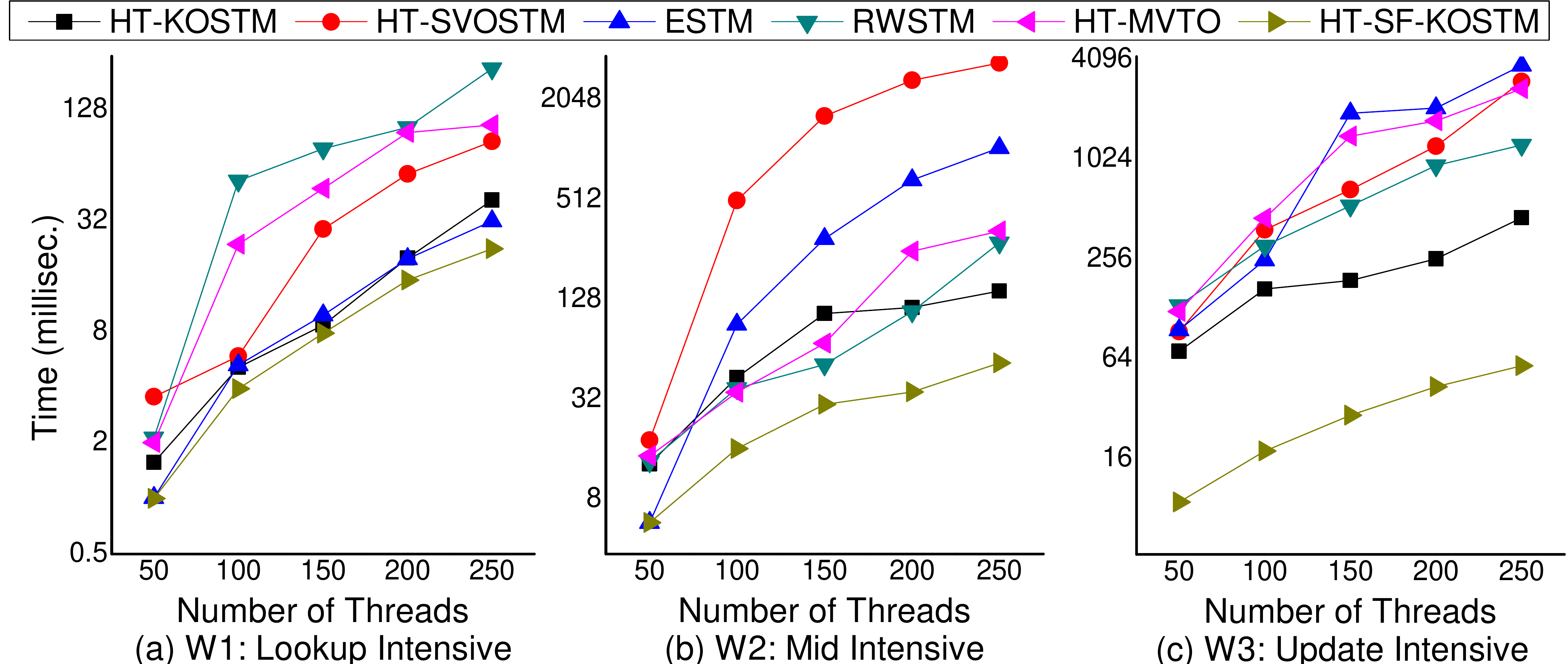}
\vspace{-.3cm}	\caption{ Performance analysis of SF-KOSTM and State-of-the-art STMs on hash table}
	\label{fig:htcomp}
\end{figure}

\noindent	
\textbf{Methodology:} 
 We have considered three different types of workloads namely, W1 (Lookup Intensive - 5\% insert, 5\% delete, and 90\% lookup), W2 (Mid Intensive - 25\% insert, 25\% delete, and 50\% lookup), and W3 (Update Intensive - 45\% insert, 45\% delete, and 10\% lookup). To analyze the absolute benefit of starvation-freedom, we used a customized application called as the \emph{Counter Application} (refer the pseudo-code in \apnref{countercode}) which provides us the flexibility to create a high contention environment where the probability of transactions undergoing starvation on an average is very high. Our \emph{high contention} environment includes only 30 shared data-items (or keys), number of threads ranging from 50 to 250, each thread spawns upon a transaction, where each transaction performs 10 operations depending upon the workload chosen. To study starvation-freedom of various algorithms, we have used \emph{max-time} which is the maximum time required by a transaction to finally commit from its first incarnation, which also involves time taken by all its aborted incarnations. We perform each of our experiments 10 times and consider the average of it to avoid the effect of outliers. 



\vspace{-.05cm}
\noindent
\textbf{Results Analysis:} All our results reflect the same ideology as proposed showcasing the benefits of Starvation-Freedom in Multi-Version OSTMs. We started our experiments with \emph{hash table} data structure of bucket size 5 and 
compared \emph{max-time} for a transaction to commit by proposed HT-SF-KOSTM (best among all the proposed algorithms shown in \figref{htus} of \apnref{ap-result}) with hash table based state-of-the-art STMs. HT-SF-KOSTM achieved an average speedup of 3.9x, 32.18x, 22.67x, 10.8x and 17.1x over HT-KOSTM, HT-SVOSTM, ESTM, RWSTM and HT-MVTO respectively as shown in \figref{htcomp}.


We further considered another data structure \emph{linked-list} and 
compared \emph{max-time} for a transaction to commit by proposed list-SF-KOSTM (best among all the proposed algorithms shown in \figref{listus} of \apnref{ap-result}) with list based state-of-the-arts STMs. 
list-SF-KOSTM achieved an average  speedup of 2.4x, 10.6x, 7.37x, 36.7x, 9.05x, 14.47x, and 1.43x over list-KOSTM, list-SVOSTM, Trans-list, Boosting-list, NOrec-list, list-MVTO and list-KSFTM respectively as shown in \figref{listcomp}. We consider number of versions in the version list $K$ as 5 and value of C as 0.1. 

 
 \begin{figure}
 	\includegraphics[width=12cm, height=5cm]{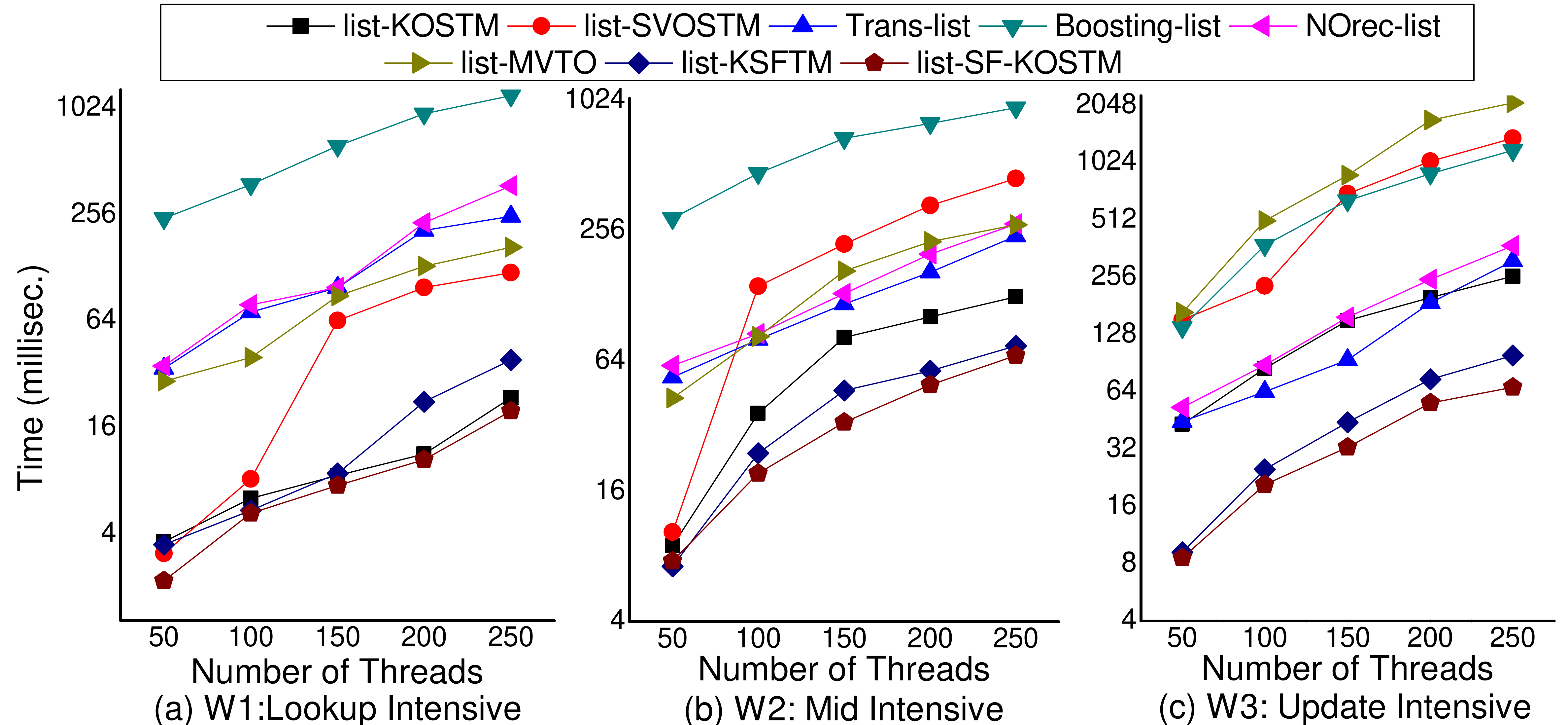}
\vspace{-.3cm} 	\caption{ Performance analysis of SF-KOSTM and State-of-the-art STMs on list}
 	\label{fig:listcomp}
 \end{figure}

\begin{figure}[H]
	\includegraphics[width=12cm, height=5cm]{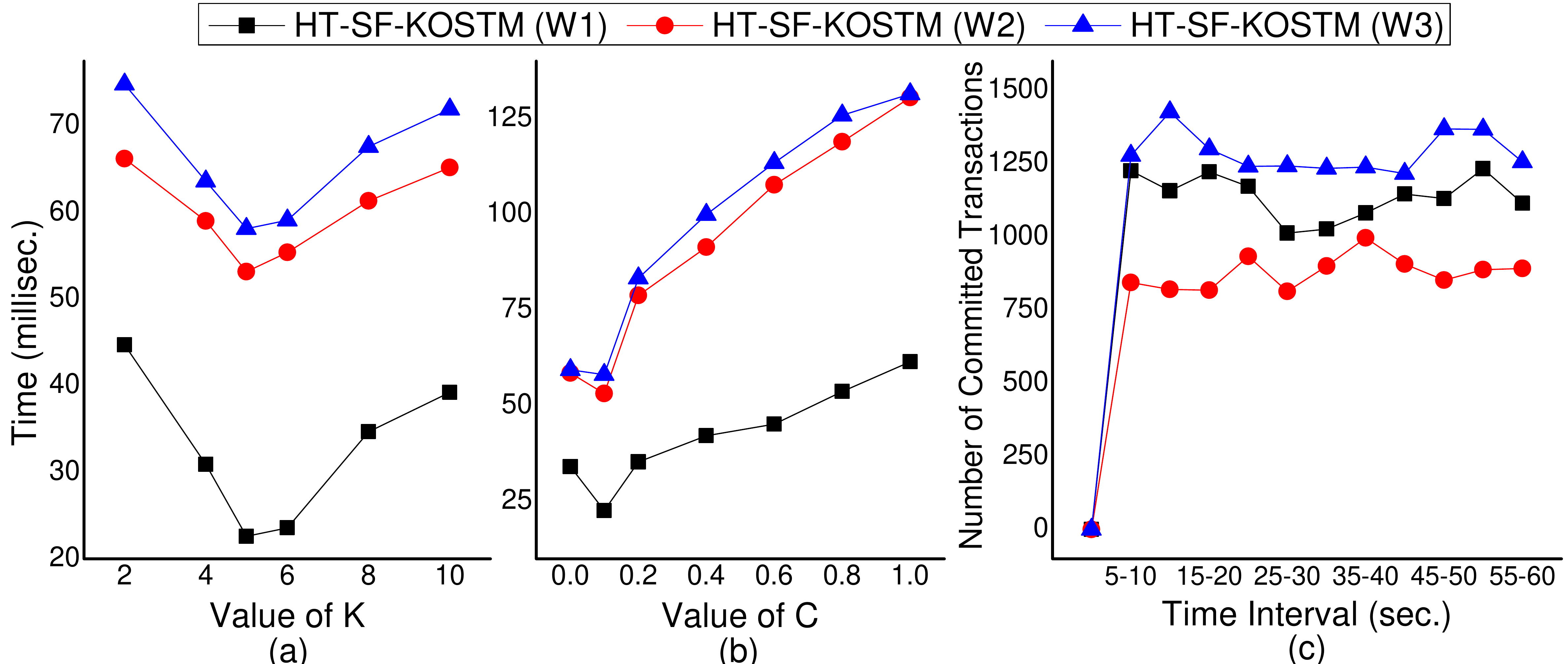}
	\caption{ Optimal value of K and C along with Stability for hash table}
	\label{fig:htck}
\end{figure}

\noindent
\textbf{Best value of $K$, $C$, and Stability in SF-KOSTM:} We identified the  best value of $K$ for both HT-SF-KOSTM and list-SF-KOSTM algorithms. \figref{htck}.(a). demonstrates the best value of $K$ as 5 for HT-SF-KOSTM on counter application. We achieve this while varying value of $K$ on \emph{high contention environment} with 64 threads on all the workloads $W1, W2, W3$. \figref{htck}.(b). illustrates the best value of $C$ as 0.1 for HT-SF-KOSTM on all the workloads $W1, W2, W3$. \figref{htck}.(c). represents the $stability$ of HT-SF-KOSTM algorithm overtime for the counter application. For this experiment, we fixed 32 threads, 1000 shared data-items (or keys), the value of $K$ as 5, and $C$ as 0.1. Along with this, we consider, 5 seconds warm-up period on all the workloads $W1, W2, W3$. Each thread invokes transactions until its time-bound of 60 seconds expires.
We calculate the number of transactions committed in the incremental interval of 5 seconds. \figref{htck}.(c). shows that over time HT-SF-KOSTM is stable which helps to hold the claim that the performance of HT-SF-KOSTM will continue in the same manner if time is increased to higher orders. Similarly, we perform the same experiments for the linked-list data structure as well. \figref{listck}.(a). and \figref{listck}.(b). demonstrate the best value of $K$ as 5 and $C$ as 0.1 for list-SF-KOSTM on all the workloads $W1, W2, W3$. Similarly, \figref{listck}.(c). illustrates the $stability$ of list-SF-KOSTM and shows that it is stable over time on all the workloads $W1, W2, W3$.

\begin{figure}[H]
	\includegraphics[width=12cm, height=5cm]{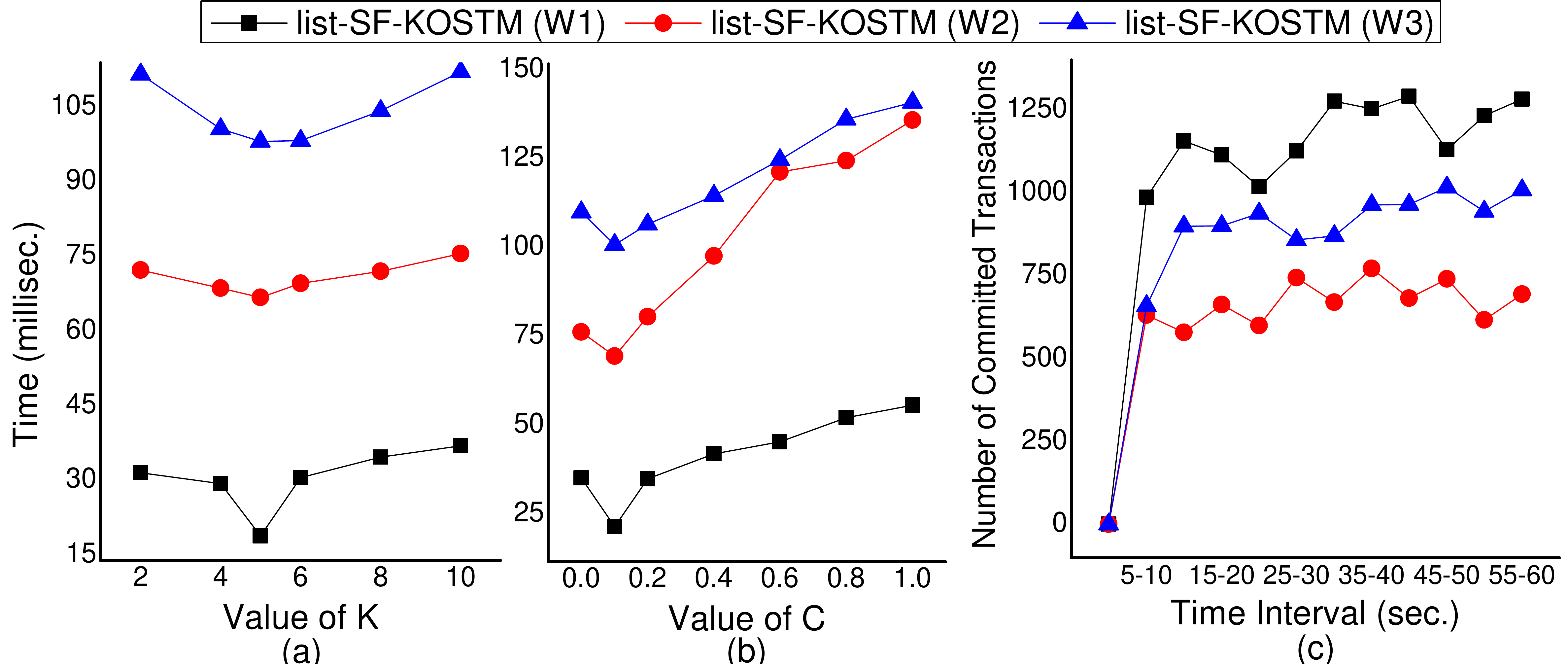}
	\caption{ Optimal value of K and C along with Stability for list}
	\label{fig:listck}
\end{figure}
\begin{figure}[H]
	\includegraphics[width=12cm, height=5cm]{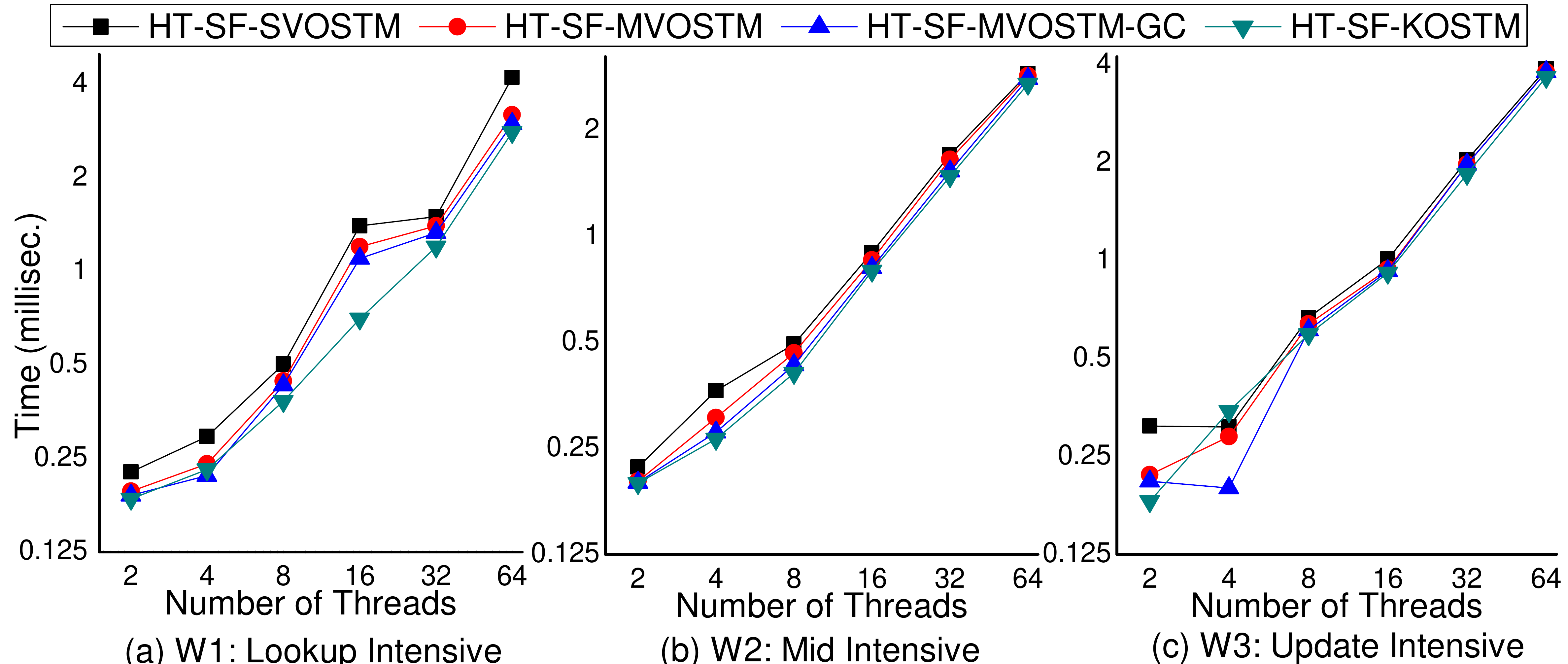}
	\caption{Time comparison among SF-SVOSTM and variants of SF-KOSTM on hash table}
	\label{fig:httimeus}
\end{figure}

We have done some experiments on \textit{low contention environment} which involves 1000 keys, threads varying from 2 to 64 in power of 2, each thread spawns one transaction and each transaction executes 10 operations (insert, delete and lookup) depending upon the workload chosen. We observed that HT-SF-KOSTM performs best out of all the proposed algorithms on $W1$ and $W2$ workload as shown in \figref{httimeus}. For a lesser number of threads on $W3$, HT-SF-KOSTM was taking a bit more time than other proposed algorithms as shown in \figref{httimeus}.(c). This may be because of the finite version, finding and replacing the oldest version is taking time. After that, we consider  HT-SF-KOSTM and compared against state-of-the-art STMs.  \figref{httimecomp} shows that our proposed algorithm performs better than all other state-of-art-STMs algorithms but slightly lesser than the non starvation-free HT-KOSTM. But to provide the guarantee of starvation-freedom this slight slag of time is worth paying. For the better clarification of speedups, please refer to \tabref{htspeed}.  Similarly, \figref{listtimeus} represents the analysis of \emph{low contention environment} for list data structure where list-SF-KOSTM performs best out of all the proposed algorithms. \figref{listtimecomp} demonstrates the comparison of proposed list-SF-KOSTM with list based state-of-the-art STMs and shows the significant performance gain in terms of a speedup as presented in \tabref{listspeed}. For low contention environment, starvation-freedom is appearing as an overhead so, both HT-SF-KOSTM and list-SF-KOSTM  achieve a bit less speedup than HT-KOSTM and list-KOSTM. But for high contention environment, starvation-free algorithms are always better so, both HT-SF-KOSTM and list-SF-KOSTM achieve better speedup than HT-KOSTM and list-KOSTM as explained in \secref{exp}.

\begin{figure}[H]
	\includegraphics[width=12cm, height=5cm]{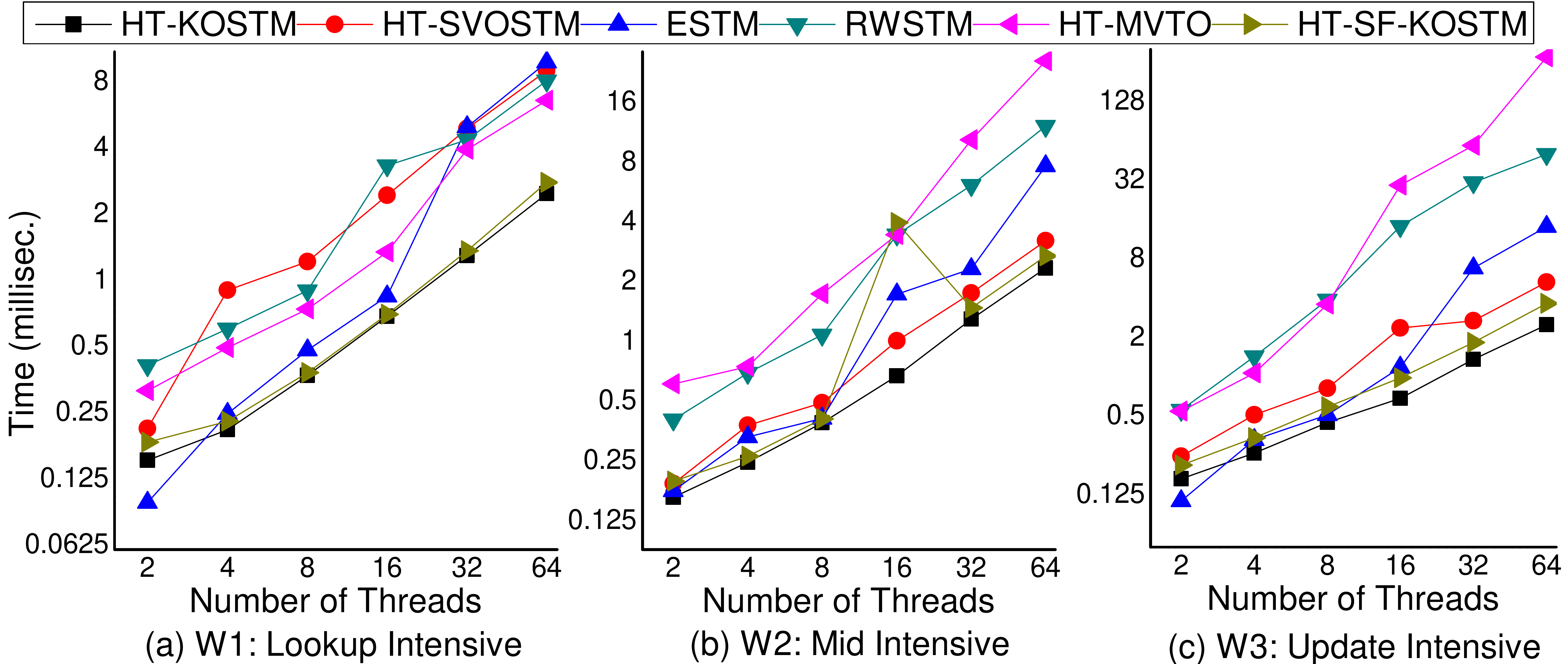}
	\caption{Time comparison of SF-KOSTM and State-of-the-art STMs on hash table}
	\label{fig:httimecomp}
\end{figure}

\begin{figure}[H]
	\includegraphics[width=12cm, height=5cm]{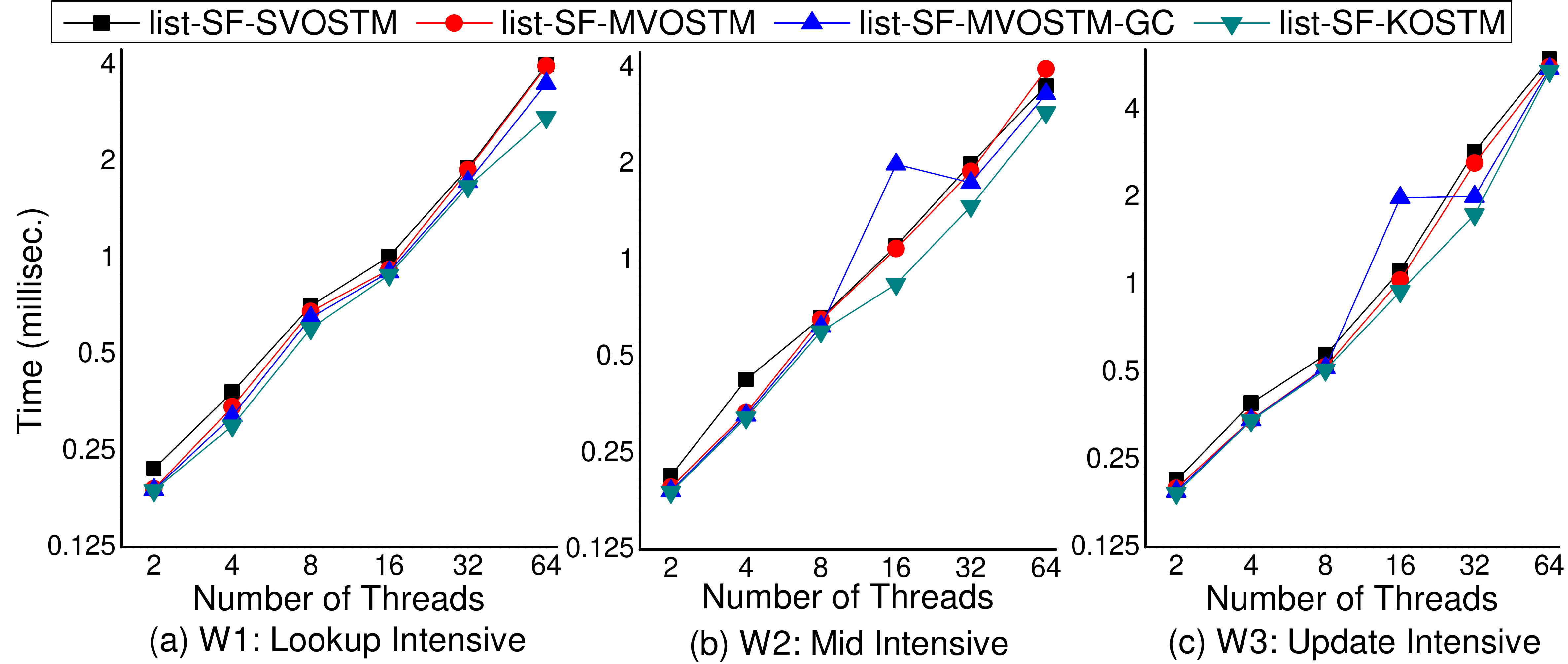}
	\caption{Time comparison among SF-SVOSTM and variants of SF-KOSTM on list}
	\label{fig:listtimeus}
\end{figure}
\begin{figure}[H]
	\includegraphics[width=12cm, height=5cm]{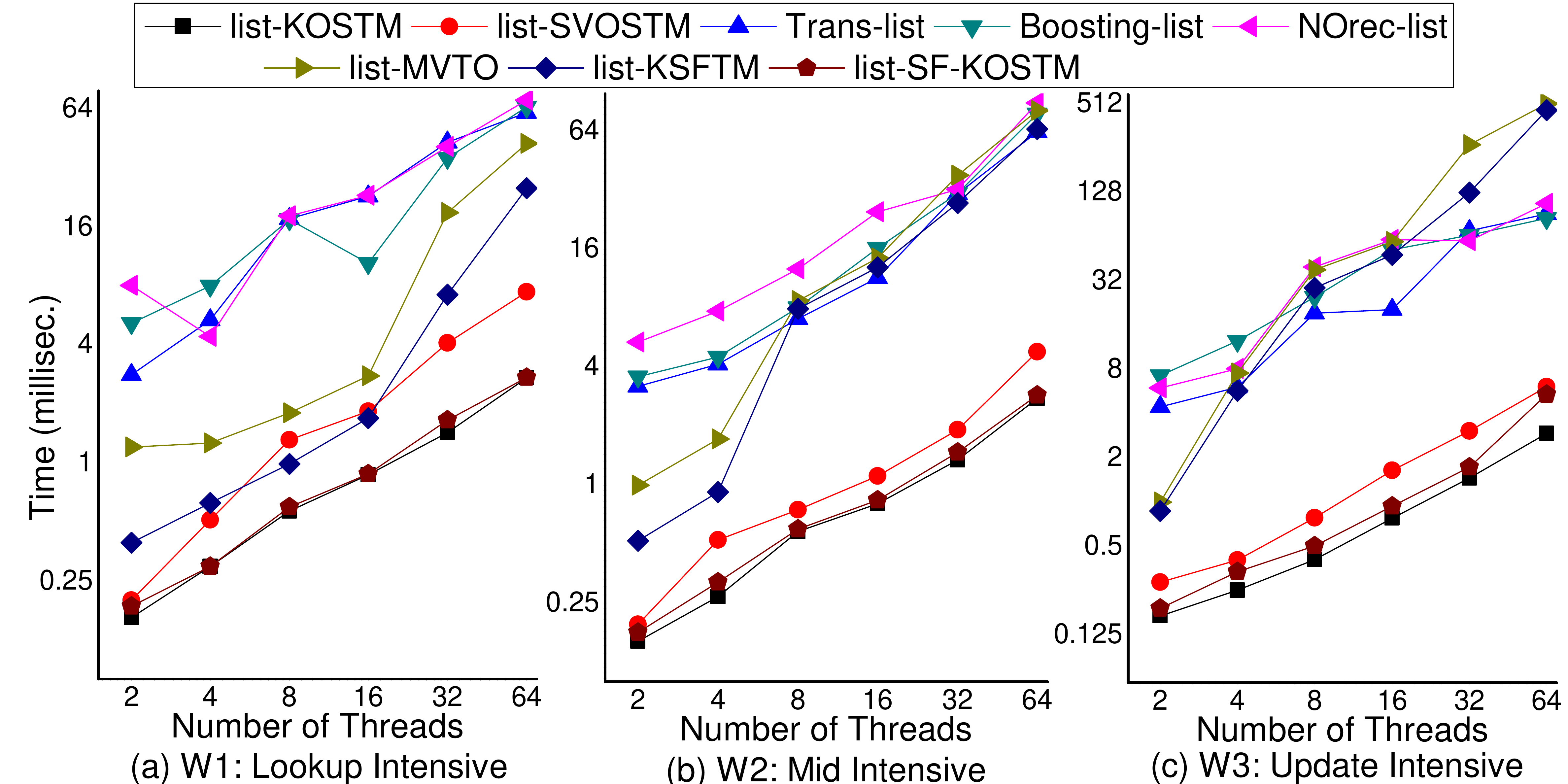}
	\caption{Time comparison of SF-KOSTM and State-of-the-art STMs on list}
	\label{fig:listtimecomp}
\end{figure}

\begin{figure}[H]
	\includegraphics[width=12cm, height=5cm]{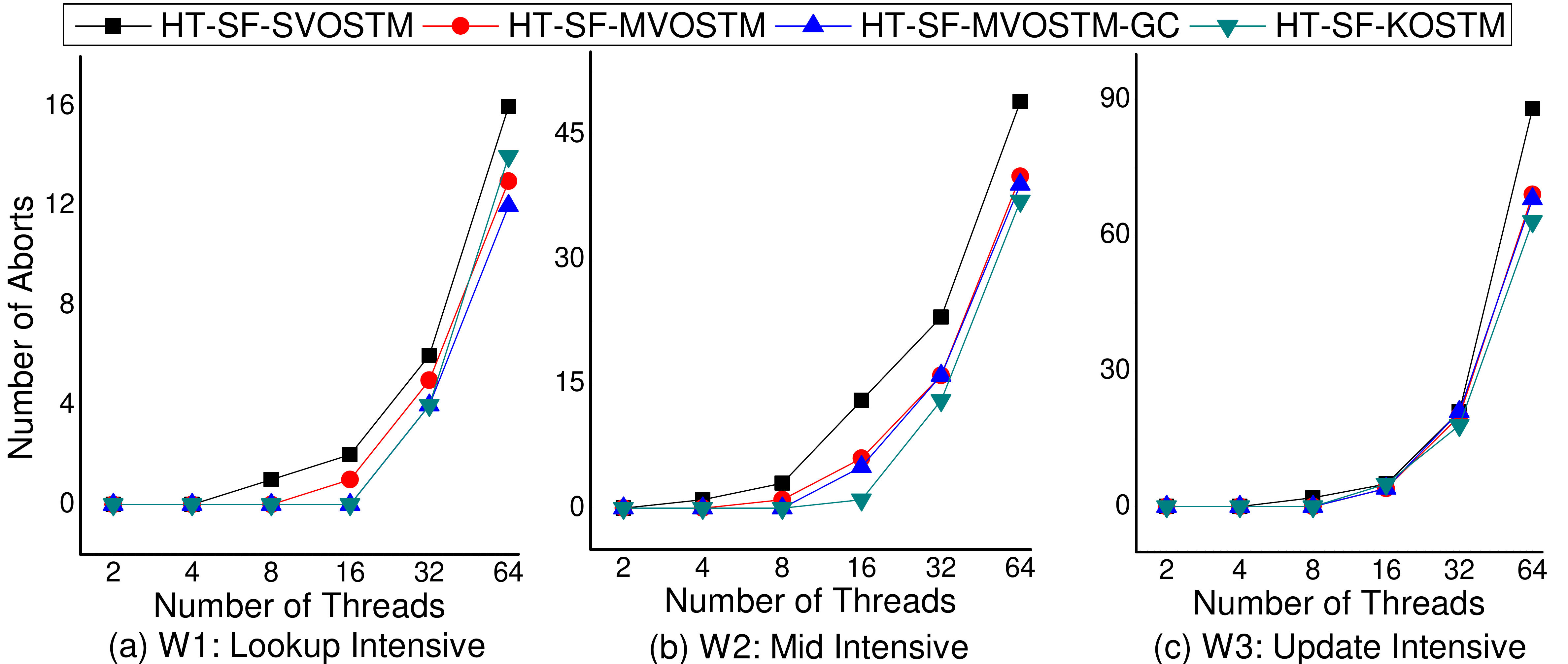}
	\caption{Abort Count of SF-SVOSTM and variants of SF-KOSTM on hash table}
	\label{fig:htabortus}
\end{figure}
\begin{figure}[H]
	\includegraphics[width=12cm, height=5cm]{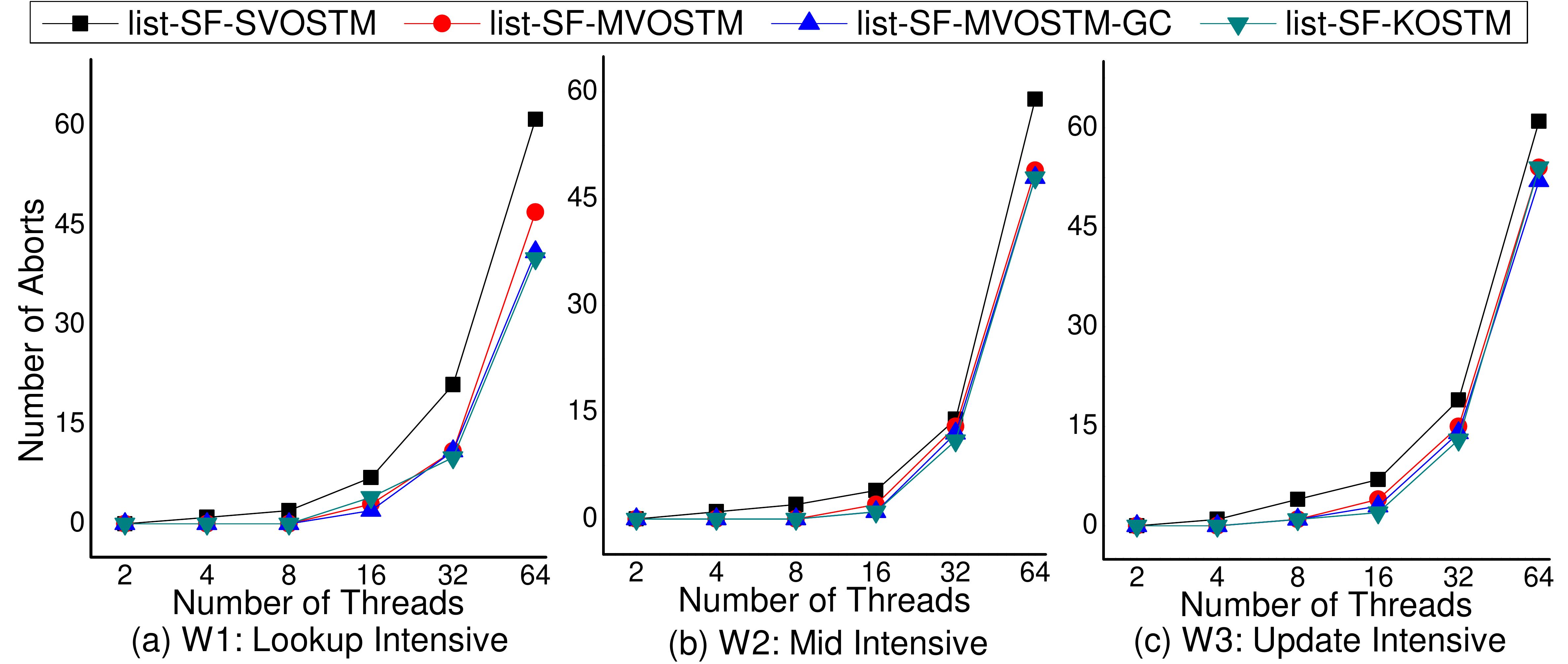}
	\caption{Abort Count of SF-SVOSTM and variants of SF-KOSTM on list}
	\label{fig:listabortus}
\end{figure}

\noindent
\textbf{Abort Count:} We analyzed the number of aborts on low contention environment as defined above. \figref{htabortus} and \figref{listabortus} show the number of aborts comparison among all the proposed variants in all three workloads W1, W2 and W3 for both data structures hash table and linked-list. The results show that HT-SF-KOSTM and list-SF-KOSTM have relatively less  number of aborts than other proposed algorithms. Similarly, \figref{htabortcomp} and \figref{listabortcomp} shows the number of aborts comparison among proposed HT-SF-KOSTM with hash table based state-of-the-art STMs and proposed list-SF-KOSTM with list based state-of-the-art STMs in all three workloads W1, W2, and W3. The result shows that the least number of aborts are happening with HT-SF-KOSTM and list-SF-KOSTM.
\noindent


\begin{figure}[H]
	\includegraphics[width=12cm, height=5cm]{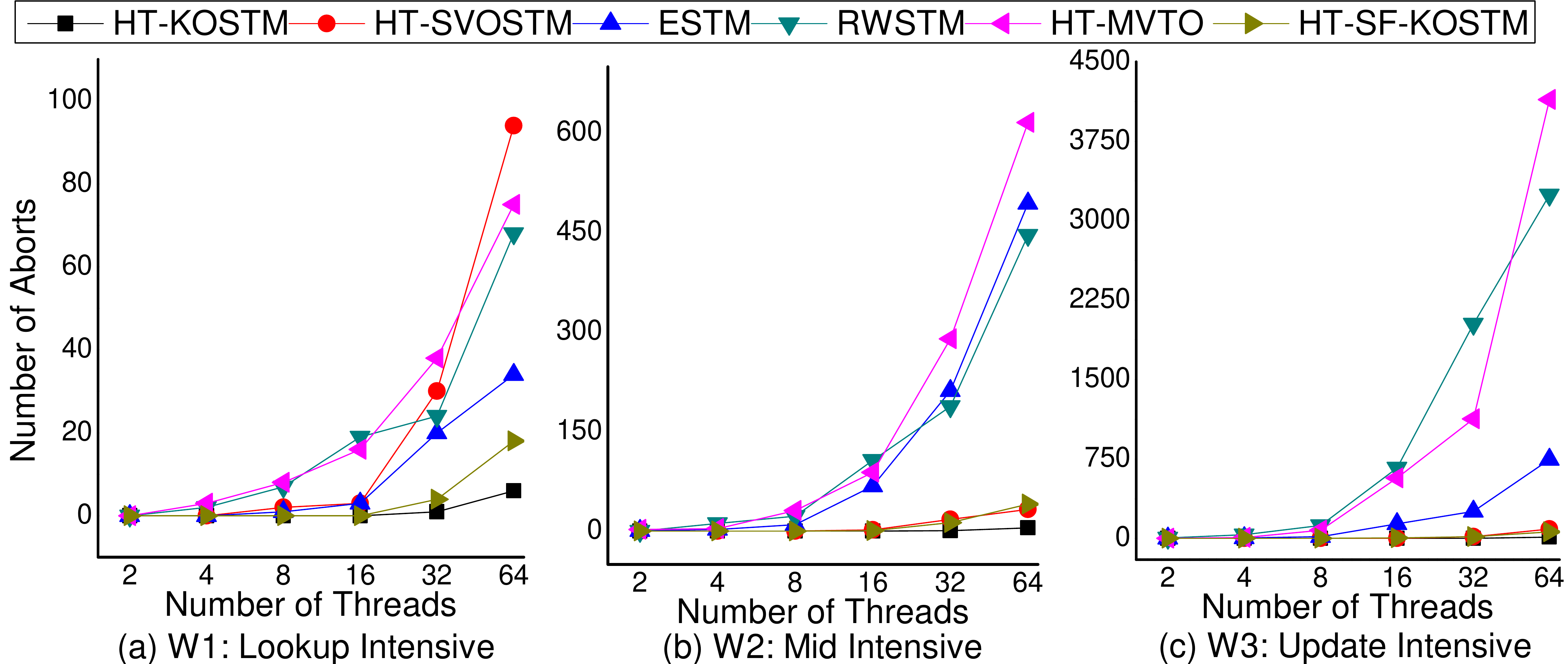}
	\caption{Abort Count of SF-KOSTM and State-of-the-art STMs on hash table}
	\label{fig:htabortcomp}
\end{figure}

\begin{figure}[H]
	\includegraphics[width=12cm, height=5cm]{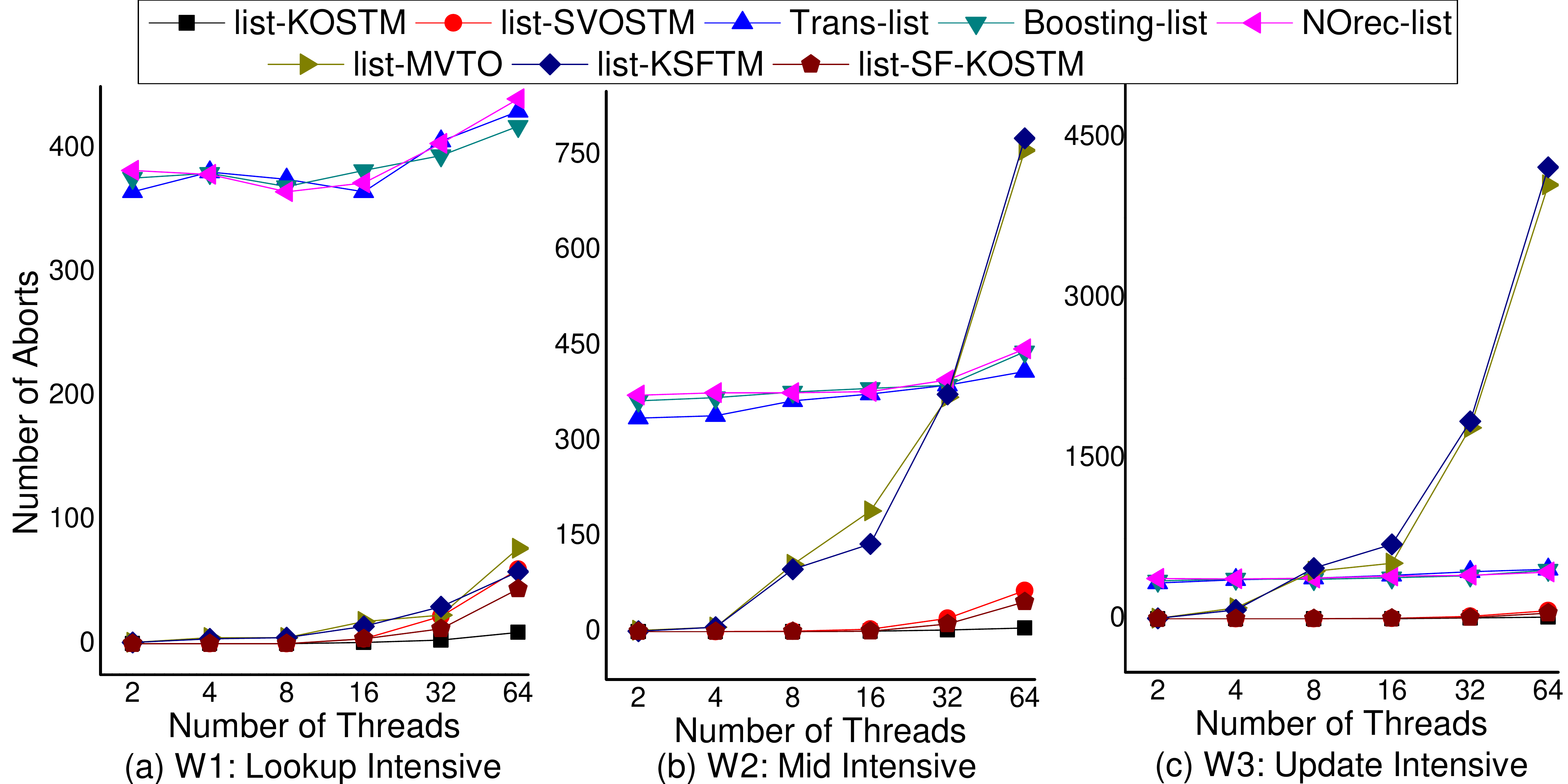}
	\caption{Abort Count of SF-KOSTM and State-of-the-art STMs on list}
	\label{fig:listabortcomp}
\end{figure}

\noindent
\textbf{Garbage Collection:} To delete the unwanted versions, we use garbage collection mechanism in SF-MVOSTM and proposed SF-MVOSTM-GC. The results show that SF-MVOSTM-GC performs better than SF-MVOSTM. Garbage collection method deletes the unwanted version corresponding to the key. In garbage collection we use a \textit{livelist}, this livelist contains all the transaction that are live, which means every transaction on the start of its first incarnation logs its time in livelist and when commit/abort remove its entry from livelist in sorted order of transactions on the basis of $wts$. Garbage collection is achieved by deleting the version which is not latest, whose timestamp is smaller than the $wts$ of smallest live transaction. \figref{mc} represents the \textbf{\emph{Memory Consumption}} by SF-MVOSTM-GC and SF-KOSTM algorithms for \emph{high contention (or HC)} and \emph{low contention (or LC)} environment on workload $W_3$ for the linked-list data structure. Here, each algorithm creates a version corresponding to the key after successful \emph{STM\_tryC()}. We calculate the memory consumption based on the \emph{Version Count (or VC)}. If version is created then VC is incremented by 1 and after garbage collection, VC is decremented by 1. \figref{mc}. (a). demonstrates that memory consumption is kept on increasing in SF-MVOSTM-GC but memory consumption by SF-KOSTM is constant because of maintaining finite versions corresponding to the less number of keys (high contention). \figref{mc}. (b). shows for \emph{low contention} environment where memory consumption are keeps on increasing in SF-MVOSTM-GC as well as SF-KOSTM. But once limit of $K$-version reach corresponding to all the keys in SF-KOSTM, memory consumption will be stable. Similar observation can be found for other workloads $W_1, W_2$ and other data structure hash table as well.

\begin{figure}[H]
	\centering
	\includegraphics[width=10cm, height=5cm]{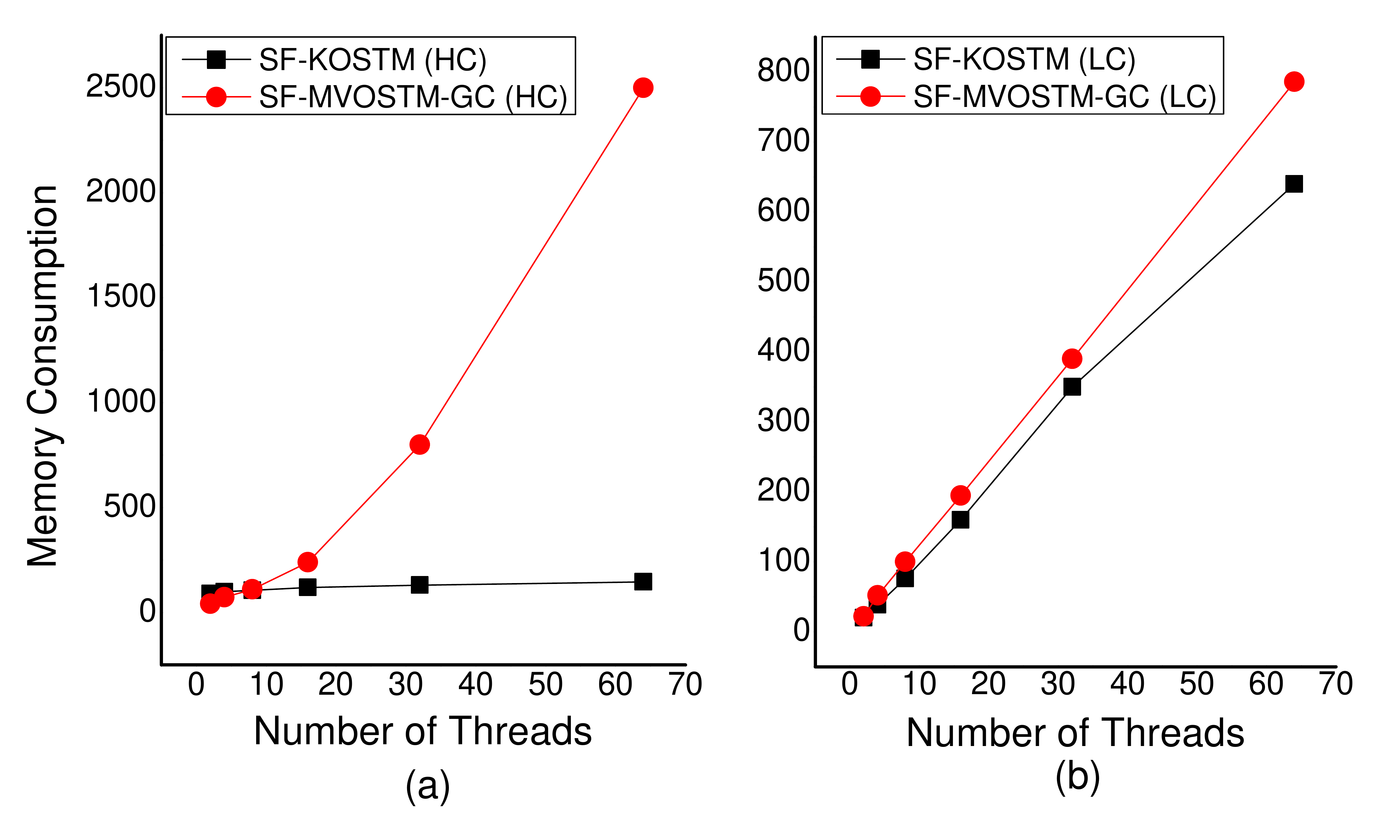}
	\vspace{-.4cm}	\caption{Comparison of Memory Consumption in SF-KOSTM and SF-MVOSTM-GC}
	\label{fig:mc}
\end{figure}

\begin{table}[H]
	\small
	\begin{tabular}{ | m{4.5cm} | m{2.5cm}| m{2.5cm} | m{2.5cm} | } 
		\hline
		Algorithm & W1 & W2 & W3 \\ 
		\hline
		HT-SF-SVOSTM & 1.49 & 1.13 & 1.09 \\
		\hline
		HT-SF-MVOSTM & 1.22 & 1.08 & 1.04 \\
		\hline
		HT-SF-MVOSTM-GC	& 1.13 & 1.03 & 1.02 \\
		\hline
		HT-SVOSTM & 3.3 & 0.77 & 1.57 \\
		\hline
		ESTM & 2.92 & 1.4 & 3.03 \\
		\hline
		RWSTM & 3.13 & 2.65 & 13.36  \\
		\hline
		HT-MVTO & 2.37 & 4.74 & 49.39  \\
		\hline
		HT-KOSTM & 0.91 & 0.7 & 0.8 \\
		\hline
		
		\hline
		
	\end{tabular}
	\captionsetup{justification=centering}
	\caption{Speedup by HT-SF-KOSTM  }
	\label{tab:htspeed}
\end{table}

\begin{table}[H]
	\small
	\begin{tabular}{ | m{4.5cm} | m{2.5cm}| m{2.5cm} | m{2.5cm} | } 
		\hline
		Algorithm & W1 & W2 & W3 \\ 
		\hline
		list-SF-SVOSTM & 1.29 & 1.26 & 1.22 \\
		\hline
		list-SF-MVOSTM & 1.25 & 1.29 & 1.12 \\
		\hline
		list-SF-MVOSTM-GC	& 1.14 & 1.5 & 1.13 \\
		\hline
		list-SVOSTM & 2.4 & 1.5 & 1.4 \\
		\hline
		Trans-list & 24.15 & 19.06 & 23.26 \\
		\hline
		Boosting-list & 22.43 & 22.52 & 27.2  \\
		\hline
		NOrec-list & 26.12 & 27.33 & 31.05  \\
		\hline
		list-MVTO & 10.8 & 23.1 & 19.57  \\
		\hline
		list-KSFTM & 5.7 & 18.4 &  74.20 \\
		\hline
		list-KOSTM	 & 0.96 & 0.98 & 0.8 \\
		\hline
		
		\hline
		
	\end{tabular}
	\captionsetup{justification=centering}
	\caption{Speedup by list-SF-KOSTM  }
	\label{tab:listspeed}
\end{table}
\vspace{-.5cm}
\section{Conclusion}
\label{sec:conc}
\vspace{-0.3cm}
\noindent
We proposed a novel \emph{Starvation-Free K-Version Object-based STM (SF-KOSTM)} which ensure the \emph{\stfdm} while maintaining the latest $K$-versions corresponding to each key and satisfies the correctness criteria as \emph{local-opacity}. The value of $K$ can vary from 1 to $\infty$. When $K$ is equal to 1 then SF-KOSTM boils down to \emph{Single-Version Starvation-Free OSTM (SF-SVOSTM)}. When $K$ is $\infty$ then SF-KOSTM algorithm maintains unbounded versions corresponding to each key known as \emph{Multi-Version Starvation-Free OSTM (SF-MVOSTM)}. To delete the unused version from the version list,  SF-MVOSTM calls a separate Garbage Collection (GC) method and proposed SF-MVOSTM-GC. SF-KOSTM provides greater concurrency and higher throughput using higher-level methods. We implemented all the proposed algorithms for \emph{hash table} and \emph{linked-list} data structure but it is generic for other data structures as well. 
Results of SF-KOSTM shows significant performance gain over state-of-the-art STMs. 


	\bibliographystyle{splncs}
		\vspace{-.4cm}
	\bibliography{citations}
	\vspace{-1cm}
\clearpage
\clearpage
\appendix
\section*{Appendix}
\label{apn:appendix}

\noindent The appendix section is organized as follows:

\begin{table}[H]
	\caption{Table of Contents}
	\begin{tabular}{ | m{.9cm} | m{5.8cm}| m{1.8cm} | m{1.3cm} | m{1.6cm} | } 
		\hline
		S. No.& Section Name  & Section No. & Page No. & Page Count \\ 
		\hline
		1. & Remaining Experimental Evaluation & \apnref{ap-result} & 16 & 8 \\
		\hline
		2. & Working of SF-SVOSTM Algorithm and Importance of Timestamp ranges & \apnref{ap-rcode} & 26 & 10\\
		\hline
		
	\end{tabular}
	\captionsetup{justification=centering}
	
	\label{tbl:appen}
\end{table}

\section{Remaining Experimental Evaluation}
\label{apn:ap-result}
This section explains the additional results and analysis that boost the potential of our work experimentally. 

\begin{figure}[H]
	\includegraphics[width=12cm, height=5cm]{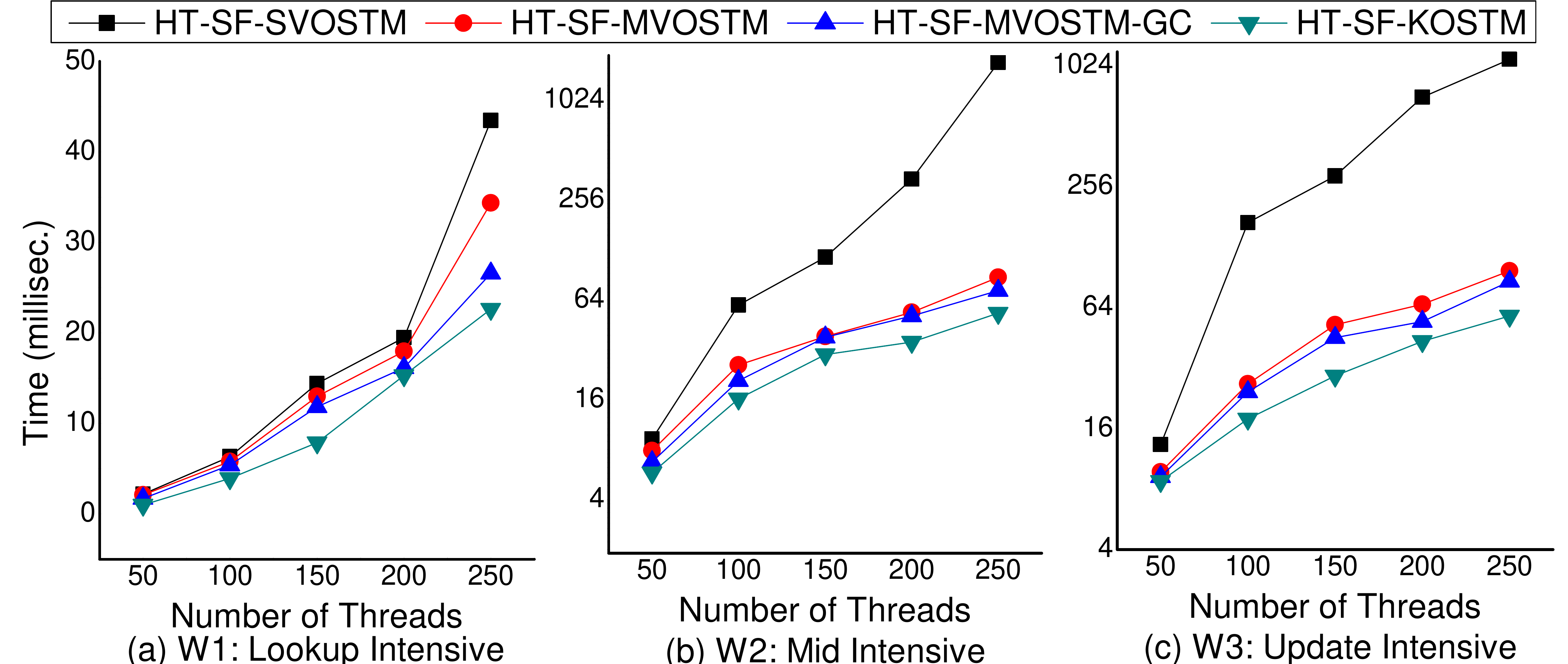}
	\caption{ Performance analysis among SF-SVOSTM and variants of SF-KOSTM on hash table}
	\label{fig:htus}
\end{figure}

\noindent
We started our experiments with hash table data structure of bucket size 5 as shown in \figref{htus}. First, we compared \emph{max-time} taken by all our proposed HT-SF-SVOSTM and variations of HT-SF-KOSTM (HT-SF-MVOSTM, HT-SF-MVOSTM-GC, and HT-SF-KOSTM) while varying the number of threads from 50 to 250. We did these experiments on \emph{high contention environment} which includes only 30 keys, number of threads ranging from 50 to 250, each thread spawns upon a transaction, where each transaction performs 10 operations (insert, delete and lookup) depending upon the workload chosen. We analyze through \figref{htus} that the finite version HT-SF-KOSTM performs best among all the proposed algorithms on all the three types of workloads (W1, W2, and W3 defined in \secref{exp}) with value of $K$ and C as 5 and 0.1 respectively. Here, $K$ is the number of versions in the version list and C is the variable used to derive the $wts$. Similarly, we consider another data structure as linked-list on \emph{high contention environment} and compared \emph{max-time} taken by all our proposed list-SF-SVOSTM, list-SF-MVOSTM, list-SF-MVOSTM-GC, and list-SF-KOSTM algorithms. \figref{listus} represents that list-SF-KOSTM performs best for all the workloads (W1, W2, and W3) with value of $K$ and C as 5 and 0.1 respectively. 

\begin{figure}[H]
	\includegraphics[width=12cm, height=5cm]{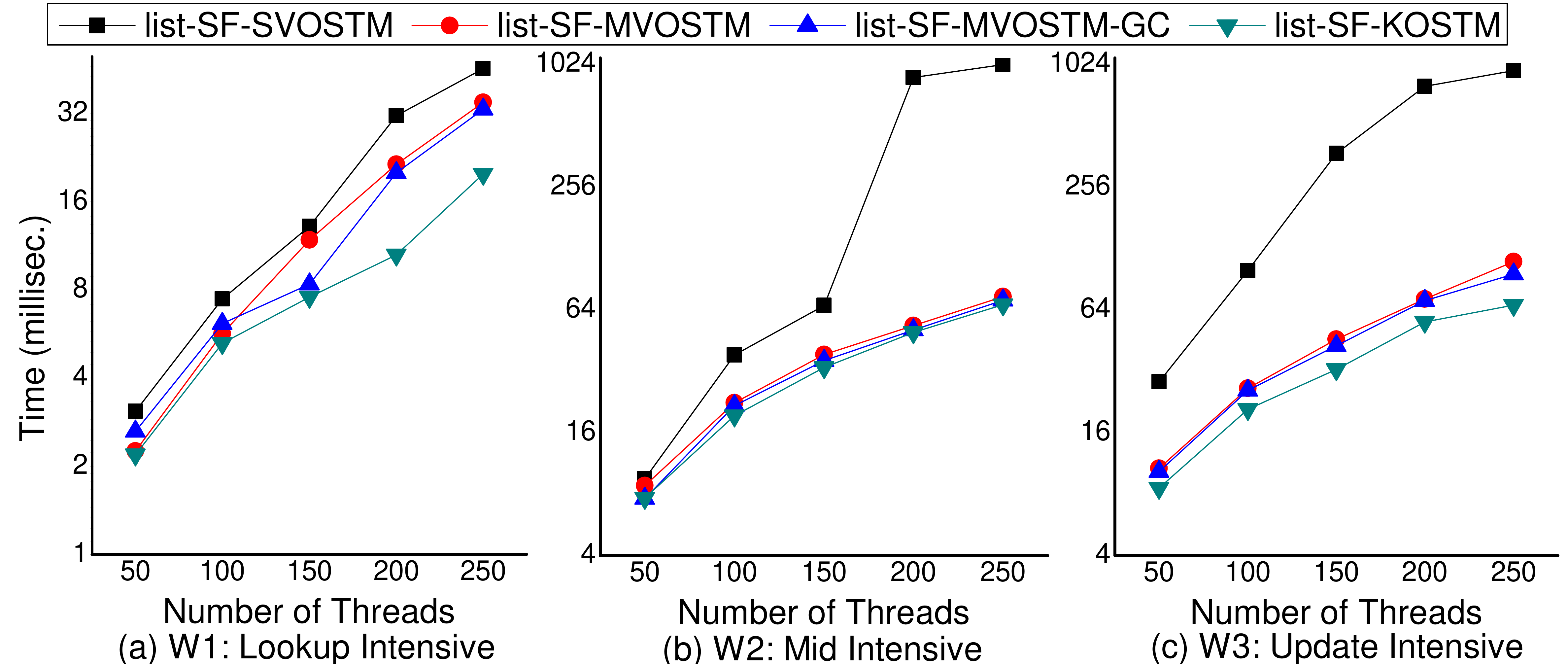}
	\caption{ Performance analysis among SF-SVOSTM and variants of SF-KOSTM on list}
	\label{fig:listus}
\end{figure}
\subsection{Pseudo code of Counter Application}
\label{apn:countercode}
To analyze the absolute benefit of starvation-freedom, we use a \emph{Counter Application} which provides us the flexibility to create a high contention environment where the probability of transactions undergoing starvation on an average is very high. In this subsection we describe the detailed functionality of \emph{Counter Application} though pseudo code as follows:    
\label{apn:conters}
\begin{algorithm}  
	\caption{\emph{main()}: The main function invoked by \emph{Counter Application}.} \label{algo:main} 
	\begin{algorithmic}[1]
	\makeatletter\setcounter{ALG@line}{143}\makeatother		
	\State /*Each thread $th_i$ log \emph{abort counts, average time taken by each transaction to commit and worst case time} (maximum time to commit the transaction) in $abortCount_{th_i}$, $timeTaken_{th_i}$ and $worstTime_{th_i}$ respectively;*/
		\ForAll{(numOfThreads)} /*Multiple threads call the helper function*/
		\State \emph{helperFun()};
		\EndFor 
		\ForAll{(numOfThreads)} 
		\State /*Join all the threads*/ 
		\EndFor 
		\ForAll{(numOfThreads)}
		\If{($maxWorstTime$ $<$ $worstTime_{th_i}$)}
		\State /*Calculate the \emph{Maximum Worst Case Time}*/
		\State $maxWorstTime$ = $worstTime_{th_i}$; 
		\EndIf
		\State /*Calculate the \emph{Total Abort Count}*/
		\State $totalAbortCount$ += $abortCount_{th_i}$;
		\State /*Calculate the \emph{Average Time Taken}*/
		\State $AvgTimeTaken$ /= $TimeTaken_{th_i}$;
		\EndFor
	\end{algorithmic}
\end{algorithm}

\vspace{1mm}
\begin{algorithm} [H] 
	\caption{\emph{helperFun()}:Multiple threads invoke this function.} \label{algo:testFunc} 
	\begin{algorithmic}[1]
		\makeatletter\setcounter{ALG@line}{160}\makeatother		
	    \State Initialize the Transaction Count $txCount_i$ of $T_i$ as 0;
		\While{(\emph{numOfTransac\texttt{}tions})} /*Execute until number of  transactions are non zero*/
		    \State $startTime_{th_i}$ = timeRequest(); /*get the start time of thread $th_i$*/
		    \State\Comment{Execute the transactions $T_i$ by invoking \emph{testSTM} functions;}
		    \State $abortCount_{th_i}$ = $testSTM_i()$;
		    \State Increment the $txCount_i$ of $T_i$ by one.
		    \State $endTime_{th_i}$ = timeRequest(); /*get the end time of thread $th_i$*/
		    \State /*Calculate the \emph{Total Time Taken} by each thread $th_i$*/
		    \State $timeTaken_{th_i}$ += ($endTime_{th_i}$ - $startTime_{th_i}$);
		    \State /*Calculate the \emph{Worst Case Time} taken by each thread $th_i$*/
		    \If{($worstTime_{th_i}$ $<$ ($endTime_{th_i}$ - $startTime_{th_i}$))}
		    \State $worstTime_{th_i}$ = ($endTime_{th_i}$ - $startTime_{th_i}$); 
		    \EndIf
		    \State Atomically, decrement the \emph{numOfTransactions};
\cmnt{
		    			\algstore{myalg}
		\end{algorithmic}
	\end{algorithm}
	\begin{algorithm}
		\begin{algorithmic}
			\algrestore{myalg}
}			
		\EndWhile
		\State /*Calculate the \emph{Average Time} taken by each thread $th_i$*/
		\State $TimeTaken_{th_i}$ /= $txCount_i$;
	\end{algorithmic}
\end{algorithm}

\vspace{1mm}
\begin{algorithm} [H] 
	\caption{$testSTM_i()$: Main function which executes the methods of the transaction $T_i$ (or $i$) by thread $th_i$.} \label{algo:testFunx} 
	\begin{algorithmic}[1]
	\makeatletter\setcounter{ALG@line}{177}\makeatother			
    \While{(\emph{true})}
    \If{(\emph{i.its} != \emph{nil})}\label{lin:begin}
    \State \emph{STM\_begin(i.its)};  /*If $T_i$ is an incarnation*/
	\Else
	\State \emph{STM\_begin(nil)}; /*If $T_i$ is first invocation*/
    \EndIf
    \ForAll{(\emph{numOfMethods})}
    \State $k_i$ = rand()\%totalKeys;/*Select the key randomly*/
    \State $m_i$ = rand()\%100;/*Select the method randomly*/
    		\Switch{($m_i$)}\label{lin:cminer8}
    \Case{($m_i$ $\leq$ \emph{STM\_lookup}()):}\label{lin:cminer9}
    \State $v$ $\gets$ \emph{STM\_lookup{($k_i$)}}; /*Lookup key $k$ from a shared memory*/\label{lin:cminer10}
    \If{($v$ == $abort$)}\label{lin:cminer11}
    \State $txAbortCount_i++$; /*Increment the transaction abort count*/
    \State goto \Lineref{begin};\label{lin:cminer12}
    \EndIf\label{lin:cminer13}
    \EndCase
    \Case{(\emph{STM\_lookup()} $<$ $m_i$ $\leq$ \emph{STM\_insert}()):} \label{lin:cminer14}
    \State /*Insert key $k_i$ into $T_i$ local memory with value $v$*/\label{lin:cminer15}
    \State \emph{STM\_insert($k_i, v$)}; \label{lin:cminer16}
    \EndCase
    \Case{(\emph{STM\_insert()} $<$ $m_i$ $\leq$ \emph{STM\_delete}()}):
    \State /*Actual deletion happens after successful \emph{STM\_tryC()}*/
    \State \emph{STM\_delete($k_i$)};
    \EndCase
    \Case{default:}\label{lin:cminer17}
    \State /*Neither lookup nor insert/delete on shared memory*/\label{lin:cminer18}
    \EndCase
    \State $v$ = \emph{STM\_tryC()}; /*Validate all the methods of $T_i$ in tryC*/
    \If{($v$ == \emph{abort})}
     \State $txAbortCount_i++$;
        \State goto \Lineref{begin};
    \EndIf
        \EndSwitch
    \EndFor
    \State return $\langle txAbortCount_i \rangle$;
    \EndWhile
	\end{algorithmic}
\end{algorithm}
\section{Working of SF-SVOSTM and Remaining Pseudo code of SF-KOSTM Algorithms}
\label{apn:ap-rcode}
In this section, we cover the working of \emph{Starvation-Free Single-Version OSTM (SF-SVOSTM)} algorithm and remaining pseudo code of \emph{Starvation-Free Multi-Version OSTM (SF-MVOSTM)}.
\subsection{Working of SF-SVOSTM Algorithm and Importance of Timestamp ranges}
\label{apn:ap-sfostm}

This subsection describes the working of \emph{Starvation-Free Single-Version OSTM (SF-SVOSTM)} algorithm which includes the detail description of SF-SVOSTM methods and challenges to make it starvation-free.  

SF-SVOSTM invokes \tbeg{()}, \tlook{()}, \tdel{()}, \tins{()}, and \tryc{()} methods. \tlook{()} and \tdel{()} works as \rvmt{()} which lookup the value of key $k$ from shared memory and returns it. Whereas \tins{()} and \tdel{()} work as \upmt{()} that modifies the value of $k$ in shared memory. We propose optimistic SF-SVOSTM, so, \upmt{()} first update the value of $k$ in its local log $txLog$ and the actual effect of \upmt{()} will be visible after successful \tryc{()}. This subsection explains the functionality of each method as follows:

\noindent
\textbf{\tbeg{():}} 
We show the high-level view of \tbeg{()} in \algoref{obegin1}. When a thread $Th_i$ invokes transaction $T_i$ for the first time (or first incarnation) then \tbeg{()} assigns a unique timestamp known as \emph{current timestamp} $(cts)$ as shown in \Lineref{obeg6}. It is incremented atomically with the help of atomic global counter ($gcounter$). If $T_i$ gets aborted then thread $Th_i$ executes it again with new incarnation of $T_i$, say $T_j$ with the new $cts$ until $T_i$ commits but retains its initial $cts$ as \emph{initial timestamp} $(its)$ at \Lineref{obeg8}. $Th_i$ uses $its$ to inform the STM system that whether $T_i$ is a new invocation or an incarnation. 
If $T_i$ is the first incarnation then $its$ and $cts$ are same as $cts_i$ so, $Th_i$ maintains $\langle \emph{$its_i$, $cts_i$}\rangle$. If $T_i$ gets aborted and retries with $T_j$ then $Th_i$ maintains $\langle \emph{$its_i$, $cts_j$}\rangle$. We use timestamp (\emph{ts}) $i$ of $T_i$ as $cts_i$, i.e., $\langle cts_i$ = $i\rangle$ for SF-SVOSTM. 



By assigning priority to the lowest $its$ (i.e. transaction have been in the system for longer time) \emph{Starvation-Freedom} can achieve in \emph{Single-Version OSTM}. 

\begin{algorithm}[H] 
	\caption{\emph{STM\_begin($its$)}: This method is invoke by a thread $Th_i$ to start a new transaction $T_i$. It pass a parameter $its$ which is the initial timestamp of the first incarnation of $T_i$. If $T_i$ is the first incarnation then $its$ is $nil$.}
	\label{algo:obegin1} 
		\begin{algorithmic}[1]
			\makeatletter\setcounter{ALG@line}{209}\makeatother
			\Procedure{\emph{STM\_begin($its$)}}{}  \label{lin:obeg1}
			\State Create a local log $txLog_i$ for each transaction $T_i$. \label{lin:obeg2}
			\If {($its == nil$)} \label{lin:obeg3}
			\State /* Atomically get the value from the global counter and set it to its, and, \blank{1cm} cts.*/ \label{lin:obeg4}
			\State $its_i$ = $cts_i$ = \emph{gcounter.get\&Inc()}; \label{lin:obeg5}
			\Else  \label{lin:obeg6}
			\State /*Set the $its_i$ to first incarnation of $T_i$ $its$*/ \label{lin:obeg7}
			\State $its_i$ = $its$; \label{lin:obeg8}
			\State /*Atomically get the value from the global counter for $cts_i$*/ \label{lin:obeg9}
			\State $cts_i$ = \emph{gcounter.get\&Inc()}. \label{lin:obeg10}
			\EndIf \label{lin:obeg11}
			\State /*Initially, set the $status_i$ of $T_i$ as $live$*/ \label{lin:obeg12}
			\State $status_i$ = $live$;  \label{lin:obeg13}
			\State return $\langle cts_i, its_i\rangle$ \label{lin:obeg14}
			
			\EndProcedure
		\end{algorithmic}
\end{algorithm}

 \tbeg{()} initializes the \emph{transaction local log $(txLog_i)$} for each transaction $T_i$ to store the information in it. Whenever a transaction starts it atomically sets its \emph{status} to be \emph{live} as a global variable at \Lineref{obeg13}. Transaction \emph{status} can be $\langle live, commit, false\rangle$. After successful execution of \tryc{()}, $T_i$ sets its \emph{status} to be \emph{commit}. If the \emph{status} of the transaction $T_i$ is \emph{false} then it returns \emph{abort} which says that $T_i$ is not having the lowest $its$ among other concurrent conflicting transactions. So to propose starvation-free SVOSTM other conflicting transactions sets its \emph{status} field as \emph{false} and force transaction $T_i$ to \emph{abort}.



\noindent
\textbf{\tlook{()} and \tdel{()} as \rvmt{s()}:} \emph{\rvmt{(ht, k, val)}} returns the value (\emph{val}) corresponding to the key \emph{k} from the shared memory as hash table (\emph{ht}). We show the high-level overview of the \rvmt{()} in \algoref{orvmt}. First, it identifies the key $k$ in the transaction local log as \emph{$txLog_i$} for transaction $T_i$. If $k$ exists then it updates the $txLog_i$ and returns the $val$ at \Lineref{orvm2}.

If key $k$ does not exist in the $txLog_i$ then before identifying the location in share memory \rvmt{()} checks the \emph{status} of $T_i$ at \Lineref{orvm5}. If the \emph{status} of $T_i$ (or $i$) is \emph{false} then $T_i$ has to \emph{abort} which says that $T_i$ is not having the lowest $its$ among other concurrent conflicting transactions. So to propose starvation-free SVOSTM other conflicting transactions sets its \emph{status} field as \emph{false} and force transaction $T_i$ to \emph{abort}.

\begin{algorithm}
	\caption{\emph{\rvmt{(ht, k, val)}:} It can either be $\tdel_i(ht, k, val)$ or $\tlook_i(ht, k, val)$ on key $k$ of transaction $T_i$.} \label{algo:orvmt} 	
	\setlength{\multicolsep}{0pt}
	\begin{algorithmic}[1]
		\makeatletter\setcounter{ALG@line}{224}\makeatother
		\Procedure{$\rvmt_i$}{$ht, k, val$}	\label{lin:orvm1}	
		\If{($k \in \llog_i$)}
		\State Update the local log and return $val$.  \label{lin:orvm2}
		\Else \label{lin:orvm3}
		\State /*Atomically check the \emph{status} of its own transaction $T_i$ (or $i$)*/ \label{lin:orvm4} 
		\If{(\emph{i.status == false})} return $\langle abort_i \rangle$. \label{lin:orvm5}
		\EndIf\label{lin:orvm6}
		\State Identify the \emph{preds[]} and \emph{currs[]} for $k$ in bucket $M_k$ of \lsl using \bn \blank{1cm} and \rn . \label{lin:orvm7}
		\State Acquire locks on \emph{preds[]} \& \emph{currs[]} in increasing order. \label{lin:orvm8}
		\If{(\emph{!rv\_Validation(preds[], currs[])})} \label{lin:orvm9}
		\State Release the locks and goto \linref{orvm7}. \label{lin:orvm10}
		\EndIf \label{lin:orvm11}
		\If{($k  ~ \notin ~ M_k.\lsl$)} \label{lin:orvm12} 
		\State Create a new node $n$ with key $k$ as:  $\langle$\emph{key=k, lock=false, mark=true, \blank{1.5cm} rvl=i, nNext=}$\phi\rangle$./*$n$ is marked*/  \label{lin:orvm13}
		\State Insert $n$ into $M_k.$\emph{\lsl} s.t. it is accessible only via \rn{s}. /*\emph{lock} sets \blank{1.6cm}\emph{true}*/  \label{lin:orvm15}  	
		\State Release locks; update the $\llog_i$ with $k$.\label{lin:orvm16}
		\State return $\langle$\emph{val}$\rangle$. /*\emph{val} as $null$*/\label{lin:orvm17}
		\Else
		\State Add $i$ into the $rvl$ of \emph{currs[]}. \label{lin:orvm29} 
		\State Release the locks; update the $\llog_i$ with $k$ and value. \label{lin:orvm30}
		\State return $\langle$$val$$\rangle$. \label{lin:orvm32}
		\EndIf \label{lin:orvm18}
		\EndIf \label{lin:orvm31}
		\EndProcedure \label{lin:orvm33}
	\end{algorithmic}
\end{algorithm}
\vspace{.2mm}

If the \emph{status} of $T_i$ is not \emph{false} and $k$ is not exist in the $txLog_i$ then it identifies the location optimistically (without acquiring the locks similar to the \emph{lazy-list}\cite{Heller+:LazyList:PPL:2007}) in the shared memory at \Lineref{orvm7}. SF-SVOSTM maintains the shared memory in the form of a hash table with $M$ buckets as shown in \subsecref{dsde}, where each bucket stores the keys in the form of \emph{\lsl}. Each node contains two pointer $\langle \rn, \bn \rangle$. So, it identifies the two \emph{predecessors (pred)} and two \emph{current (curr)} with respect to each node. First, it identifies the pred and curr for key $k$ in \bn{} as $\langle \bp, \bc \rangle$. After that it identifies the pred and curr for key $k$ in \rn{} as $\langle \rp, \rc \rangle$. If $\langle \rp, \rc \rangle$ are not marked then $\langle \bp=\rp, \bc=\rc\rangle$. SF-SVOSTM maintains the keys are in increasing order. So the order among the nodes are $\langle \bp.key \leq \rp.key < k \leq \rc.key \leq \bc.key\rangle$.

\rvmt{()} acquires the lock in predefined order on all the identified preds and currs for key $k$ to avoid the deadlock at \Lineref{orvm8} and do the \emph{rv\_Validation()} at \Lineref{orvm9}. If $\langle \bp \vee \bc\rangle$ is marked or preds are not pointing to identified currs as $\langle (\bp.\bn \neq \bc) \vee (\rp.\rn \neq \rc)\rangle$ shown in \algoref{llsearch} then it releases the locks on all the preds and currs and identify the new preds and currs for key $k$ in shared memory.

If key $k$ does not exist in the \emph{\lsl{}} of the corresponding bucket $M_k$ at \Lineref{orvm12} then it creates a new node $n$ with key $k$ as $\langle \emph{key=k, lock=false, mark=true, rvl=i, nNext=}\phi \rangle$ at \Lineref{orvm13}. $T_i$ adds its \emph{$cts_i$} in the \emph{rvl}. Finally, it inserts the node $n$ into $M_k.\lsl$ such that it is accessible via \rn{} only at \Lineref{orvm15}. \rvmt{()} releases the locks and update the $txLog_i$ with key $k$ and value as $null$ (\Lineref{orvm16}). Eventually, it returns the \emph{val} as $null$ at \Lineref{orvm17}.

\begin{algorithm}[H]
	\caption{\emph{rv\_Validation(preds[], currs[])}: It is mainly used for \emph{rv\_method()} validation.}
	\label{algo:llsearch} 
	\begin{algorithmic}[1]
		\makeatletter\setcounter{ALG@line}{248}\makeatother
		\Procedure{$rv\_Validation{(\emph{preds[], currs[]})}$}{} \label{lin:orvv1}
		\If{$((\bp.mark) || (\bc.mark) || ((\bp.\bn)\neq \blank{1.8cm} \bc) || ((\rp.\rn) \neq {\rc})$)}\label{lin:orvv2}
		return $\langle false \rangle$. 
		\Else{} return $\langle true \rangle$. \label{lin:orvv3}
		\EndIf \label{lin:orvv4}
		\EndProcedure \label{lin:orvv5}
	\end{algorithmic}
\end{algorithm}


\vspace{-.5mm}

If key $k$ exists in the $M_k.\lsl$ then 
it adds the $cts_i$ of $T_i$ as $i$ in the $rvl$ of \emph{currs[]} at \Lineref{orvm29}. Finally, it releases the lock and updates the $txLog_i$ with key $k$ and value as val at \Lineref{orvm30}. Eventually, it returns the \emph{val} at \Lineref{orvm32}.

\noindent
\textbf{\tins{()} and \tdel{()} as \upmt{s()}:} Actual effect of \tins{()} and \tdel{()} comes after successful \tryc{()}. 
We shows the high level view of \tryc{()} in \algoref{otryc}. First, \tryc{()} checks the \emph{status} of the transaction $T_i$ at \Lineref{otc3}. If \emph{status} of $T_i$ is \emph{false} then $T_i$ has to \emph{abort} same as explained above in \rvmt{()}.

If the \emph{status} is not false then \tryc{()} sort the keys (exist in $txLog_i$ of $T_i$) of \upmt{s()} in increasing order. 
It takes one by one method ($m_{ij}$) from the $txLog_i$ and identifies the location of the key $k$ in \emph{$M_k$.\lsl} as explained above in \rvmt{()}. After identifying the preds and currs for $k$ it acquires the locks in predefined order to avoid the deadlock at \Lineref{otc10} and calls \emph{tryC\_Validation()} to validate the methods of $T_i$.

\emph{tryC\_Validation()} identifies whether the methods of invoking transaction $T_i$ are inserting/updating a node corresponding to the keys while ensuring the \emph{starvation-freedom}. 
First, it do the \emph{rv\_Validation()} at \Lineref{otcv2} as explained in \rvmt{()}. If \emph{rv\_Validation()} is successful and key $k$ exists in the $M_k.\lsl$ then 
it maintains the All Return Value List (allRVL) from \emph{currs[]} of key $k$ at \Lineref{otcv6}. 
Acquire the locks on \emph{status} of all the transactions present in allRVL list including $T_i$ it self in predefined order to avoid the deadlock at \Lineref{otcv8}. First, it checks the \emph{status} of its own transaction $T_i$ at \Lineref{otcv10}. If the \emph{status} of $T_i$ is \emph{false} then $T_i$ has to \emph{abort} the same reason as explained in \rvmt{()}.

If \emph{status} of $T_i$ is not \emph{false} then it compares the $its_i$ of its own transaction $T_i$ with the $its_p$ of other transactions $T_p$ ($p$) present in the allRVL at \Lineref{otcv13}. Along with this it checks the \emph{status} of $p$. If above conditions $\langle (its_i < its_p) \&\& (p.status==live))\rangle$ succeed then it includes $T_p$ in the Abort Return Value List (abortRVL) at \Lineref{otcv14} to abort $T_p$ later otherwise abort $T_i$ itself at \Lineref{otcv15}.  

\begin{algorithm}
	
	\caption{\tryc{($T_i$)}: Validate the \upmt{s()} of the transaction $T_i$ and returns \emph{commit}.}
	\label{algo:otryc}
	\begin{algorithmic}[1]
		\makeatletter\setcounter{ALG@line}{253}\makeatother
		\Procedure{$\tryc{(T_i)}$}{} \label{lin:otc1}
		\State /*Atomically check the \emph{status} of its own transaction $T_i$ (or $i$)*/ \label{lin:otc2}
		\If{(\emph{i.status == false})} return $\langle abort_i \rangle$. \label{lin:otc3}
		\EndIf \label{lin:otc4}
		\State /*Sort the $keys$ of $\llog_i$ in increasing order.*/ \label{lin:otc5}
		\State /*Method ($m$) will be either \tins or \emph{STM\_delete}*/\label{lin:otc6}
		\ForAll{($m_{ij}$ $\in$ $\llog_i$)} \label{lin:otc7}
		\If{($m_{ij}$==\tins$||$$m_{ij}$==\tdel)}\label{lin:otc8}
		\State Identify the \emph{preds[]} \& \emph{currs[]} for \emph{k} in bucket $M_k$ of \emph{\lsl} using \blank{1.5cm} \bn \& \rn. \label{lin:otc9}
		\State Acquire the locks on \emph{preds[]} \& \emph{currs[]} in increasing order. \label{lin:otc10}
		\If{($!tryC\_Validation()$)}
		\State return $\langle abort_i \rangle$.\label{lin:otc11}
		\EndIf\label{lin:otc12}
		\EndIf\label{lin:otc13}
		\EndFor\label{lin:otc14}
		\ForAll{($m_{ij}$ $\in$ $\llog_i$)}\label{lin:otc15}
		\State \emph{poValidation()}  modifies the \emph{preds[]} \& \emph{currs[]} of current method which \blank{1cm} would have been updated by previous method of the same transaction.\label{lin:otc16}
		\If{(($m_{ij}$==\tins)\&\&(k$\notin$$M_k$.\lsl))}\label{lin:otc17}
		\State Create new node $n$ with $k$ as: $\langle$\emph{key=k, lock=false, mark=false, rvl=$\phi$, \blank{1.5cm} nNext=$\phi$}$\rangle$. \label{lin:otc18}
		\State Insert node $n$ into $M_k$.\emph{\lsl} such that it is accessible via \rn{} as \blank{1.5cm} well as \bn{}. /*\emph{lock} sets \emph{true}*/ \label{lin:otc20}
		\ElsIf{($m_{ij}$ == \tins{})}\label{lin:otc21}
		\State /*Sets \emph{rvl} as $\phi$ and update the value*/.
		\State Node (\emph{currs[]}) is accessible via \rn and \bn. /*\emph{mark} sets \emph{false}*/
		\EndIf\label{lin:otc23}
		\If{($m_{ij}$ == \tdel{})}\label{lin:otc24}
		\State /*Sets \emph{rvl} as $\phi$ and 			\emph{mark} as \emph{true}*/.
		\State Node (\emph{currs[]}) is accessible via \rn only.
		\label{lin:otc25}
		\EndIf\label{lin:otc26}
		\State Update the \emph{preds[]} \& \emph{currs[]} of $m_{ij}$ in $\llog_i$.\label{lin:otc27}
		
		\EndFor \label{lin:otc28}
		\State Release the locks.
		\State return $\langle commit_i \rangle$.\label{lin:otc29}
		\EndProcedure \label{lin:otc30}
	\end{algorithmic}
\end{algorithm}


At \Lineref{otcv39}, \tryc{()} aborts all other conflicting transactions which are present in the abortRVL while modifying the \emph{status} field to be \emph{false} to achieve \emph{starvation-freedom}.

All the steps of the \emph{tryC\_Validation()} are successful than the actual effect of the \tins{()} and \tdel{()} will be visible to the shared memory. At \Lineref{otc16}, \tryc{()} checks for \emph{poValidation()}. When two subsequent methods $\langle m_{ij}, m_{ik}\rangle$ of the same transaction $T_i$ identify the overlapping location of preds and currs in \emph{\lsl}. Then \emph{poValidation()} updates the current method $m_{ik}$ preds and currs with the help of previous method $m_{ij}$ preds and currs.

If $m_{ij}$ is \tins{()} and key $k$ does not exist in the $M_k.\lsl{}$ then it creates the new node $n$ with key $k$ as $\langle \emph{key=k, lock=false, mark=false, rvl=$\phi$, nNext=}\phi \rangle$ at \Lineref{otc18}. 
Finally, it inserts the node $n$ into $M_k.\lsl$ such that it is accessible via \rn{} as well as \bn{} at \Lineref{otc20}. If $m_{ij}$ is \tins{()} and key $k$ exists in the $M_k.\lsl{}$ then it updates the value and \emph{rvl} to $\phi$ for node corresponding to the key $k$. 

\vspace{.2mm}
\begin{algorithm}
	\caption{\emph{tryC\_Validation():} It is only use for \tryc{()} validation.}
		\begin{algorithmic}[1]
			\makeatletter\setcounter{ALG@line}{286}\makeatother
			\Procedure{\emph{tryC\_Validation{()}}}{} \label{lin:otcv1}
			\If{(\emph{!rv\_Validation()})}
			Release the locks and \emph{retry}.\label{lin:otcv2}
			\EndIf\label{lin:otcv3}
			\If{(k $\in$ $M_k.\lsl$)}\label{lin:otcv4}
			\State Maintain the list of \emph{currs[].rvl} as allRVL for all key \emph{k} of $T_i$.\label{lin:otcv6}
			\State /*p is the \emph{timestamp} of transaction $T_p$*/
			\If{($p$ $\in$ allRVL)} /*Includes $i$ as well in allRVL*/\label{lin:otcv7}
			\State Lock \emph{status} of each $p$ in pre-defined order. \label{lin:otcv8}
			\EndIf\label{lin:otcv9}
			\If{(\emph{i.status == false})} return $\langle false \rangle$. \label{lin:otcv10}
			\EndIf\label{lin:otcv11}
			\ForAll{($p$ $\in$ allRVL)}\label{lin:otcv12}
			\If{(($its_i$$<$$its_p$)$\&\&$(\emph{p.status==live}))} \label{lin:otcv13}
			\State Maintain \emph{abort list} as abortRVL \& includes \emph{p} in it.\label{lin:otcv14}
			\Else{} return $\langle false \rangle$. /*abort $i$ itself*/\label{lin:otcv15}
			\EndIf\label{lin:otcv16}
			
			\EndFor\label{lin:otcv17}
			\cmnt{			\ForAll{($ver$ $\in$ currVL)}\label{lin:otcv18}
				\State Calculate $tltl_i$ = min($tltl_i$, $ver.vrt+1$).\label{lin:otcv19}
				\EndFor\label{lin:otcv20}
				\ForAll{($ver$ $\in$ nextVL)}\label{lin:otcv21}
				\State Calculate $tutl_i$ = min($tutl_i$, $ver.vNext$ \blank{1.1cm}$.vrt-1$).\label{lin:otcv22}
				\EndFor\label{lin:otcv23}
				\If{($tltl_i$ $>$ $tutl_i$)} /*abort $i$ itself*/
				\State return $\langle false \rangle$. \label{lin:otcv24}
				\EndIf\label{lin:otcv25}
				\ForAll{($p$ $\in$ smallRVL)}\label{lin:otcv26}
				\If{($tltl_p$ $>$ $tutl_i$)}\label{lin:otcv27}
				\If{(($its_i$$<$$its_p$)$\&\&$(\emph{p.status==live}))} \label{lin:otcv28}
				\State Includes $p$ in abortRVL list.\label{lin:otcv29}
				\Else{} return $\langle false \rangle$. /*abort $i$ itself*/\label{lin:otcv30}
				\EndIf\label{lin:otcv31}
				\EndIf	\label{lin:otcv32}	
				\EndFor \label{lin:otcv33}
				\State $tltl_i$ = $tutl_i$. /*After this point $i$ can't abort*/\label{lin:otcv34}
				\ForAll{($p$ $\in$ smallRVL)}
				\State /*Only for \emph{live} transactions*/ \label{lin:otcv35}
				\State Calculate the $tutl_p$ = min($tutl_p$, $tltl_i-1$).\label{lin:otcv36}
				\EndFor \label{lin:otcv37}
			}			
			\ForAll{($p$ $\in$ abortRVL)} \label{lin:otcv38}
			\State Set the \emph{status} of $p$ to be $false$. \label{lin:otcv39}
			\EndFor \label{lin:otcv40}
			\EndIf\label{lin:otcv41}
			\State return $\langle true \rangle$.\label{lin:otcv42}
			\EndProcedure \label{lin:otcv43}
		\end{algorithmic}
\end{algorithm}
\vspace{.2mm}

If $m_{ij}$ is \tdel{()} and key $k$ exists in the $M_k.\lsl{}$ then it sets the \emph{rvl} as $\phi$ and \emph{mark} field as $true$ for node corresponding to the key $k$ at \Lineref{otc25}. At last, it updates the preds and currs of each $m_{ij}$ into its $txLog_i$ to help the upcoming methods of the same transactions in \emph{poValidation()} at \Lineref{otc27}. Finally, it releases the locks on all the keys in a predefined order and returns \emph{commit} at \Lineref{otc29}.
\subsection{Importance of Timestamp Ranges in SF-KOSTM}
\label{apn:wtsdis}

\textbf{Violation of Real-Time Order by wts:} As described in \subsecref{working}, $cts$ respects the real-time order among the transactions but SF-KOSTM uses $wts$ which may not respect real-time order. Sometimes, the value of $wts$ is significantly larger than $cts$ which leads to violate the real-time order among the transactions. \figref{sfmv-correct} illustrates it with history $H$: $l_1(ht,k_1,v_0) l_2(ht,k_2,v_0) i_1(ht,k_1,v_{10}) C_1 i_2(ht,k_1, v_{20}) C_2 l_3(ht,k_1,v_{10}) i_3(ht,k_3,v_{25})\\ C_3$ consists of three transactions $T_1, T_2, T_3$ with $cts$ as 100, 110, 130 and $wts$ as 100, 150, 130 respectively. $T_1$ and $T_2$ has been committed before the beginning of $T_3$, so $T_1$ and $T_2$ are in real-time order with $T_3$. Formally, $T_1 \prec_{H}^{RT} T_3$ and $T_2 \prec_{H}^{RT} T_3$. But, $T_2$ has higher $wts$ than $T_3$. Now, $T_3$ lookups key $k_1$ from $T_1$ and returns the value as $v_{10}$  because $T_1$ is the available largest $wts$ (100) smaller than $T_3$ $wts$ (130). The only possible equivalent serial order $S$ to history $H$ is $T_1 T_3 T_2$ which is \legal as well. But $S$ violates real-time order because $T_3$ is serialized before $T_2$ in $S$ but $T_2$ has been committed before the beginning of $T_3$ in $H$. It can  easily be seen that, such history $H$ can be accepted by the algorithm when it uses only $wts$ instead of $cts$. But this should not happen because its violating the real-time order which says it does not satisfy the correctness criteria as \emph{local opacity}. 

A simple solution to this issue is by delaying the committing transaction say $T_i$ with $wts_i$ until the real-time catches up to the $wts_i$. Delaying such $T_i$ will ensure the correctness criteria as \emph{local opacity} while making the $wts$ of the transaction same as real-time. But, this is highly unacceptable. It seems like transaction $T_i$ acquires the locks on all the keys it wants to update and wait. It will show the adverse effect and reduces the performance of SF-KOSTM system.

\noindent
\textbf{Regaining the Real-Time Order using Timestamp Ranges along with wts:} We require that all the transactions of history $H$ generated by SF-KOSTM are serialized based on their $wts$ while respecting the real-time order among them. Another efficient solution is to allow the transaction $T_i$ with $wts_i$ to catch up with the actual time if $T_i$ does not violates the real-time order. So, to respect the real-time order among the transactions SF-KOSTM uses the time constraints. SF-KOSTM uses the idea of timestamp ranges \cite{Riegel+:LSA:DISC:2006,Guer+:disc1:2008, Crain+:RI_VWC:ICA3PP:2011} along with $\langle its_i, cts_i, wts_i \rangle$ for transaction $T_i$ in \tbeg{()}. It maintains the \emph{transaction lower timestamp limit ($tltl_i$)} and \emph{transaction upper timestamp limit ($tutl_i$)} for $T_i$. Initially, $\langle its_i, cts_i, wts_i, tltl_i \rangle$ are the same for $T_i$. $tutl_i$ would be set as a largest possible value denoted as $+\infty$ for $T_i$. After successful execution of \emph{\rvmt{s()}} or \tryc{()} of $T_i$, $tltl_i$ gets incremented and $tutl_i$ gets decremented to respect the real-time order among the transactions as explained in \subsecref{working}. 

\begin{figure}
	\centerline{
		\scalebox{0.48}{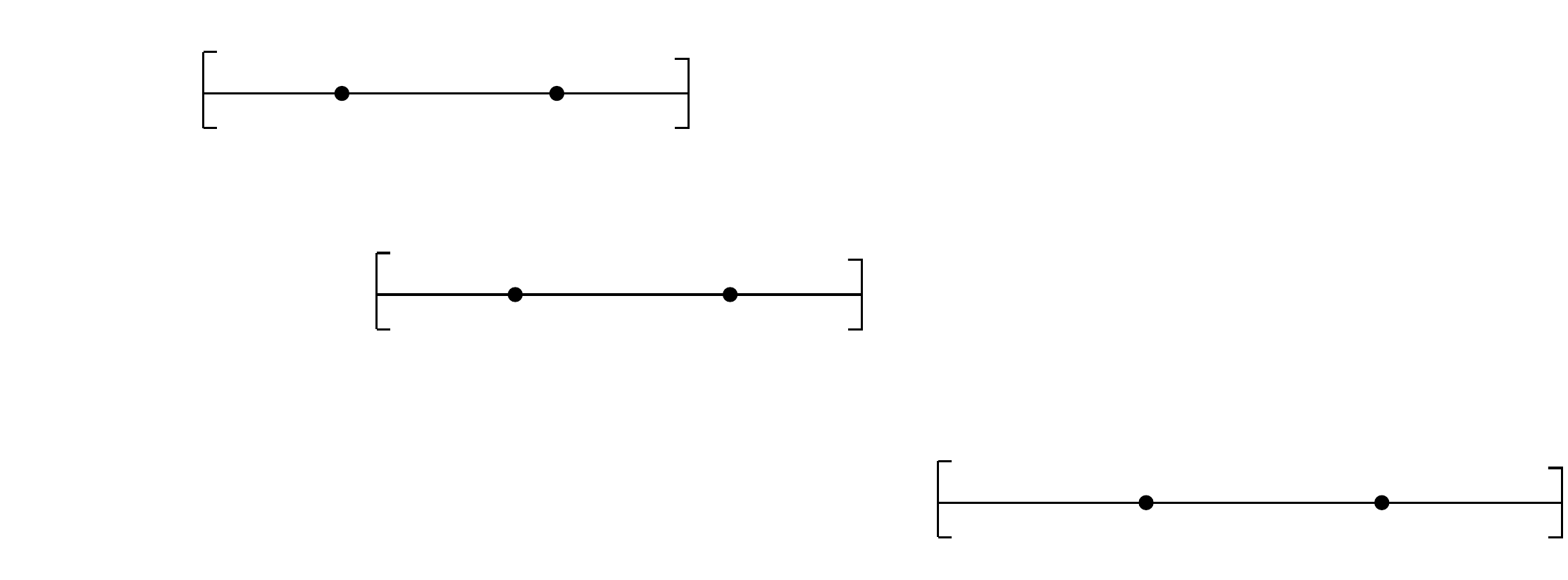}}
	\captionsetup{justification=centering}
	\caption{Violating the \emph{real-time order} by \emph{wts}}
	\label{fig:sfmv-correct}
\end{figure}

For better understanding consider \figref{tltltutl}, which shows the regaining the \emph{real-time order} using timestamp ranges ($tltl$ and $tutl$) along with \emph{wts} on history $H$. Initially, $T_1$ begins with $cts_1$ = $wts_1$ = $tltl_1$ = 100, $tutl_1$ = $\infty$ and $T_1$ returns commit. We assume at the time of commit of $T_1$, $gcounter$ is 120. So, $tutl_1$ reduces to 120. After that $T_2$ commits and with suppose $tutl_2$ reduces to 121 (so, the current value $gcounter$ is 121). $T_1$ and $T_2$ both access the key $k_1$ and $T_2$ is updating $k_1$. So, $T_1$ and $T_2$ are conflicting. Hence, $tltl_2$ is incremented to a value greater than $tutl_1$, say 121. Now, when $T_3$ begins , it assigns 
$cts_3$ = $wts_3$ = $tltl_3$ = 130, and $tutl_3$ = $\infty$. At the time of $l_3(ht,k_1,v_{10})$, as $T_3$ lookups the version of $k_1$ from $T_1$, so, $T_3$ reduces its $tutl_3$ less than $tltl_2$ (currently, $tltl_2$ is 121). Hence, $tutl_3$ becomes say 120. But, $tltl_3$ is already 130. So, $tltl_3$ has crossed the limit of $tutl_3$ which is causing $T_3$ to abort. Intuitively, this implies that $wts_3$ and real-time order are out of synchrony and can not be reconciled. Hence, by using the timestamp ranges $H$ executes correctly by SF-KOSTM algorithm with equivalent serial schedule $T_1 T_2$.


\cmnt{
\begin{figure}
	\centerline{
		\scalebox{0.5}{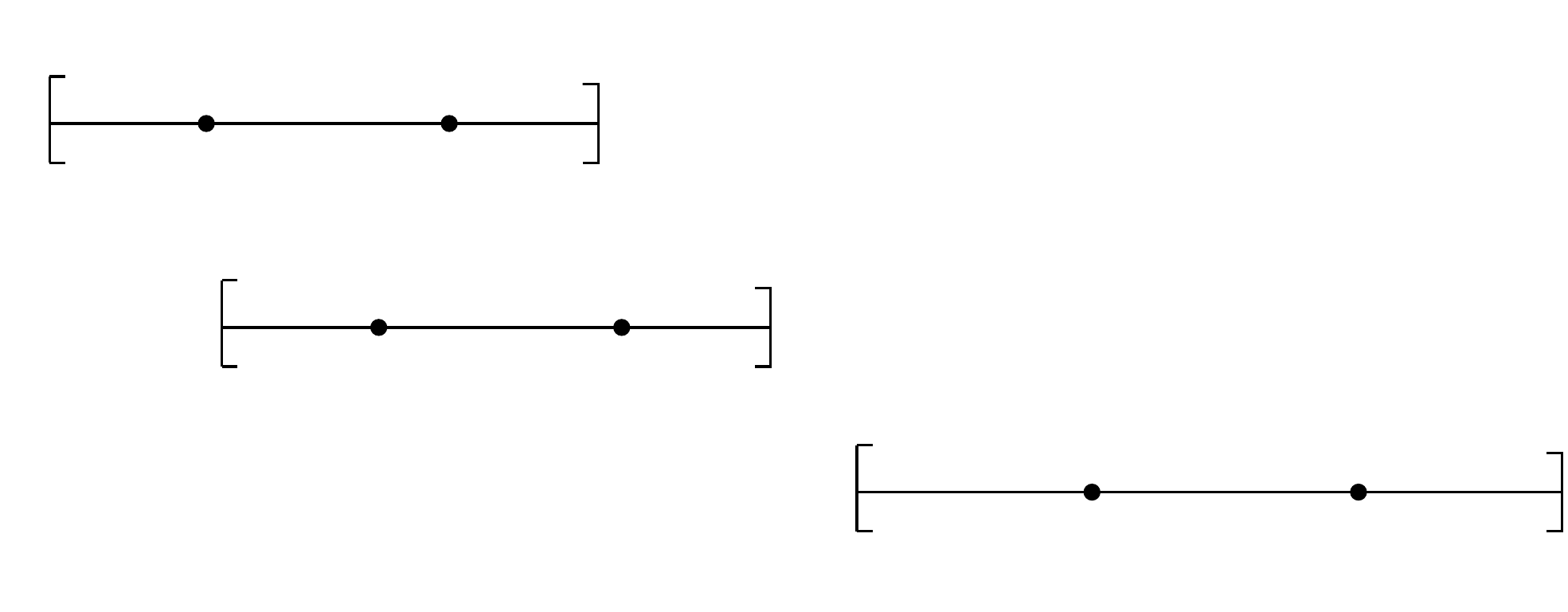}}
	\captionsetup{justification=centering}
	\caption{\emph{wts} violating \emph{real-time order}}
	\label{fig:sfmv-correct}
\end{figure}
}

\begin{figure}
	\centerline{
		\scalebox{0.43}{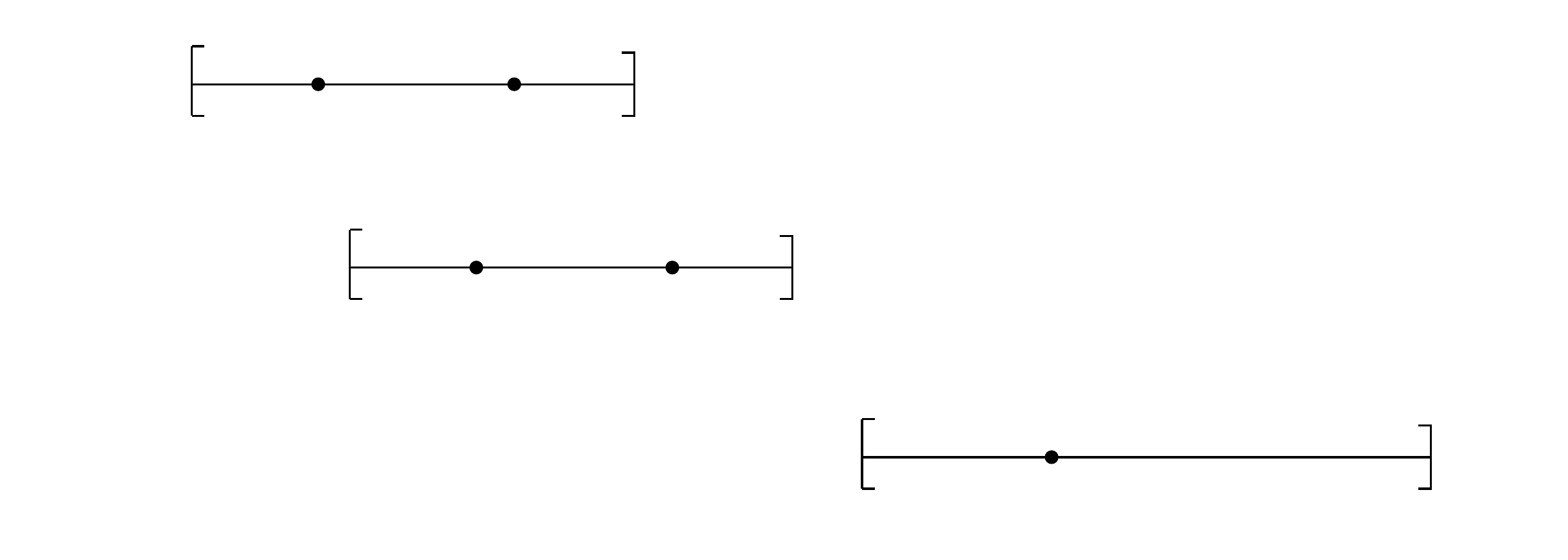}}
	\captionsetup{justification=centering}
	\caption{Regaining the \emph{real-time order} using Timestamp Ranges along with \emph{wts}}
	\label{fig:tltltutl}
\end{figure}

\end{document}